\documentclass[12pt]{article}
\usepackage{amssymb}
\usepackage{amsmath}
\usepackage{graphicx,color}
\usepackage{amsfonts}
\usepackage[ignoreall]{geometry}
\usepackage{natbib}
\usepackage{enumitem}
\usepackage{setspace}
\usepackage{xr}
\usepackage{relsize}
\usepackage{comment}
\usepackage{subcaption}
\usepackage{chngcntr}
\usepackage{apptools}
\usepackage{url}
\usepackage{longtable}
\AtAppendix{\counterwithin{lemma}{section}}
\AtAppendix{\counterwithin{theorem}{section}}
\AtAppendix{\counterwithin{algorithm}{section}}
\AtAppendix{\counterwithin{remark}{section}}
\AtAppendix{\counterwithin{definition}{section}}
\AtAppendix{\counterwithin{figure}{section}}
\AtAppendix{\counterwithin{equation}{section}}

\newcommand{\0}{\mathbf 0}

\newcommand{\AAA}{{\mathcal A}}
\newcommand{\aaaa}{\mbox{$\mathbf a$}}

\newcommand{\ttt}{\boldsymbol \theta}

\newcommand{\BBB}{\mbox{$\mathcal B$}}
\newcommand{\bb}{\mathbf b}

\newcommand{\CCC}{\mathcal C}
\newcommand{\cc}{\mathbf c}

\newcommand{\HH}{\mathcal H}

\newcommand{\OO}{\mathcal O}

\newcommand{\PP}{\mathcal P}

\newcommand{\R}{\mathbb R}

\newcommand{\x}{\mathbf x}
\newcommand{\y}{\mathbf y}
\newcommand{\z}{\mathbf z}

\newcommand{\argmin}{\operatornamewithlimits{arg\,min}}

\newtheorem{theorem}{Theorem}

\newtheorem{algorithm}{Algorithm}

\newtheorem{definition}{Definition}

\newtheorem{lemma}{Lemma}

\newtheorem{proposition}{Proposition}
\newtheorem{remark}{Remark}

\title{Unfolding-Model-Based Visualization:  Theory, Method and Applications}
\author{Yunxiao Chen\footnote{Department of Statistics, London School of Economics and Political Science. Address: Columbia House, Room 5.16, Houghton Street, London, WC2A 2AE. Email: y.chen186@lse.ac.uk}, Zhiliang Ying\footnote{Department of Statistics, Columbia University}, Haoran Zhang\footnote{Shanghai Center for Mathematical Sciences, Fudan University}\\
 }
\date{}

\newenvironment{proof}[1][Proof]{\noindent\textbf{#1.} }{\ \rule{0.5em}{0.5em}}
\begin{document}
\maketitle
\doublespacing

\begin{abstract}
Multidimensional unfolding methods are widely used for visualizing item response data. Such methods project respondents and items simultaneously onto a low-dimensional Euclidian space, in which respondents and items are represented by ideal points, with person-person, item-item, and person-item similarities being captured by the Euclidian distances between the points. In this paper, we study the visualization of multidimensional unfolding from a statistical perspective. We cast multidimensional unfolding into an estimation problem, where the
respondent and item ideal points are treated as parameters to be estimated.
An estimator is then proposed for the simultaneous estimation of these parameters.
Asymptotic theory is provided for the recovery of the ideal points, shedding lights on the validity of model-based visualization.
An alternating projected gradient descent algorithm is proposed for the parameter estimation. We provide two
illustrative examples, one on users' movie rating and the other on senate roll call voting.
\end{abstract}


\section{Introduction}

Multidimensional unfolding (MDU) methods are widely used as an important data visualization tool in social and behavioral sciences such as psychology \citep[][]{van2007multidimensional,papesh2010multidimensional}, political science \citep[][]{poole2000nonparametric,poole2005spatial,clinton2004statistical,bakker2013bayesian}, and marketing \citep[][]{desarbo1987constructing,desarbo1997parametric,ho2010unfolding}. It is regarded as the dominant method in the scaling of both preferential choice and attitude \citep{de2005multidimensional}.
The basic idea of MDU is to place both respondents and items in a joint Euclidean space based on data, with the understanding that respondents tend to prefer items that are close to them in the space. This joint visualization may lead to better understanding and interpretations of both the respondents and the items, as compared with separately visualizing the respondents and the items by themselves.
MDU has its origin in psychology  \citep{bennett1956determination,bennett1960multidimensional,hays1961multidimensional,coombs1964theory}.
It is closely related to multidimensional scaling (MDS) methods \citep{kruskal1964multidimensional,kruskal1978multidimensional,borg2005modern} and several other recent approaches to nonlinear dimension reduction and manifold learning \citep{tenenbaum2000global, lu2005framework, chen2009local, zhang2016distance}.

MDU methods can be categorized into two types,  algorithm-based and model-based. Algorithm-based methods \citep[e.g.,][]{takane1977nonmetric,greenacre1986efficient,smacof} estimate the ideal points by minimizing a certain objective function, also known as the stress function in the literature of MDU.
The classical algorithm-based methods have been implemented in the R package \emph{smacof} \citep{smacof} that is widely used for MDU and MDS analysis.
Model-based methods \citep[e.g.,][]{desarbo1987constructing,hinich2005new,bakker2013bayesian}, however, infer the locations of the ideal points by making use of a probabilistic model. Such a model typically assumes that, up to some measurement error, the similarity between a person and an item is a decreasing function of some defined distance between the corresponding ideal points. The specification of MDU models is closely related to
item response theory models in psychometrics \citep[see e.g.,][]{embretson2013item,rabe2004generalized,bartholomew2011latent}.


The MDU problem is closely related to MDS. The key difference is that data for the former do not contain direct measurement of within-set (i.e., person-person and item-item) similarities, while data for MDS typically have such information.
Largely due to the missing information contained in the within-set similarities, the MDU problem tends to be more challenging.
As a result, degenerate solutions are often encountered in the applications of MDU methods, in which case the visualization and the corresponding interpretations convey no information \citep[e.g.,][]{busing2005avoiding,borg2005modern}, while  MDS results tend to be more stable.
These empirical observations suggest that it is of importance to study the validity of MDU solutions, which motivates the research in this paper.

This paper studies the visualization of MDU from the statistical perspective. First, for binary choice data, we formulate the MDU problem into a parameter estimation problem under a general family of probabilistic MDU models, where the respondent and item ideal points are treated as parameters to be estimated.
Second, an estimator is proposed for the ideal points and an asymptotic theory is provided for this estimator,
shedding lights on the validity of model-based visualization.  Finally, an efficient alternating projected gradient algorithm is proposed for the computation which is scalable to large-scale problems.

We illustrate the proposed method through two applications, one on movie rating and the other on senate roll call voting. The movie dataset is a subset from the famous MovieLens dataset \citep{harper2016movielens}. We unfold the 943 users and 338  movies in the dataset. Specifically, we study the users' movie watching decisions. Based on the ideal points of movies in a two-dimensional space, it is found that one dimension of the space corresponds to the popularity of the movies and the other dimension corresponds to the release date of the movies. Good understanding of the user ideal points is further obtained based on their distances to the movie ideal points. The senate voting dataset is based on the senate roll call voting records from the 108th congress in 2003-2004. Based on the unfolding of the senators and roll calls, it is found that most of the ideal points lie around a one-dimensional line, with the two extremes of the line representing the most liberal and the most conservative political standings.

The rest of the paper is organized as follows. In Section~\ref{sec:model}, we introduce a family of MDU models and formulate the problem of joint configuration recovery into an estimation problem.
In Section~\ref{sec:theory}, we propose an estimator, for which
statistical theory is established that guarantees
the consistency of configuration recovery under reasonable conditions.
Simulation studies and real data examples are presented in Sections~\ref{sec:simu} and \ref{sec:real}, respectively. We end with discussions on future directions in Section~\ref{sec:disc}. An application to cluster analysis, proofs of the theoretical results, and  numerical comparison with classical MDU methods are provided as supplementary materials.

\section{Distance-based MDU} \label{sec:model}

\subsection{Distance-based Unfolding Model for Binary Data}

Consider $N$ respondents making choice on $J$ binary items (e.g., ``agree/disagree"). Let $Y_{ij}$ be a random variable, denoting the response from respondent $i$ to item $j$, taking value 0 or 1, and let $y_{ij}$ be its realization. For example, such data can come from senate roll call voting, where the respondents are senators and the items correspond to roll calls. Response $Y_{ij} = 1$ means that senator $i$ supports roll call $j$ and $Y_{ij}=0$ otherwise.

We provide a simulated example in Figure~\ref{fig:illustrative} to illustrate MDU analysis. Panel (a) shows the heat map of an observed response matrix which consists of 20 respondents and 10 items,
where 0 and 1 responses are represented by red and green colors, respectively.  Given choice data in panel (a), an
MDU method aims at representing respondents and items by ideal points in the same low-dimensional Euclidian space  $\mathbb R^K$ as in panel (b) of Figure~\ref{fig:illustrative}  that can be easily visualized, where the respondent-respondent, respondent-item, and item-item relationships are captured by the between-points distance. The dimension $K$ of the Euclidian space  is often set to be 2 or 3 for the purpose of visualization.

\begin{figure}
  \centering
  \begin{subfigure}[b]{0.4\textwidth}
        \centering
        \includegraphics[scale = 0.6]{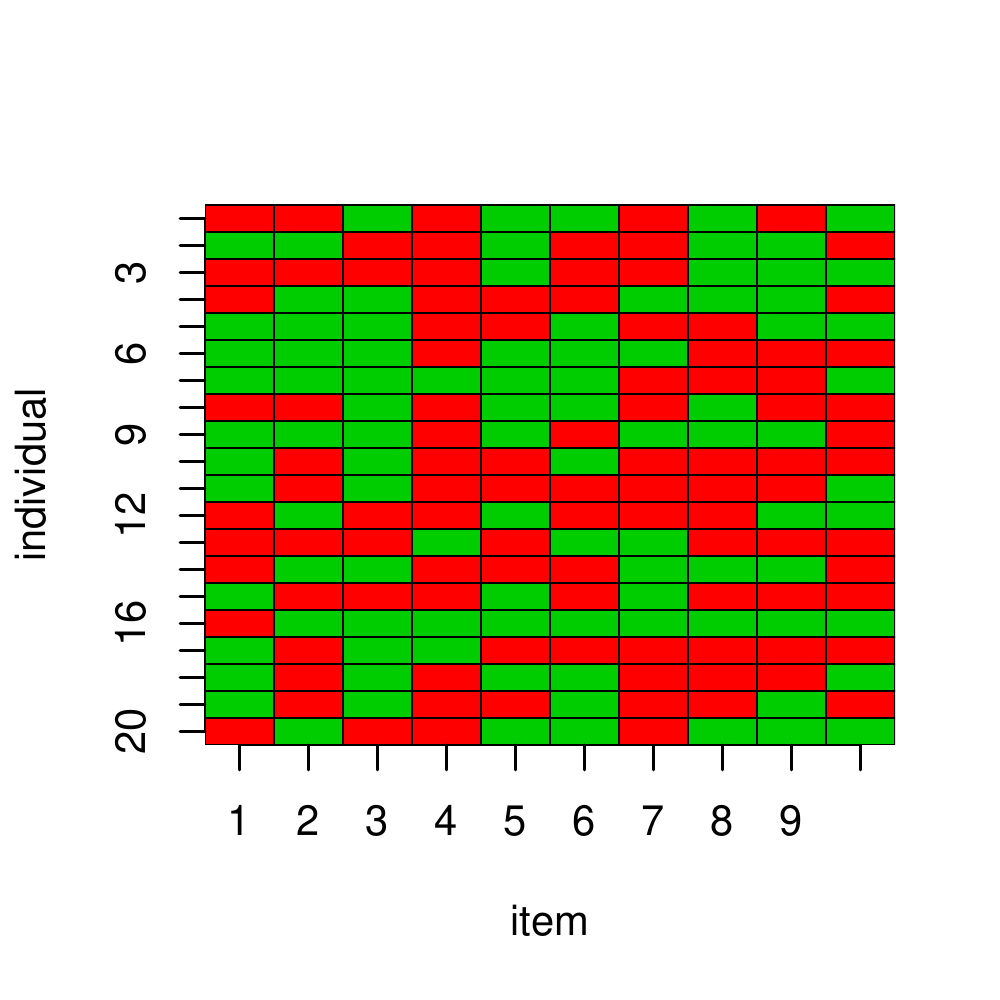}
        \caption{}
    \end{subfigure}
    \begin{subfigure}[b]{0.4\textwidth}
        \centering
        \includegraphics[scale = 0.6]{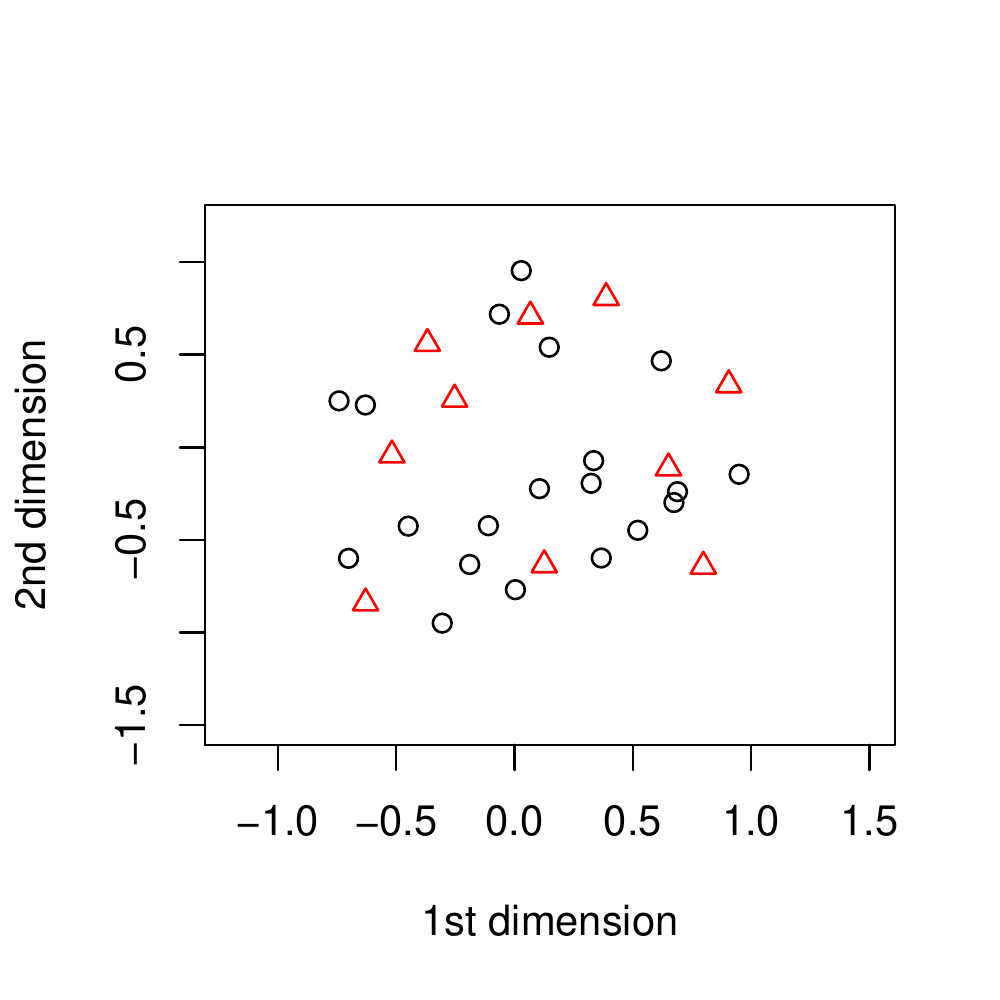}
        \caption{}
    \end{subfigure}
\caption{An illustrative example. Panel (a): The heatmap of a response matrix, where 0 and 1 responses are represented by red and green colors, respectively.
Panel (b): The respondent and item ideal points where black circles represent respondents and red triangles represent items. }
\label{fig:illustrative}
\end{figure}

One way to conduct MDU is via a statistical model.
An MDU model typically assumes that each respondent/item is associated with a true ideal point in $\mathbb R^K$ that is represented by a $K$-dimensional parameter vector. Let $\ttt_i = (\theta_{i1}, ..., \theta_{iK})^\top$ and $\aaaa_j = (a_{j1}, ..., a_{jK})^\top$ denote the parameter vectors of respondent $i$ and item $j$, respectively.
It is assumed that  response $Y_{ij}$ is determined by the Euclidian distance between $\ttt_i$ and $\aaaa_j$ in $\mathbb R^K$. Finally, we use $\Theta^N = (\theta_{ik})_{N\times K}$ and $A^J = (a_{jk})_{J\times K}$ to denote the matrices containing all the person and the item ideal points, respectively.
Under such a statistical model, the goal of MDU becomes to estimate the person and item parameters based on data.


In this paper, we focus on MDU models taking the form
\begin{equation}\label{eq:model}
P(Y_{ij} = 1\mid \ttt_i, \aaaa_j) = f(\Vert\ttt_i - \aaaa_j\Vert^2),
\end{equation}
where $\Vert\cdot\Vert$ denotes the standard $L_2$ norm and $f: [0, \infty) \rightarrow [0, 1]$ is a pre-specified link function. It is assumed that the responses $Y_{ij}$ are conditionally independent, given the ideal points $\ttt_i$ and $\aaaa_j$, $i = 1, ..., N, j =  1, ..., J$. This model falls under the general framework of the MDU threshold model for binary choice data \citep[see][]{desarbo1987constructing}. According to the form of \eqref{eq:model}, the distribution of data only depends on the squared distance between every pair of person and item ideal points, $d_{ij} = \Vert\ttt_i - \aaaa_j\Vert^2, i = 1, ..., N, j = 1, ..., J$. The matrix $D_{N,J} = (d_{ij})_{N\times J}$ is known as the corresponding \textit{partial distance matrix}, where the subscripts of $D_{N, J}$ emphasize the dependence of this matrix on the numbers of respondents and items. 

In addition, the link function $f$ is often assumed to be a monotone decreasing function, so that a larger distance implies a lower probability of $Y_{ij} = 1$. An example of such a link function is $f(x) = {2}/{(1+\exp(x))}$. When $f(x)$ takes this form, $P(Y_{ij} = 1\mid \ttt_i, \aaaa_j) = 1$ when the distance between $\ttt_i$ and $\aaaa_j$ is 0, i.e., the two points are identical, and the probability $P(Y_{ij} = 1\mid \ttt_i, \aaaa_j)$ decays towards 0 when the distance increases.

In what follows, we provide two remarks on this modeling framework.
\begin{remark}
  We remark on the link function $f$ which plays a similar role as the dissimilarity transformation function in the classical MDS and MDU methods \citep[e.g., Chapter 9,][]{borg2005modern}. Assuming a pre-specified $f$ is similar to assuming an identity transformation in classical MDU.

  In classical MDS and MDU, the dissimilarity transformation function can be unknown and estimated from data parametrically or non-parametrically. Similar treatment can be applied to the link function $f$. For example, one may assume
  $$f(\Vert\ttt_i - \aaaa_j\Vert^2) = g(\beta_0 + \beta_1 \Vert\ttt_i - \aaaa_j\Vert^2),$$
  where $g:  \mathbb R \rightarrow [0, 1]$ is a given monotone decreasing function and $\beta_0$ and $\beta_1$ are additional parameters to be estimated from data together with the person- and item-specific parameters. This form is similar in spirit to the interval transformation in classical MDU.
  When no constraint is imposed on the scales of $\ttt_i$s and $\aaaa_j$s,
  $\beta_1$ needs to be fixed to be a constant (e.g., $\beta_1 = 1$) for model identifiability.
  One may also estimate $f$ non-parametrically, for example, by using monotone splines.

  Under suitable regularity conditions, our theoretical development in Section~\ref{sec:theory} can be extended to the case when $f$ also needs to be estimated from data.
\end{remark}






\begin{remark}\label{rmk:other}

Although we focus on binary data, the introduced modeling framework can be easily extended to other types of preference data, such as rating and ranking data. For example, consider rating data
$Y_{ij} \in \{0, ..., T\}$, where 0, 1, ..., $T$ are $T+1$ ordered response categories. A higher category implies a higher level of agreement between the respondent and the item.
Then one can assume the following unfolding model
\begin{equation}
P\left(Y_{ij} \geq t \mid \ttt_i,\aaaa_j\right) = g\left(d_t + \| \ttt_i - \aaaa_j \|^2\right),
\end{equation}
for $t\in \{1,...,T\},$ where  $g:  \mathbb R \rightarrow [0, 1]$ is a given monotone decreasing function and
$d_1$, ..., $d_T$ are additional model parameters.
It implies that the larger the distance, the smaller the probability for $Y_{ij}$ to take a large value.
This model is closely related to the graded response model \citep{samejima1997graded} in item response theory.
For another example, consider ranking data consisting of pair-wise comparisons, where each response is a comparison between two items $j$ and $j'$. Following the same idea as above, one may model the probability that item $j$
is preferred over $j'$ to take the form $g(\Vert \ttt_i - \aaaa_j \Vert^2 - \Vert \ttt_i - \aaaa_{j'} \Vert^2)$. That is, the probability decreases with the difference of their squared distances to person $i$.
Our theoretical results and computational algorithm given below can be adapted to these situations.


\end{remark}

%
%
%
%

\subsection{Recovery of Configuration}\label{subsec:config}

Our main goal is the simultaneous recovery of the ideal points $\ttt_i$ and $\aaaa_j$, based on the observed binary responses $y_{ij}, i = 1, ..., N, j = 1, ..., J$. Since the model only relies on the Euclidian distance between the ideal points, two sets of points lead to the same model if they have the same configuration, i.e.,
one set of points can be obtained by applying an isometry mapping to the other.
 This is because, the distance between points is invariant under an isometry mapping. An isometry mapping $F$ in $\mathbb R^{K}$ takes the form
$$F(\mathbf x) = O\mathbf x + \mathbf b, ~~\forall \mathbf x \in \mathbb R^{K}, $$
where $O$ is a $K\times K$ orthogonal matrix and $\mathbf b $ is a vector in  $\mathbb R^{K}$ \citep[see, e.g.,][]{olver1999classical}.
We further denote $\mathcal A_K$ as the set of all isometry mappings on $\mathbb R^{K}$.
Without additional information,
the best possible result one can expect is recovering the ideal points up to an isometry mapping. We refer to this problem as the recovery of ideal point configuration.

It is worth noting that regularity conditions are needed to ensure the recovery of the configuration. That is, it is possible that there exist multiple sets of ideal points with different configurations that
lead to the same distribution of $Y_{ij}$s. In other words, the configuration of $\{\ttt_1,...,\ttt_N,\aaaa_1,...,\aaaa_J\}$ may not be unique only given the partial distance matrix.
This is known as the situation of degeneration, in which case the visualization does not convey information or can even be misleading. A simple example is given in Figure~\ref{fig:example}, where the two different configurations in the two panels have the same partial distance matrix.

\begin{figure}[h]
  \centering
   \centering
    \begin{subfigure}[b]{0.3\textwidth}
        \centering
        \includegraphics[width=0.99\textwidth]{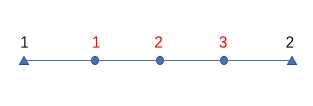}
        \caption{}
    \end{subfigure}%
    \begin{subfigure}[b]{0.3\textwidth}
        \centering
        \includegraphics[width=0.99\textwidth]{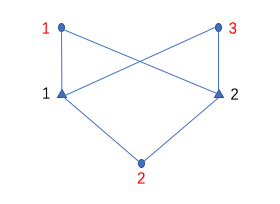}
        \caption{}
    \end{subfigure}
  \caption{An example of degenerate situation: The triangles represent item points and circles represent person points. The two configurations in $\mathbb R^2$ share the same partial distance matrix, where $d_{11} = 1^2$, $d_{12} = 3^2$,  $d_{21} = 2^2$, $d_{22} = 2^2$, $d_{31} = 3^2$, and $d_{32} = 1^2$.}\label{fig:example}
\label{fig:example}
\end{figure}

Following the above discussion, the validity of unfolding-model-based visualization relies on the accuracy of configuration recovery, a problem to be discussed.
Specifically,
we consider the following loss function for configuration recovery,
\begin{equation}\label{eq:loss}
\min_{F\in \mathcal A_K}\frac{\sum_{i=1}^N \Vert \ttt_i^* - F(\hat \ttt_i)\Vert	^2}{N} + \frac{\sum_{j=1}^J \Vert \aaaa_j^* - F(\hat \aaaa_j)\Vert^2}{J},
\end{equation}
where $\ttt_i^*$ and $\aaaa_j^*$ denote the true ideal points and $\hat \ttt_i$ and $\hat \aaaa_j$ denote the estimates from data $(y_{ij})_{N\times J}$. Note that \eqref{eq:loss} quantifies the accuracy of configuration recovery in an average sense, where isometry indeterminacy is bypassed by the minimization in \eqref{eq:loss} with respect to all isometry mappings in $\mathcal A_K$.
We call \eqref{eq:loss} the \textit{average loss} for the recovery of ideal point configuration.
Error bounds will be established for \eqref{eq:loss} under reasonable conditions, which ensures the accurate recovery of the loss function when both $N$ and $J$ are large.

\subsection{Connection with Other Scaling Methods}

MDU is closely related to MDS, a class of methods for
visualizing the similarity pattern between data points \citep{borg2005modern}.
More precisely, MDS maps a set of variables onto a low dimensional space, based on data measuring the similarity between variables. As pointed out in Chapter 14, \cite{borg2005modern},
MDU can be viewed as a special case of MDS, where the set of variables in MDS composes of both the respondents and items and the item response data $(y_{ij})_{N\times J}$ are regarded as measures of similarity between the respondents and the items, while the similarities within the two sets (i.e., respondents and items) are structurally missing; see Figure~\ref{fig:mds} for an illustration that is a reproduction of Figure 14.1 of \cite{borg2005modern}.

\begin{figure}[h]
\centering
\includegraphics[scale=0.45]{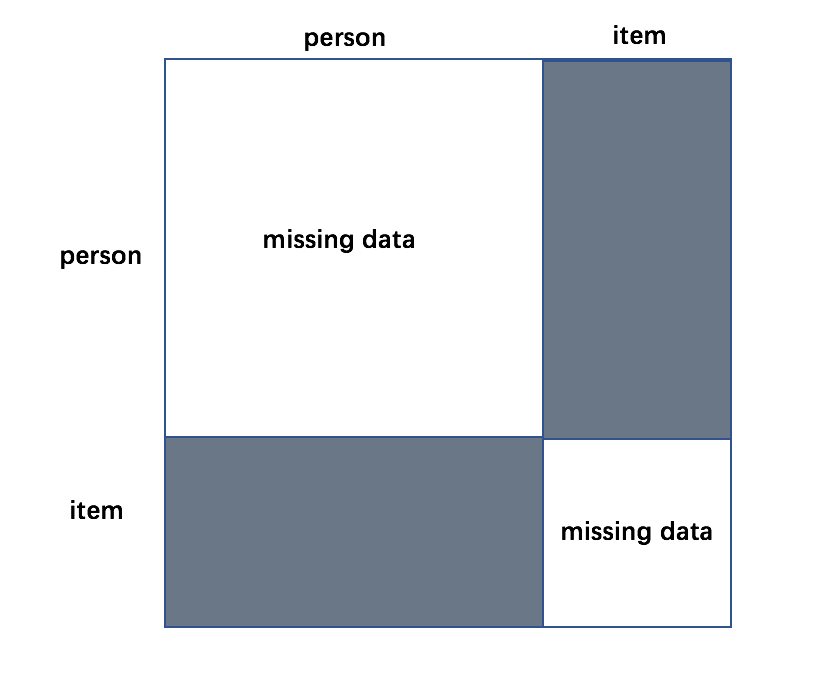}
\caption{In MDU, the diagonal blocks are missing. All we observe are the off-diagonal blocks.}
\label{fig:mds}
\end{figure}

Little statistical theory has been developed for the recovery of configuration based on MDS models. The most relevant work is \cite{zhang2016distance}, in which an error bound is developed for the recovery of the complete distance matrix, under a linear MDS model without structurally missing data. However, little discussion is provided on the recovery of ideal point configuration, under an MDU setting.

The recovery of configuration is relatively easier under the setting of MDS with no structurally missing data. This is because, the complete data matrix of similarities will provide sufficient information on the complete distance matrix.
The accurate recovery of the complete distance matrix further implies the
accurate recovery of configuration under weak conditions, due to the one-to-one relationship between the complete distance matrix and the ideal point configuration as described in Proposition~\ref{prop:configuration}. Under the MDU setting, the recovery of configuration requires additional regularity conditions, due to the lack of direct measurement of
 within-set distances. 

\begin{proposition}\label{prop:configuration}
For $\{\x_1,...,\x_n\} \subset \mathbb{R}^K$, $\{\y_1,...,\y_n\} \subset \mathbb{R}^K$, if
$\|\x_i-\x_j\| = \|\y_i-\y_j\|$
for all $i$ and $j$,
then there exists an isometry mapping $F \in \mathcal A_K$ such that
$F(\x_i)=\y_i$
for all $i = 1, ..., n$.
\end{proposition}

MDU is also related to other scaling methods for binary data such as item response theory \citep[IRT;][]{embretson2013item, reckase2009multidimensional} and multiple correspondence analysis \citep{gifi1990nonlinear,le2010multiple}.
Specifically, probabilistic models are available from IRT for multivariate binary data. An IRT model also represents respondents and items by low-dimensional parameter vectors, say $\ttt_i$ and $\aaaa_j$.
It also assumes that the probability of $Y_{ij} = 1$ is a function of  $\ttt_i$ and $\aaaa_j$.
In this sense, the model introduced above can be viewed as a special IRT model, in which
the probability of $Y_{ij} = 1$  is assumed to be a monotone decreasing function of $\| \ttt_i - \aaaa_j \|$. However, the classical IRT models \citep[see e.g.,][]{embretson2013item, reckase2009multidimensional} are not specified in this way.
Consequently, it does not make sense to visualize the person and item parameter vectors jointly.

Multiple correspondence analysis is an algorithm-based approach that can
be applied to binary data and produce low-dimensional scores for both respondents and items. These score vectors can be plotted jointly in the same space. However, as a common issue with algorithm-based approaches,
the meaning of the distance between the score vectors is not clear and the uncertainty associated with the visualization is hard to quantify.

\section{Theoretical Results}\label{sec:theory}

\subsection{Configuration Recovery based on Perturbed Partial Distances}
We first study the recovery of configuration from a perturbed partial distance matrix, when both $N$ and $J$ grow to infinity. Let $\ttt_i^*$, $i = 1, ..., N$, and $\aaaa_j^*$, $j = 1, ..., J$ be the true person and item ideal points in $\R^{K}$, respectively, and let $D^*_{N, J}$ be the corresponding partial distance matrix.
In addition, let $\tilde \ttt_i \in \R^{K_+}$ and $\tilde \aaaa_j \in \R^{K_+}$ correspond to a perturbed version of the true configuration, satisfying
\begin{equation}\label{eq:bound}
{\|\tilde{D}_{N,J}-D^*_{N,J}\|_F^2} = o(NJ)
\end{equation}
and $K_+ \geq K,$ where $\tilde{D}_{N,J}$ denotes the partial distance matrix given by the perturbed configuration and $\Vert \cdot \Vert_F$ denotes the matrix Frobenius norm. One can think of $K_+$ as the latent dimension of the MDU model being applied to data, and $\tilde \ttt_i$ and $\tilde \aaaa_j$ as some estimates of the person and item ideal points. For the time being, we treat $K_+$, $\tilde \ttt_i$, and $\tilde \aaaa_j$ as given.


Based on the definition of matrix Frobenius norm, the left side of \eqref{eq:bound} has $NJ$ terms, each of which is a squared distance between a true person-item distance and its perturbed value. Equation~\eqref{eq:bound} implies that the perturbed partial distance matrix converges to the true one in an average sense, when both $N$ and $J$ grow to infinity.

We denote
$\ttt_i^+ = ((\ttt_i^*)^\top, \mathbf 0^\top)^\top$ and  $\aaaa_j^+ = ((\aaaa_j^*)^\top, \mathbf 0^\top)^\top$ in $\mathbb R^{K_+}$ as the embedding of the true ideal points in $\R^{K_+}$,
where $\bf 0$ denotes a zero vector.
In what follows, we show that
$$\min_{F\in \mathcal A_{K_+}}\frac{\sum_{i=1}^N \Vert \ttt_i^+ - F(\tilde \ttt_i)\Vert^2}{N} + \frac{\sum_{j=1}^J \Vert \aaaa_j^+ - F(\tilde \aaaa_j)\Vert^2}{J} \to 0$$
as $N$ and $J$ grow to infinity, under reasonable conditions on the true ideal points.

Throughout this paper, we assume that ideal points are constrained in a compact set in $\mathbb R^{K}$.
\begin{itemize}
\item[A0.] There exists a constant $M$ such that
$\|\ttt_i^*\| \leq M$ and $\|\aaaa_j^*\| \leq M$ for all $i$ and $j$.
\end{itemize}
To impose regularity conditions on the true configuration of the $N+J$ ideal points, which can vary with $N$ and $J$, we introduce the notion of anchor points,
two finite sets of points in $\mathbb R^{K}$ satisfying certain regularities that are independent of $N$ and $J$.

\begin{definition}\label{def:anchor}
Two sets of points, $\{\bb_1^*,...,\bb_{k_1}^*\}, \{\cc_1^*,...,\cc_{k_2}^*\} \subset B^K_{\bf 0}(M)$, are called a collection of anchor points of $\mathbb R^{K}$, if they satisfy conditions A1 and A2 below, where  $B^K_{\bf 0}(M)$ denotes a closed ball in $\mathbb R^K$ centered at $\bf 0$ with radius $M$.

\end{definition}
Let $D^* = (\Vert \bb_i^* - \cc_j^*\Vert^2)_{k_1\times k_2}$ be the partial distance matrix based on the anchor points, whose entries are assumed to be all positive (i.e., there is no identical points). 
\begin{itemize}
\item[A1.] There exists $\eta > 0$ such that for any partial distance matrix $D \in \mathbb R^{k_1\times k_2}$ satisfying $\Vert D - D^* \Vert_F < \eta,$
$D$ has a unique configuration.
\item[A2.] Both $\{\bb_1^*,...,\bb_{k_1}^*\}$ and $\{\cc_1^*,...,\cc_{k_2}^*\}$ can affine span $\mathbb R^{K}$.
\end{itemize}

\begin{remark}\label{rk:anchor}
According to Definition \ref{def:anchor}, we still get a collection of anchor points when slightly perturbing the points in a given anchor point collection  in $\R^K$.

%
%
%
%
%
%
%
\end{remark}

According to condition A1, the anchor points are well-behaved points whose configuration can be uniquely determined by the partial distance matrix, even after a small perturbation.
In addition, thanks to A2, the anchor points will help to anchor the rest of the points in $\mathbb R^{K}$, i.e., determining the configuration of a larger set of respondent and item ideal points.

Following the above concept of anchor points, it is intuitive that if there exist anchor points
$\{\bb_1^*,...,\bb_{k_1}^*\}$ and $\{\cc_1^*,...,\cc_{k_2}^*\}$, satisfying that each $\bb_i^*$
is surrounded by sufficiently many  respondent ideal points and each $\cc_j^*$ is  surrounded by sufficiently many item ideal points; that is, there exist a sufficient number of anchor points. Then it is relatively easy to recover the configuration of the ideal points from  a perturbed partial distance matrix. This intuition is
formalized by condition A3 below.

\begin{itemize}
\item[A3.]
There exists a collection of anchor points $\{\bb_1^*,...,\bb_{k_1}^*\}$ and $\{\cc_1^*,...,\cc_{k_2}^*\} \subset B^K_{\bf 0}(M) \subset \mathbb R^K$ and $0<\epsilon<M/10$ such that the closed balls with $\bb_1^*,...,\bb_{k_1}^*$ and $\cc_1^*,...,\cc_{k_2}^*$ as centers and radius $\epsilon$, denoted by $B_{\mathbf b_1^*}(\epsilon),...,B_{\mathbf b_{k_1}^*}(\epsilon)$ and $B_{\cc_1^*}(\epsilon),...,B_{\cc_{k_2}^*}(\epsilon)$, do not overlap. The following two conditions are required to hold.
\begin{enumerate}
  \item[(1)]  For any
$\bb_1 \in B_{\bb_1^*}(\epsilon),...,\bb_{k_1} \in B_{\bb_{k_1}^*}(\epsilon)$ and $\cc_1 \in B_{\cc_1^*}(\epsilon),...,\cc_{k_2} \in B_{\cc_{k_2}^*}(\epsilon)$, $\{\bb_1,...,\bb_{k_1}\}$ and $\{\cc_1,...,\cc_{k_2}\}$ are also a collection of anchor points.
  \item[(2)] When $N$ and $J$ grow to infinity,
  \begin{align*}
& p_i = \liminf\limits_{N \to \infty}\frac{\sum_{l=1}^N 1_{\{\|\ttt^*_l - \bb_i^*\| < \epsilon\}}}{N} > 0, \quad i = 1,...,k_1,\\
& q_j = \liminf\limits_{J \to \infty}\frac{\sum_{l=1}^J 1_{\{\|\aaaa^*_l - \cc_j^*\| < \epsilon\}}}{J} > 0, \quad j = 1,...,k_2.
\end{align*}
\end{enumerate}
\end{itemize}


\begin{theorem}\label{thm:loss}
Suppose that A0 and A3 are satisfied for the true ideal points $\ttt_i^*$ and $\aaaa_j^*, i = 1, ..., N, j = 1, ..., J$. Let $\tilde \ttt_i, \tilde \aaaa_j \in B^{K_+}_{\bf 0}(M)$ correspond to a perturbed version of the true configuration, for some $K_+ \geq K$. Further let  $\tilde{D}_{N,J}$ be the corresponding partial distance matrix.
Suppose that
 $\|\tilde{D}_{N, J} - D^*_{N, J}\|_F^2 = o(NJ)$,
 when $N$ and $J$ grow to infinity. Then
\begin{equation}\label{eq:theorem}
\limsup\limits_{N, J \to \infty}\left(\min_{F\in \mathcal A_K}\frac{\sum_{i=1}^N \Vert \ttt^+_i - F(\tilde\ttt_i)\Vert^2}{N} + \frac{\sum_{j=1}^J \Vert \aaaa^+_j - F(\tilde\aaaa_j)\Vert^2}{J}\right)
\leq C\epsilon^2.
\end{equation}
where $C$ is a constant that does not depend on $N$ and $J.$  If there exists a fixed collection of anchor points, for which
A3 is satisfied  for any sufficiently small $\epsilon > 0$, then we have
\begin{equation}\label{eq:theorem_2}
\limsup\limits_{N, J \to \infty}\left(\min_{F\in \mathcal A_K}\frac{\sum_{i=1}^N \Vert \ttt^+_i - F(\tilde\ttt_i)\Vert^2}{N} + \frac{\sum_{j=1}^J \Vert \aaaa^+_j - F(\tilde\aaaa_j)\Vert^2}{J}\right)
= 0.
\end{equation}
\end{theorem}

\begin{remark}\label{rk:identifiable}

Theorem \ref{thm:loss} shows that the configuration can be recovered asymptotically
when both $N$ and $J$ grow to infinity and suitable conditions hold. The conditions required by Theorem~\ref{thm:loss} are quite mild. It first requires all the true and perturbed ideal points to be located in a compact set. Second,
as will be shown in Proposition~\ref{prop:random} below,
condition A3 is satisfied with high probability when the true person and item points are i.i.d. samples from two distributions satisfying mild conditions, respectively.
Finally, it requires that the perturbation of the partial distance matrix is not too large, i.e.,
$\|\tilde{D}_{N, J} - D^*_{N, J}\|_F^2 = o(NJ)$. As will be shown in Proposition~\ref{prop:distance}, this condition holds with high probability when $\tilde{D}_{N, J}$ is given by a likelihood-based estimator.


\end{remark}

\begin{proposition}\label{prop:random}
Suppose that $\ttt_1^*, ...,\ttt_N^*$ and $\aaaa_1^*,...,\aaaa_J^*$ are independent and identically distributed samples from distributions  $P_1$ and $P_2$,
where $P_1$ and $P_2$ have positive and continuous density functions within a ball $G \subset B^K_{\bf 0}(M).$ Then A3 holds almost surely for any sufficiently small $\epsilon > 0$.

\end{proposition}
\begin{remark}\label{rmk:decay}
We remark that constant $C$ is determined and only determined by the configuration of the anchor points in A3, according to our proof in the supplementary material.
Roughly, the more regular the  set of anchor points is (in terms of affine spanning $\mathbb R^{K}$),
the smaller the value of $C$.
\end{remark}

\begin{remark}\label{rk:confirmatory}

As discussed in Section~\ref{subsec:config}, we can only recover the ideal points up to an isometry mapping.
This isometry mapping may be fixed if one is willing to make further assumptions such as
non-negativity \citep{donoho2004does, hoyer2004non} and sparsity \citep{chen2019structured}.
In that case, one may further interpret each coordinate of the latent space. We leave this problem for future investigation.
\end{remark}

\subsection{Likelihood-based Estimation}\label{subsec:estimator}
In what follows, we propose a constrained maximum likelihood estimator
and show its properties. Given the assumptions of the MDU model,
our likelihood function takes the form
\begin{equation*}
L(\ttt_1,...,\ttt_N,\aaaa_1,...,\aaaa_J) = \prod_{i=1}^N\prod_{j=1}^J f(\|\ttt_i-\aaaa_j\|^2)^{y_{ij}}(1-f(\|\ttt_i-\aaaa_j\|^2)^{1-y_{ij}}.
\end{equation*}
Based on this likelihood function, we consider the following estimator
\begin{align}\label{eq:optimization}
\begin{split}
(\hat{\ttt}_1,...,\hat{\ttt}_N,\hat{\aaaa}_1,...,\hat{\aaaa}_J) =& \argmin_{\ttt_1,...,\ttt_N,\aaaa_1,...,\aaaa_J \in \mathbb R^{K_+}} -\log L(\ttt_1,...,\ttt_N,\aaaa_1,...,\aaaa_J)\\
 s.t.~~~ &~  \|\ttt_i\| \leq M,  \ \|\aaaa_j\| \leq M, \quad i = 1,...,N,\ j = 1,...,J.
\end{split}
\end{align}
where $K_+$ and $M$ are pre-specified. We denote $\hat D_{N,J}$ as the partial distance matrix based on $\hat \ttt_i$s and $\hat \aaaa_j$s from \eqref{eq:optimization}.

We impose the following regularity condition on the link function $f$, which requires $f$ to be neither
too steep nor too flat in the feasible domain. Similar conditions are assumed in \cite{davenport20141} for solving a 1-bit matrix completion problem.
\begin{itemize}
\item[A4.]
The link function $f: \mathbb R \rightarrow (0, 1)$ is a smooth and monotone decreasing function, satisfying
$L_{4M^2} < \infty$ and $\beta_{4M^2} < \infty,$
where
\begin{equation*}
L_{\alpha} = \sup\limits_{|x|\leq\alpha} \frac{|f'(x)|}{f(x)(1-f(x))}, \mbox{~~and~~}  \beta_{\alpha} = \sup\limits_{|x|\leq\alpha}\frac{f(x)(1-f(x))}{|f'(x)|^2}.
\end{equation*}
\end{itemize}


\begin{proposition}\label{prop:distance}
Suppose that A0 and A4 are satisfied and $K_+ \geq K$.
Then there exist $C_1$ and $C_2$ independent of $N$ and $J$, such that
\begin{equation*}
\frac{1}{NJ}\|\hat{D}_{N,J}-D_{N,J}^*\|_F^2 \leq   C_1M^2L_{4M^2}\beta_{4M^2} \sqrt{\frac{N+J}{NJ}}\sqrt{1+\frac{\log(NJ)}{N+J}},
\end{equation*}
with probability at least $1-{C_2}/{(N+J)}.$
\end{proposition}

Proposition \ref{prop:distance} implies that $\|\hat D_{N,J}-D_{N,J}^*\|_F^2 = o_p(NJ)$, which, combined with Theorem \ref{thm:loss}, leads to Theorem~\ref{thm:main} below. 
\begin{theorem}\label{thm:main}
Suppose that A0, A3 and A4 are satisfied and $K_+ \geq K$.
Then
\begin{equation}\label{eq:maintheorem}
\lim\limits_{N, J \to \infty} P \left(\min_{F\in \mathcal A_{K_+}}\frac{\sum_{i=1}^N \Vert \ttt_i^+ - F(\hat\ttt_i)\Vert^2}{N} + \frac{\sum_{j=1}^J \Vert \aaaa_j^+ - F(\hat\aaaa_j)\Vert^2}{J}
\leq C\epsilon^2 \right) = 1,
\end{equation}
where $\hat \ttt_i$ and $\hat \aaaa_j$, $i = 1, ..., N$, and $j = 1, ..., J$, are given by \eqref{eq:optimization},
$\epsilon$ is from condition A3, and  $C$ is a constant independent of $\epsilon$, $N$, and $J$. 

\end{theorem}

\begin{remark}
We remark that if A3 holds for any sufficiently small $\epsilon$, then \eqref{eq:maintheorem} implies that the loss
$$\min_{F\in \mathcal A_{K_+}}\frac{\sum_{i=1}^N \Vert \ttt_i^+ - F(\hat\ttt_i)\Vert^2}{N} + \frac{\sum_{j=1}^J \Vert \aaaa_j^+ - F(\hat\aaaa_j)\Vert^2}{J}$$
converges to zero in probability. Further note that according to Proposition~\ref{prop:random}, A3 holds with high probability for any sufficiently small $\epsilon$, under a random design for the true ideal points. Therefore, the loss can be shown to converge to zero in probability, under this random design. This result is summarized in Theorem~\ref{thm:random} below.

\end{remark}

\begin{theorem}\label{thm:random}
Suppose that A0 and A4 are satisfied and $K_+ \geq K$.
Further suppose that $\ttt_1^*, ...,\ttt_N^*$ and $\aaaa_1^*,...,\aaaa_J^*$ are independent and identically distributed samples from distributions  $P_1$ and $P_2$, respectively,
where $P_1$ and $P_2$ have positive and continuous density functions within a ball $G \subset B^K_{\bf 0}(M).$
Then for
$\hat \ttt_i$ and $\hat \aaaa_j$, $i = 1, ..., N$, and $j = 1, ..., J$, given by \eqref{eq:optimization},
the loss function 
\begin{equation*}
\min_{F\in \mathcal A_{K_+}}\frac{\sum_{i=1}^N \Vert \ttt_i^+ - F(\hat\ttt_i)\Vert^2}{N} + \frac{\sum_{j=1}^J \Vert \aaaa_j^+ - F(\hat\aaaa_j)\Vert^2}{J}
\end{equation*}
goes to 0 in probability as $N$ and $J$ grow to infinity.
\end{theorem}

\begin{remark}
We remark that the probability measures in Theorems~\ref{thm:main} and \ref{thm:random} are slightly different. The probability in Theorem~\ref{thm:main} is
based on the conditional distribution of $Y_{ij}$s given $\ttt_i^*$ and $\aaaa_j^*$, while that for
Theorem~\ref{thm:random} is based on the joint distribution of $Y_{ij}$,  $\ttt_i^*$ and $\aaaa_j^*, i = 1, ..., N,  j = 1, ..., J$.

\end{remark}

\begin{remark}

A stress function is a squared error loss function that plays an important role in the classical MDS/MDU algorithms.  It serves  not only as the objective function in the search for the MDS/MDU solution, but also
as the basis for assessing the goodness-of-fit of the solution \citep{mair2016goodness}.
In the proposed framework, the negative joint log-likelihood function plays a similar role as the stress function.
It replaces the squared loss in the stress function by a loss function based on the Kullback–Leibler divergence. Similar goodness-of-fit measures in classical MDU can be developed  under the proposed framework, based on the negative joint log-likelihood.

\end{remark}

\begin{remark}\label{rk:dimension}
We remark on the choice of latent dimension.
Theorems \ref{thm:main} and \ref{thm:random} suggest that as long as we choose $K_+$ to be no less than the true dimension $K,$ then the unfolding result is asymptotically valid. When there is no such prior knowledge about an upper bound of $K$, one can estimate the latent dimension $K$ using data.
Several methods from factor analysis and network data analysis may be adapted to the current problem, such as
trace-norm regularization \citep{bach2008consistency},
cross-validation \citep{chen2018network,li2020network}, and information criteria \citep{bai2002determining}.
We believe that consistency results on the selection of $K$ can be established.


\end{remark}


\begin{remark}
  We point out that the result of Proposition~\ref{prop:distance} can be easily extended to other MDU models, such as  models with additional parameters in the link function and models for rating and ranking data.  Then, by making use of Theorem~\ref{thm:loss}, the results of Theorems~\ref{thm:main} and \ref{thm:random} can also be extended to these models.
\end{remark}

We propose an alternating minimization algorithm
for solving \eqref{eq:optimization}.
To handle the constraints in \eqref{eq:optimization}, a projected gradient descent update is used in each iteration. For $\x \in \mathbb{R}^{K_+}$, we define the following projection operator:
\begin{equation*}
\text{Proc}_{M}(\x) = \argmin_{\|\y\| \leq M} \|\y-\x\| = \left\{
\begin{array}{ll}
\x \quad &\text{if } \|\x\| \leq M,\\
M\x/\|\x\| \quad &\text{if } \|\x\| > M.
\end{array}\right.
\end{equation*}

\begin{algorithm}[Alternating minimization algorithm]\label{algo:estimate}
~

\begin{itemize}
  \item[] \textbf{Input}: Data $(y_{ij})_{N\times J}$, pre-specified dimension $K_+$, constraint $M$, iteration number $m=1$,
   and the initial values $\ttt_1^{(0)},...,\ttt_N^{(0)}$ and $\aaaa_1^{(0)},...,\aaaa_J^{(0)}$ in $\mathbb{R}^{K_+}$.
  \item[] \textbf{Alternating minimization}: at the $m$th iteration, perform
  \begin{itemize}
    \item[(a)] For each respondent $i$, update
    $$\ttt_i^{(m)}= \text{Proc}_{M}\left(\ttt_i^{(m-1)}+ \varrho \mathbf s_i^{(m-1)}(\ttt_i^{(m-1)})\right),$$
    where
    \begin{equation*}
    \begin{aligned}
    &\mathbf s_i^{(m-1)}(\ttt) \\
   =& \frac{\partial}{\partial \ttt} \left(\sum\limits_{j=1}^J y_{ij}\log f(\|\ttt-\aaaa_j^{(m-1)}\|^2) + (1-y_{ij})\log\big(1-f(\|\ttt-\aaaa_j^{(m-1)}\|^2)\big)\right).
   \end{aligned}
    \end{equation*}
    The step size $\varrho>0$ is chosen by line search.
  \item[(b)]  For each item $j$, update
    $$\aaaa_j^{(m)}= \text{Proc}_{M}\left(\aaaa_j^{(m-1)}+ \varrho \tilde{\mathbf s}_j^{(m-1)}(\aaaa_j^{(m-1)})\right),$$
    where
    \begin{equation*}
    \begin{aligned}
    &\tilde{\mathbf s}^{(m-1)}_j(\aaaa) \\
   =& \frac{\partial}{\partial \aaaa} \left(\sum\limits_{i=1}^N y_{ij}\log f(\|\ttt^{(m)}_i-\aaaa\|^2) + (1-y_{ij})\log\big(1-f(\|\ttt^{(m)}_i -\aaaa \|^2)\big)\right).
   \end{aligned}
   \end{equation*}
   The step size $\varrho>0$ is chosen by line search.
   \item[] Iteratively perform steps (a) and (b) until convergence. Let  $m^*$ be the last iteration number upon convergence.
  \end{itemize}
  \item[] \textbf{Output}:
  $\hat \ttt_1 = \ttt_1^{(m^*)},...,\hat \ttt_N = \ttt_N^{(m^*)}$ and $\hat \aaaa_1 = \aaaa_1^{(m^*)},...,\hat \aaaa_J = \aaaa_J^{(m^*)}$.
\end{itemize}
\end{algorithm}

\begin{remark}\label{rmk:nonconvex}
Since \eqref{eq:optimization} is not a convex optimization problem, there is no guarantee that Algorithm~\ref{algo:estimate} finds the global optimal solution. However, we point out that the previous theoretical results hold even when $\{\hat \ttt_1,...,\hat \ttt_N,\hat \aaaa_1,...,\hat \aaaa_J\}$ is not a global optimal point.
Specifically, Proposition \ref{prop:distance} and Theorems \ref{thm:main} and \ref{thm:random} hold for
any $\{\hat \ttt_1,...,\hat \ttt_N,\hat \aaaa_1,...,\hat \aaaa_J\}$ satisfying the constraints in \eqref{eq:optimization} and
\begin{equation}\label{eq:localopt}
L(\hat \ttt_1,...,\hat \ttt_N,\hat \aaaa_1,...,\hat \aaaa_J) \geq L(\ttt_1^*,...,\ttt_N^*,\aaaa_1^*,..., \aaaa_J^*).
\end{equation}
According to our simulation study, estimates given by Algorithm~\ref{algo:estimate} are likely to satisfy \eqref{eq:localopt}.

\end{remark}

\subsection{Analyzing Missing Data}

We further discuss the  configuration recovery problem when data have many missing values, which is commonly encountered in practice.
Denote matrix $\Omega = (\omega_{ij})_{N\times J}$, where $\omega_{ij}=1$ indicates that response $y_{ij}$ is observed and $\omega_{ij}=0$ indicates $y_{ij}$ is missing. We consider the simple case of uniformly missing, as
described in condition A5. We point out that this assumption can be relaxed to analyzing data that have non-uniformly missing entries, following the developments in \cite{cai2013max} for solving a 1-bit matrix completion problem.

\begin{itemize}
\item[A5.] Entries of $\Omega$, $\omega_{ij}$, are independent and identically distributed Bernoulli random variables with
\begin{equation*}
P(\omega_{ij}=1)=\frac{n}{NJ}.
\end{equation*}
\end{itemize}

Under this condition, there are on average $n$ entries of the data matrix $(y_{ij})_{N\times J}$ that are observable. Thanks to the ignorable missingness, given $\Omega$ and the observed data, the likelihood  becomes
$$L^{\Omega}(\ttt_1,...,\ttt_N,\aaaa_1,...,\aaaa_J) = \prod_{\omega_{ij}=1} f(\|\ttt_i-\aaaa_j\|^2)^{y_{ij}}(1-f(\|\ttt_i-\aaaa_j\|^2)^{1-y_{ij}}.$$
We still consider a constrained maximum likelihood estimator
\begin{align}\label{eq:optimizationmissing}
\begin{split}
(\hat{\ttt}_1^{\Omega},...,\hat{\ttt}_N^{\Omega},\hat{\aaaa}_1^{\Omega},...,\hat{\aaaa}_J^{\Omega}) &= \argmin_{\ttt_1,...,\ttt_N,\aaaa_1,...,\aaaa_J \in \mathbb R^{K_+}} -\log L^{\Omega}(\ttt_1,...,\ttt_N,\aaaa_1,...,\aaaa_J)\\
 s.t.~~~ &~  \|\ttt_i\| \leq M,  \ \|\aaaa_j\| \leq M, \quad i = 1,...,N, \ j = 1,...,J.\end{split}
\end{align}
Let $\hat D_{N,J}^{\Omega}$ denote the partial distance matrix for $\hat{\boldsymbol\theta}_1^{\Omega},...,\hat{\ttt}_N^{\Omega},\hat{\aaaa}_1^{\Omega},...,\hat{\aaaa}_J^{\Omega}$.
Proposition \ref{prop:distancemissing} presents a missing-data version of Proposition~\ref{prop:distance}. It
implies that we can still recover the partial distance matrix if $n$ is large enough.

\begin{proposition}\label{prop:distancemissing}
Suppose that A0, A4 and A5 are satisfied and $K_+ \geq K$.
Then there exist $C_1$ and $C_2$ independent of $N$ and $J$, such that
\begin{equation}\label{eq:miss}
\frac{1}{NJ}\|\hat{D}_{N,J}^{\Omega}-D_{N,J}^*\|_F^2 \leq  C_1M^2L_{4M^2}\beta_{4M^2} \sqrt{\frac{N+J}{n}}\sqrt{1+\frac{NJ\log(NJ)}{n(N+J)}}\end{equation}
with probability at least $1-{C_2}/{(N+J)}.$
\end{proposition}

\begin{remark}
If $n > (N+J)\log(NJ)$, then the right side of \eqref{eq:miss}
goes to 0 as $N$ and $J$ grow to infinity, which means ${\|\hat {D}_{N,J}^{\Omega}-D^*_{N,J}\|_F^2} = o_p(NJ)$. Following the discussion in Section~\ref{sec:theory},  $\{\hat{\ttt}_1^{\Omega},...,\hat{\ttt}_N^{\Omega},\hat{\aaaa}_1^{\Omega},...,\hat{\aaaa}_J^{\Omega}\}$
provides a consistent estimate of the ideal point configuration.
This consistency result is summarized in Proposition~\ref{prop:randommissing}, which is a missing-data version of Theorem~\ref{thm:random} under the random design.
\end{remark}

\begin{proposition}\label{prop:randommissing}
Suppose that A0, A4 and A5 are satisfied, and $K_+\geq K,$ and $n > (N+J)\log(NJ)$.
Further suppose that $\ttt_1^*, ...,\ttt_N^*$ and $\aaaa_1^*,...,\aaaa_J^*$ are independent and identically distributed samples from distributions  $P_1$ and $P_2$,
where $P_1$ and $P_2$ have positive and continuous density functions within  a ball $G \subset B^K_{\bf 0}(M).$
Then the loss function
\begin{equation*}
\min_{F\in \mathcal A_{K_+}}\frac{\sum_{i=1}^N \Vert \ttt_i^+ - F(\hat \ttt_i^{\Omega})\Vert^2}{N} + \frac{\sum_{j=1}^J \Vert \aaaa_j^+ - F(\hat \aaaa_j^{\Omega})\Vert^2}{J}
\end{equation*}
goes to 0 in probability as $N$ and $J$ grow to infinity.
\end{proposition}

\section{Simulation Studies}\label{sec:simu}

In what follows, simulation studies are conducted to verify our theoretical results. Specifically, we consider a random design where the true ideal points are generated from distributions. All the analyses in this section, as well as those in Section~\ref{sec:real}, are based on our implementation of Algorithm~\ref{algo:estimate} in statistical software R.


\subsection{Study I}

\paragraph{Setting.}
We first consider a setting where $K_+$ is chosen to be exactly $K$.
 We consider MDU in a two-dimensional latent space, i.e., $K = K_+ = 2$.
Diverging sequences of $J$ and $N$ are considered, by letting $J = 200,400,...,1000$ and $N=20J$.
For given $N$ and $J$, 100 independent datasets are generated.  For each dataset, we first sample $\ttt_i^*$s and $\aaaa_j^*$s uniformly from $B^2_{\bf0}(1)$, a ball in $\mathbb R^2$ with
center $\bf0$ and radius $1$. Then given the ideal points, response data $Y_{ij}$ are generated under the link function $f(x) = 2/(1+\exp(x+0.1))$. It can be easily verified that condition A4 is satisfied for this link function.

For each dataset, we obtain an estimate of the ideal points, by applying Algorithm~\ref{algo:estimate} ten times with
random starting points and then choosing the result that gives the largest likelihood function value.
The use of multiple starting points substantially reduces the risk of the algorithm  converging to bad local minima.
In the application of Algorithm \ref{algo:estimate}, the constraint $M$ is set to 1.5.


\paragraph{Results.}
We first check the obtained likelihood function values for the 100 datasets. As we point out in Remark~\ref{rmk:nonconvex}, Proposition \ref{prop:distance} and Theorems \ref{thm:main} and \ref{thm:random} still hold as long as the estimate satisfies  \eqref{eq:localopt}, even if the global solution to the optimization~\eqref{eq:optimization} is not obtained.
It is found that by using ten random starting points, the likelihood function at the estimated parameters is always larger than that at the true parameters for all the 100 datasets.

We then present the average squared Frobenius loss for the recovery of the partial distance matrix, $\|\hat D_{N,J} - D_{N,J}^*\|_F^2/(NJ)$. These results are given in Table~\ref{tab:loss_d_K22} which presents the 25\%, 50\%, and 75\% quantiles of the loss based on the 100 datasets. From this table, we see that the loss tends to decrease as the sample size increases, supporting the result of Proposition~\ref{prop:distance}.

\begin{table}
\centering
\begin{tabular}{r|r|r|r|r|r}
  \hline
 & $J = 200$ & $J = 400$ & $J = 600$ & $J= 800$ & $J = 1000$ \\
  \hline
25\% & 0.0630 & 0.0323 & 0.0218 & 0.0164 & 0.0131 \\
  median & 0.0647 & 0.0328 & 0.0222 & 0.0167 & 0.0134 \\
  75\% & 0.0666 & 0.0336 & 0.0227 & 0.0170 & 0.0135 \\
   \hline
\end{tabular}
\caption{Simulation Study I: The average squared Frobenius loss of partial distance when $J$ increases from 200 to 1000. For each $J$, the table shows the 25\%, 50\% and 75\% quantiles of the loss based on 100 independent experiments.}
\label{tab:loss_d_K22}
\end{table}

%

Table~\ref{tab:loss_c_K22} presents the results on loss
\eqref{eq:loss} for configuration recovery, where the
best isometry mapping $F$ in \eqref{eq:loss} is obtained by solving an optimization problem given the true and estimated ideal points. Similar to the results on partial distance matrix recovery, the loss \eqref{eq:loss} also decreases towards 0 as $J$ grows large, which is consistent with the result of Theorem~\ref{thm:random}.

 \begin{table}
\centering
\begin{tabular}{r|r|r|r|r|r}
  \hline
 & $J= 200$ & $J = 400$ & $J = 600$ & $J = 800$ & $J = 1000$ \\
  \hline
25\% & 0.0158 & 0.0079 & 0.0053 & 0.0040 & 0.0032 \\
  median & 0.0160 & 0.0080 & 0.0053 & 0.0040 & 0.0032 \\
  75\% & 0.0162 & 0.0080 & 0.0054 & 0.0040 & 0.0032 \\
   \hline
\end{tabular}
\caption{Simulation Study I: The average loss for  configuration  recovery when $J$ increases from 200 to 1000. For each $J$, the table shows the 25\%, 50\% and 75\% quantiles of the loss based on 100 independent experiments.}
\label{tab:loss_c_K22}
\end{table}

Finally, the computation time on a standard desktop machine\footnote{ All the computation is conducted on a single Intel\circledR Gold 6130 core.} for solving \eqref{eq:optimization} is shown in
Table~\ref{tab:time_K22}. It is worth pointing out that since the update of person and item parameters in each iteration of Algorithm~\ref{algo:estimate} can be run in parallel, the computation can be further speeded up substantially  by parallel computing.

\begin{table}
\centering
\begin{tabular}{r|r|r|r|r|r}
  \hline
 & $J = 200$ & $J = 400$ & $J = 600$ & $J = 800$ & $J = 1000$ \\
  \hline
25\% & 98.2 & 109.6 & 144.1 & 191.5 & 254.6 \\
  median & 113.8 & 120.1 & 156.0 & 201.5 & 272.6 \\
  75\% & 128.9 & 138.0 & 176.9 & 213.9 & 286.7 \\
   \hline
\end{tabular}
\caption{Simulation Study I: The computation time of optimization \eqref{eq:optimization} when $J$ increases from 200 to 1000. For each $J$, 25\%, 50\% and 75\% quantiles of the computation time from 100 independent experiments are shown.}
\label{tab:time_K22}
\end{table}

\subsection{Study II}

\paragraph{Setting.} We now consider a setting where $K_+ > K$. We take the same setting as in Study I, except that we set $K_+ = 3$ when fitting the MDU model.
The same as Study I, for each pair of $N$ and $J$, 100 independent datasets are generated. For each dataset, Algorithm~\ref{algo:estimate} is applied similarly, using 10 random starting points and constraint parameter $M = 1.5$.


\paragraph{Results.} The results are given in Tables~\ref{tab:loss_d_K23} through \ref{tab:time_K23}.
Similar to  Tables~\ref{tab:loss_d_K22}--\ref{tab:time_K22}, these three tables also show the results on partial distance matrix recovery, configuration recovery, and computation time, respectively.
Comparing with the results of Study I, we see that both losses for the recovery of partial distance matrix and configuration tend to be larger. This is due to the overfitting brought by adding unnecessary parameters in the model.
 The computation time also increases compared with that of Study I.



\begin{table}
\centering
\begin{tabular}{r|r|r|r|r|r}
  \hline
 & $J = 200$ & $J = 400$ & $J = 600$ & $J = 800$ & $J = 1000$ \\
  \hline
25\% & 0.0734 & 0.0384 & 0.0261 & 0.0198 & 0.0159 \\
  median & 0.0758 & 0.0390 & 0.0265 & 0.0200 & 0.0161 \\
  75\% & 0.0780 & 0.0398 & 0.0269 & 0.0204 & 0.0163 \\
   \hline
\end{tabular}
\caption{Simulation Study II: The average squared Frobenius loss of partial distance when $J$ increases from 200 to 1000. For each $J$, the table shows the 25\%, 50\% and 75\% quantiles of the loss based on 100 independent experiments.}
\label{tab:loss_d_K23}
\end{table}

\begin{table}
\centering
\begin{tabular}{r|r|r|r|r|r}
  \hline
 & $J = 200$ & $J = 400$ & $J = 600$ & $J = 800$ & $J= 1000$  \\
  \hline
25\% & 0.0853 & 0.0568 & 0.0452 & 0.0386 & 0.0343 \\
  median & 0.0862 & 0.0573 & 0.0455 & 0.0390 & 0.0345 \\
  75\% & 0.0877 & 0.0580 & 0.0459 & 0.0392 & 0.0346 \\
   \hline
\end{tabular}
\caption{Simulation Study II: The average loss for  configuration  recovery when $J$ increases from 200 to 1000. For each $J$, the table shows the 25\%, 50\% and 75\% quantiles of the loss based on 100 independent experiments.}
\label{tab:loss_c_K23}
\end{table}

\begin{table}
\centering
\begin{tabular}{r|r|r|r|r|r}
  \hline
 & $J = 200$ & $J = 400$ & $J = 600$ & $J = 800$ & $J = 1000$ \\
  \hline
25\% & 106.3 & 264.0 & 639.3 & 1294.0 & 2302.5 \\
  median & 110.0 & 286.5 & 698.8 & 1407.7 & 2480.8 \\
  75\% & 112.9 & 308.9 & 793.8 & 1551.0 & 2841.1 \\
   \hline
\end{tabular}
\caption{Simulation Study II: The computation time of optimization \eqref{eq:optimization} when $J$ increases from 200 to 1000. For each $J$, 25\%, 50\% and 75\% quantiles of the computation time from 100 independent experiments are shown.}
\label{tab:time_K23}
\end{table}



\section{Real Examples}\label{sec:real}

\subsection{Example I: Movie Data}\label{subsec:real1}

\paragraph{Background.} We apply MDU to a movie rating dataset from the famous MovieLens project \citep[see e.g.,][]{harper2016movielens}. The dataset analyzed in this paper is a subset of a benchmark MovieLens dataset collected during a seven-month period from September, 1997 through April, 1998\footnote{The dataset can be downloaded from \url{https://grouplens.org/datasets/movielens/100k/}}. This subset contains 943 users and 338 movies,
obtained by selecting movies that have been rated by at least 100 users. Unlike many analyses of MovieLens data that focus on the rating scores, we consider to unfold the rating behavior itself (i.e., rated/not rated) which may also reveal the users' preference patterns. More precisely, we let $Y_{ij} = 1$ if movie $j$ has been rated by user $i$ and $Y_{ij} = 0$ otherwise.

\paragraph{Analysis.} For visualization purpose, we unfold the data onto a two-dimensional space.
To apply the MDU model introduced in this paper, we need to specify the link function $f$. We assume $f$ to take the logistic form $f(x) = 2/(1+\exp(x+\delta))$, where $\delta$ is a pre-specified small positive constant. For any $\delta > 0$, it is easy to check that the regularity condition A4 is satisfied. The results presented below are based on the choice $\delta = 0.1$, but we point out that other choices of $\delta$ ($\delta = 0.05, 0.15, 0.2$) have also been tried which all lead to very similar results.  The constraint constant $M$ is set to 3.5 when applying Algorithm~\ref{algo:estimate}. After obtaining the estimate, we transform the estimated ideal points by an isometry mapping, so that the $x$-axis corresponds to the dimension along which the estimated movie ideal points have the highest variance. As will be described in the sequel, under this isometry mapping of the estimated ideal points, both the $x$- and $y$-axes receive good interpretations.


\paragraph{Results.} The results from the MDU analysis are presented in Figures \ref{fig:whole} through \ref{fig:user_xy}. Figure~\ref{fig:whole} jointly visualizes the estimated movie and user points. As we can see, the movies and the users tend to form two giant clusters that only slightly overlap.

\begin{figure}
\centering
\includegraphics[scale = 0.55]{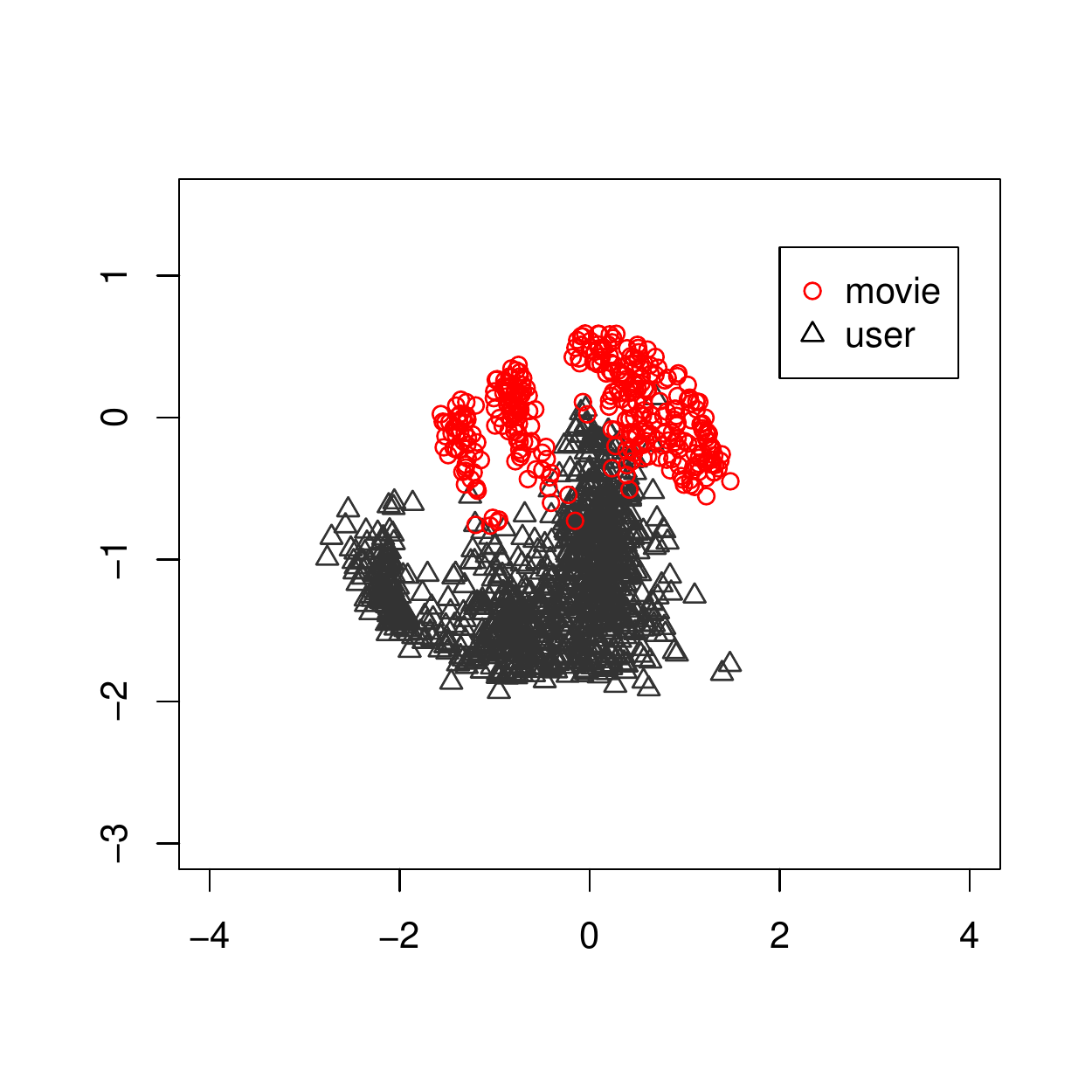}
\caption{Analysis of movie rating data: Simultaneous visualization of the estimated movie and user points.}
\label{fig:whole}
\end{figure}

We investigate the movie points. First, the y-axis of the space largely indicates, if not perfectly,  the popularity of the movies. The movies with a smaller $\hat a_{j2}$ value tends to be rated more frequently. Roughly speaking,
the shorter the average distance from a movie to the user points, the more often the movie is rated.
In fact, the Kendall's tau rank correlation between $\hat a_{j2}$s and the numbers of ratings received by the movies is $-0.66$. This phenomenon is further reflected by panel (a) of Figure~\ref{fig:movie_xy}, where movies are stratified by the numbers of ratings they received into four categories. These four categories tend to be ordered along the $y$-axis. We list four movies as examples, as indicated
in panel (a) of Figure~\ref{fig:movie_xy}. From the top to the bottom, they are Batman Forever (1995), Golden Eye (1995), Get Shorty (1995) and The Godfather (1972), respectively. Based on our interpretation of the $y$-axis, these four movies are ordered from the least popular to the most popular. 

Second, the $x$-axis of the space seems to indicate the release time of the movies. The Kendall's tau correlation
between $\hat a_{j1}$s and the release dates of the movies is -0.70. As shown in panel (b) of Figure~\ref{fig:movie_xy}, where the movies are stratified into three categories, namely ``before 1995", ``1995-1996", and ``1997-1998".  According to this figure, the clustering patten of the movies can be largely explained by the three categories based on the movie release dates. From the right to the left of the space, the points correspond to movies from the relatively older ones to the relatively more recent ones. For example, the three movies indicated in panel (b) of Figure~\ref{fig:movie_xy} are, from left to right, Citizen Kane (1941), Twelve Monkeys (1995) and The Devil's Own (1997), respectively.

\begin{figure}
  \centering
    \begin{subfigure}[b]{0.4\textwidth}
        \centering
        \includegraphics[scale = 0.5]{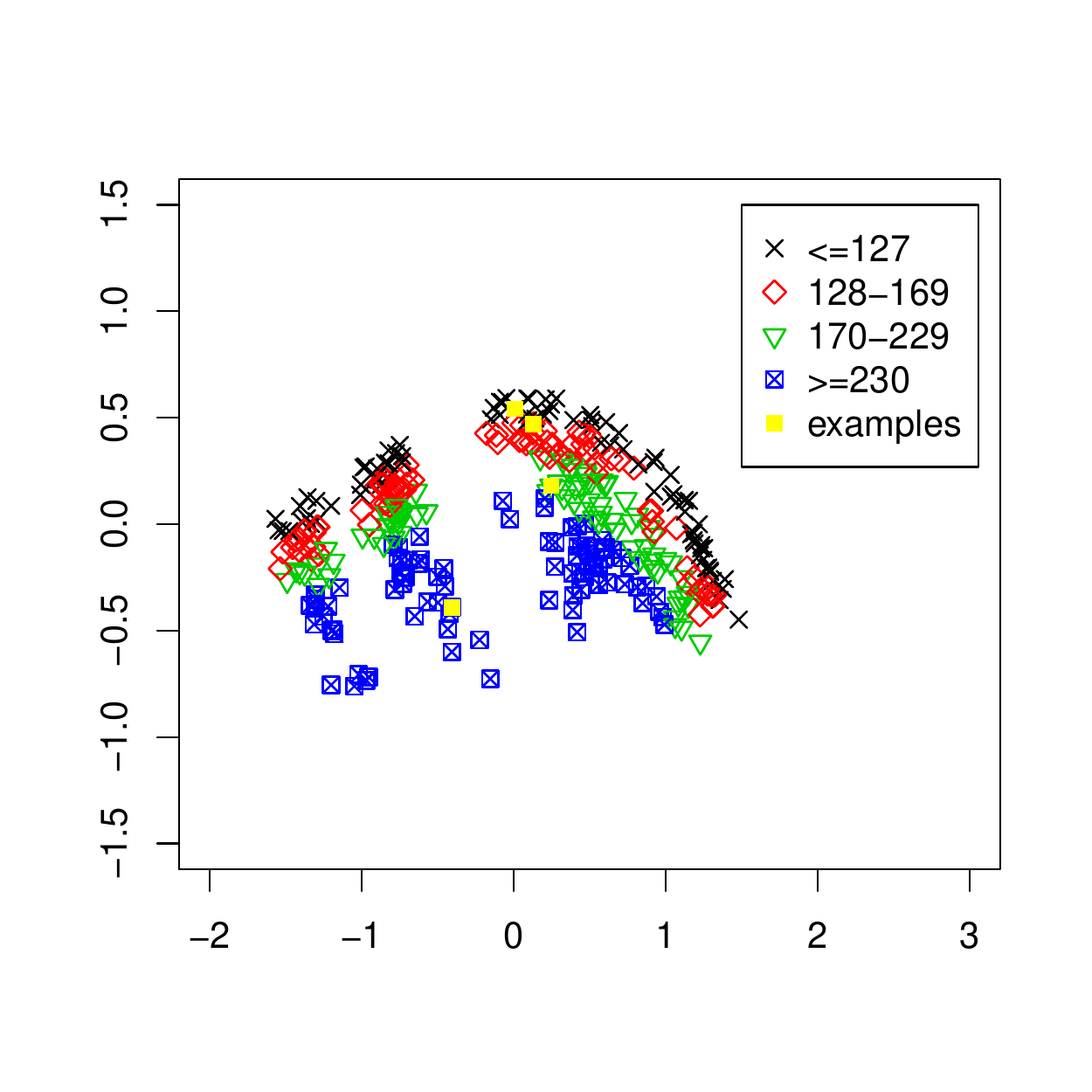}
        \caption{}
    \end{subfigure}
    \begin{subfigure}[b]{0.4\textwidth}
        \centering
        \includegraphics[scale = 0.5]{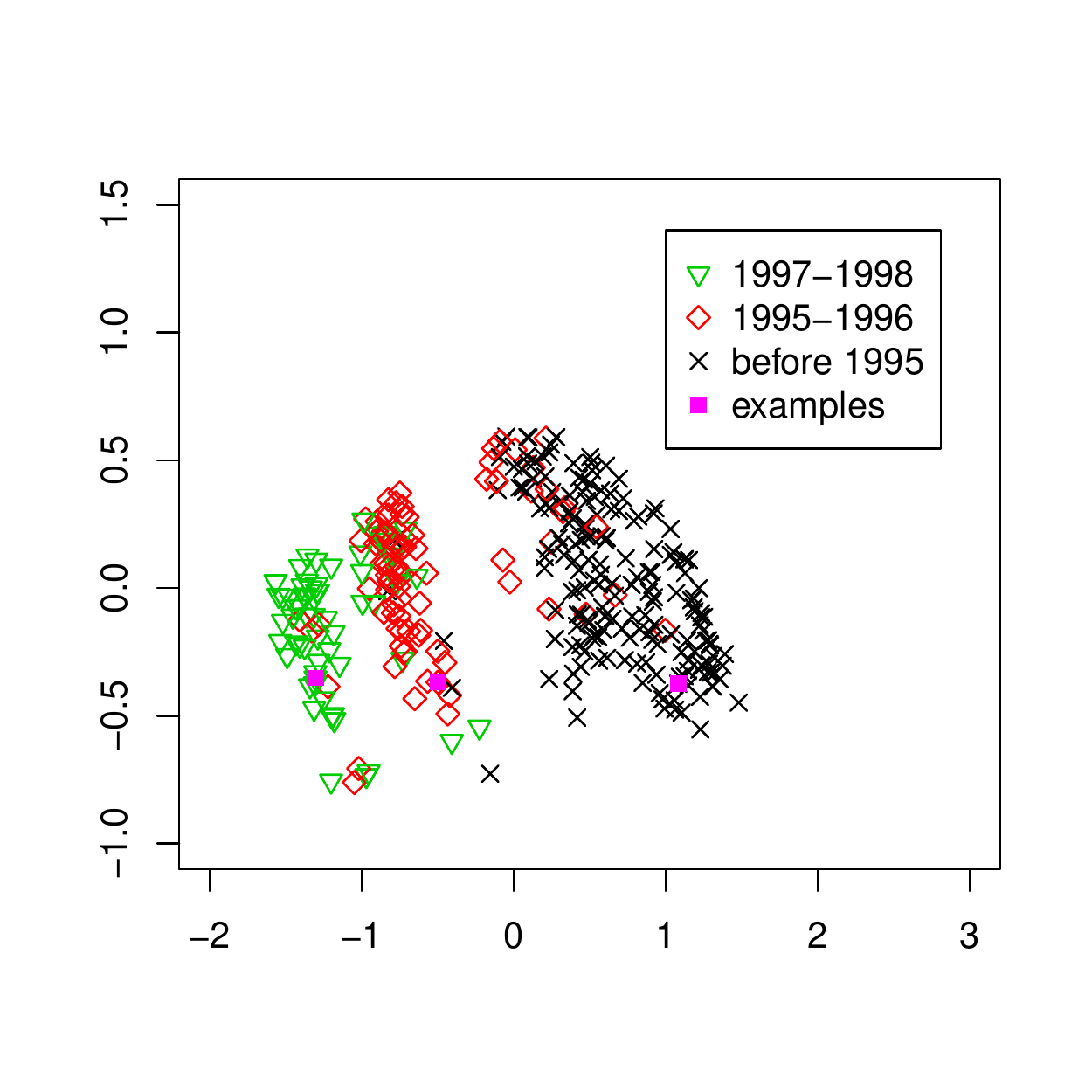}
        \caption{}
    \end{subfigure}
  \caption{Analysis of movie rating data. Panel (a): Visualization of movie points, with movies stratified into four equal-size categories based on the numbers of rating. Movies with numbers of rating less than 127, 128-169, 170-229 and more than 230 are indicated by black, red, green and blue points, respectively. Yellow points represent example movies. Panel (b): Visualization of movie points, with movies stratified into three categories based on their release time. Movies released in 1997-1998, 1995-1996, and before 1995 are indicated by green, red and black points, respectively. Purple points represent example movies.}
 \label{fig:movie_xy}
\end{figure}


\begin{figure}
 \centering
        \includegraphics[scale = 0.6]{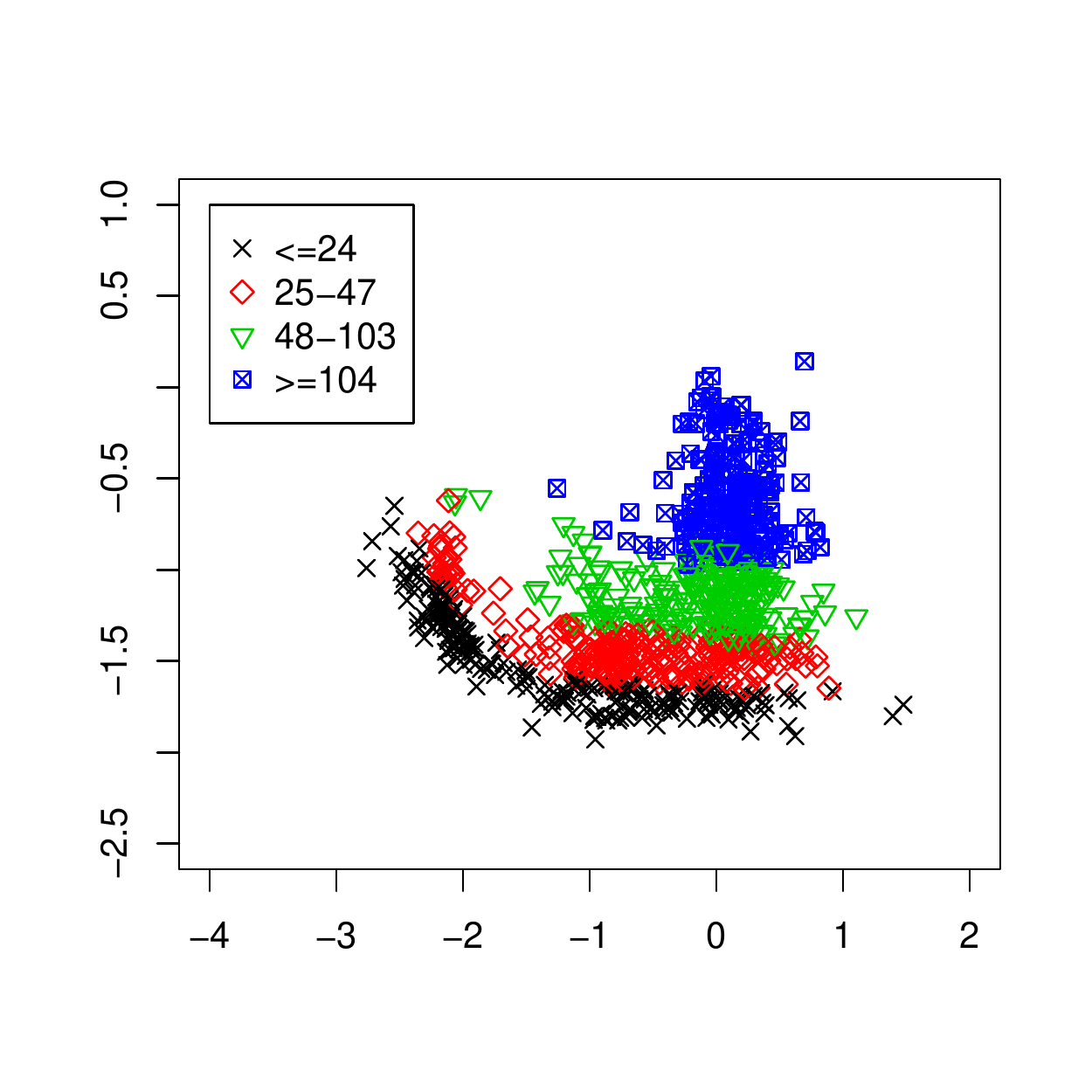}
        \caption{Analysis of movie rating data: Visualization of user points, with users classified into four equal-size categories based on the numbers of rating. Users who rated less than 24, , 25-47, 48-103 and more than 104 movies are indicated by
        black, red, green and blue points, respectively.}
 \label{fig:user_xy}
\end{figure}

The interpretation of the latent space based on movies facilitates the interpretation of the user points. First, the $y$-axis corresponds to the users' activeness. Roughly speaking, the shorter the average distance from a user point to the movies points, the more active the user is.
 The Kendall's tau rank correlation between $\hat \theta_{i2}$s and the numbers of ratings given by the users is 0.73. This is further shown via Figure~\ref{fig:user_xy}, where users are classified into four equal-size groups depending on the number of movies they rated. These groups of users, from the  most active one to the least active one, lie from the top to the bottom. Second, based on the alignment of movies along the $x$-axis, the user points from right to left may be interpreted as the ones who tend to more frequently rate relatively older movies to the ones who tend to more frequently rate relatively more recent ones.

\subsection{Example II: Senate Roll Call Voting Data}\label{subsec:real2}
\paragraph{Background.}
We now analyze a senate roll call voting dataset from the 108th congress. This dataset contains the voting records from 100 senators to 675 roll calls in years 2003 and 2004\footnote{The dataset can be downloaded from \url{https://legacy.voteview.com/dwnl.htm}.}.  Among the 100 senators, there are 48 from the Democratic party, 51 from the Republican party, and one independent politician.
For each roll call $j$, the vote of senator $i$ is recorded in three ways, ``Yea'', ``Nay'' and ``Not Voting'', treated as $Y_{ij} = 1, 0, $ and missing, respectively.

\paragraph{Analysis.}
Similar analysis as the previous one is conducted. Specifically, we unfold the data into a two-dimensional space. The same link function $f$ and constraint constant $M$ are adopted as in the analysis of movie data.
After getting the estimate, we transform the estimated ideal points by an isometry mapping, so that the $x$-axis corresponds to the dimension along which the estimated senate ideal points have the highest variance.

\paragraph{Results.}
The results are presented in Figures \ref{fig:vote108_whole} through \ref{fig:vote108_bill}. In Figure~\ref{fig:vote108_whole}, the ideal points of both roll calls and senators are visualized simultaneously. As we can see, most of the roll calls and all the senators tend to lie around a one-dimensional line.
This visualization is still valid, 
in the sense that even when the true latent dimension is one, 
according to Theorem~\ref{thm:random},
unfolding the data in a two-dimensional space is still consistent.

This phenomenon of degeneration is quite consistent with the overall unidimensional patten in the congress voting data throughout the history. It has been well recognized in the political science literature \citep{poole1991dimensionalizing,poole1991patterns} that senate voting behavior is essentially unidimensional, though slightly different latent space models are used in that literature. For example, \cite{poole1991dimensionalizing} concluded that ``to the extent that congressional voting can be described by a spatial model, a unidimensional model is largely (albeit not entirely) sufficient."

We first interpret the senators. In Figure~\ref{fig:vote108_senate}, all the senator points are visualized with their party membership indicated by different point types. In Table~\ref{table:senator}, we rank the senators based on their value of $\hat \theta_{i1}$, which is presented along the $x$-axis. According to this table, the Democrats tend to lie on the left and the Republicans tend to be on the right. In fact, this ranking is largely consistent with \textit{National Journal}'s liberalness ranking of the senators in 2003.
\textit{National Journal}'s ranking result, which is replicated in  \cite{clinton2004most},
is obtained by unfolding the senators' votes on 62 key roll calls using a model given in \cite{clinton2004statistical}. The Kendall's tau rank correlation between the result in Table~\ref{table:senator} and that given by the \textit{National Journal} is 0.79. In fact, Senator John Kerry is ranked the most liberal by both our model and by \textit{National Journal} and Senator Craig L. Thomas, who is the most conservative senator according to the ranking of \textit{National Journal}, is the third most conservative senator given by our model.

\begin{figure}
\centering
\includegraphics[scale = 0.5]{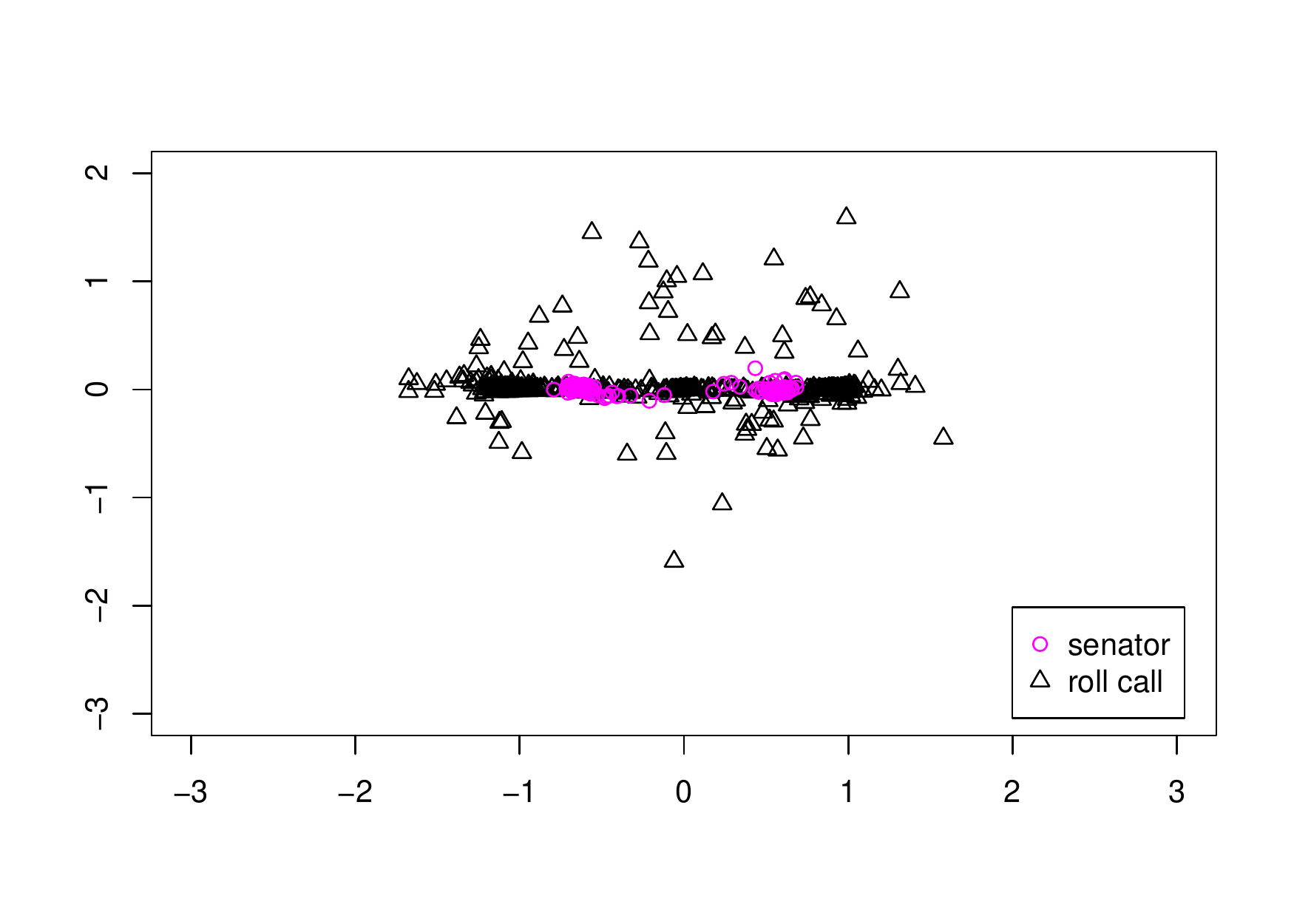}
\caption{Analysis of senator roll call data: Simultaneous visualization of the estimated senator and roll call ideal points.}
\label{fig:vote108_whole}
\end{figure}

\begin{figure}
\centering
\includegraphics[scale = 0.5]{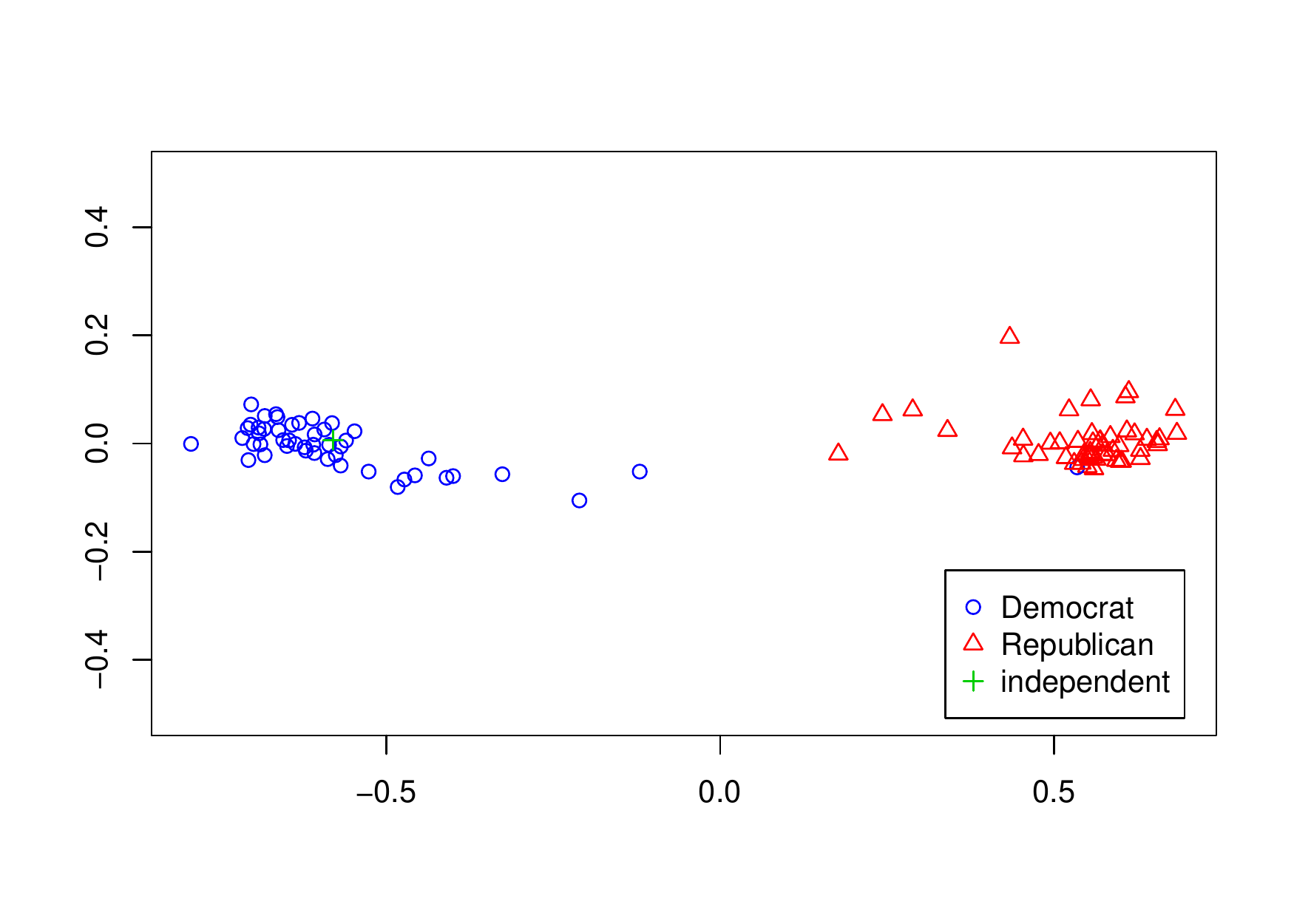}
\caption{Analysis of senator roll call data: Visualization of senator points, where senators are classified by their party membership. Specifically,
The Democrats, Republicans and an independent politician are indicated by blue, red, and green, respectively.}
\label{fig:vote108_senate}
\end{figure}


\begin{table}
  \centering
  \begin{tabular}{lll|lll|lll}
    \hline
      & Name & State&  & Name & State & & Name & State \\
      \hline
 1 & Kerry   & D-MA  & 35 & Johnson     & D-SD &69 & Grassley    & R-IO \\
  2 & Sarbanes    & D-MD &  36 & Lieberman   & D-CT &70 & Bond        & R-MO \\
  3 & Reed        & D-RH & 37 & Bingaman    & D-NM &71 & Roberts     & R-KA \\
  4 & Harkin      & D-IO & 38 & Nelson   & D-FL  &72 & Gregg       & R-NH \\
  5 & Graham   & D-FL & 39 & Dorgan      & D-ND  &73 & Allen       & R-VI \\
  6 & Lautenberg  & D-NJ &40 & Conrad      & D-ND  &74 & Domenici    & R-NM  \\
  7 & Edwards     & D-NC &41 & Carper      & D-DE &75 & Bennett     & R-UT \\
  8 & Kennedy   & D-MA  &42 & Pryor       & D-AR  &76 & Dole        & R-NC \\
  9 & Durbin      & D-IL &43 & Bayh        & D-IN &77 & Frist       & R-TN \\
  10 & Levin   & D-MI & 44 & Lincoln     & D-AR &78 & Brownback   & R-KA   \\
  11 & Akaka       & D-HA &45 & Landrieu    & D-LO &79 & Hatch       & R-UT \\
  12 & Byrd   & D-WE &46 & Baucus      & D-MT  &80 & Cochran     & R-MS  \\
  13 & Boxer       & D-CA &47 & Breaux      & D-LO &81 & Graham      & R-SC \\
  14 & Corzine     & D-NJ &48 & Nelson   & D-NE & 82 & Alexander   & R-TN   \\
  15 & Clinton     & D-NY &49 & Chafee      & R-RH &83 & Lott        & R-MS \\
  16 & Leahy       & D-VE &50 & Snowe       & R-ME &84 & Chambliss   & R-GE  \\
  17 & Dodd        & D-CT &51 & Collins     & R-ME  &85 & Burns       & R-MT \\
  18 & Stabenow    & D-MI &52 & Specter     & R-PE &86 & Bunning     & R-KE  \\
  19 & Mikulski    & D-MD &53 & Mccain      & R-AZ &87 & Crapo       & R-ID \\
  20 & Feingold    & D-WI &54 & Dewine      & R-OH  &88 & Mcconnell   & R-KE  \\
  21 & Rockefeller & D-WE &55 & Campbell    & R-CO &89 & Ensign      & R-NV   \\
  22 & Hollings    & D-SC &56 & Smith   & R-OR &90 & Cornyn      & R-TX \\
  23 & Kohl        & D-WI &57 & Coleman     & R-MN &91 & Sununu      & R-NH \\
  24 & Inouye      & D-HA &58 & Warner      & R-VI &92 & Santorum    & R-PE \\
  25 & Schumer     & D-NY &59 & Murkowski   & R-AK &93 & Craig       & R-ID   \\
  26 & Cantwell    & D-WA &60 & Voinovich   & R-OH  &94 & Inhofe      & R-OK  \\
  27 & Dayton      & D-MN &61 & Hutchison   & R-TX &95 & Allard      & R-CO \\
  28 & Murray      & D-WA &62 & Lugar       & R-IN &96 & Enzi        & R-WY \\
  29 & Wyden       & D-OR &63 & Miller      & D-GE &97 & Sessions    & R-AL  \\
  30 & Daschle     & D-SD &64 & Fitzgerald  & R-IL &98 & Thomas      & R-WY  \\
  31 & Biden       & D-DE &65 & Talent      & R-MO &99 & Kyl         & R-AZ \\
  32 & Feinstein   & D-CA &66 & Hagel       & R-NE & 100 & Nickles  & R-OK \\
  33 & Jeffords    & I-VE &67 & Stevens     & R-AK  \\
  34 & Reid        & D-NV &68 & Shelby      & R-AL \\

    \hline
  \end{tabular}
  \caption{Analysis of senator roll call data: Ranking of senators based on $\hat \theta_{i1}$.}\label{table:senator}
\end{table}

From Figure~\ref{fig:vote108_senate} and Table~\ref{table:senator}, it is also worth noting that there is a Democrat whose estimated ideal point is mixed together with those of the Republicans.
This senator is Zell Miller from the state of Georgia. He is a conservative Democrat and in fact, he supported Republican President George W. Bush against the Democratic nominee John Kerry in the presidential election in 2004.

In this congress, there is an independent senator, Jim Jeffords from the state of Vermont, who does not belong to either of the two major parties.
As we can see from both Figure~\ref{fig:vote108_senate} and Table~\ref{table:senator}, his ideal point lies on the left, mixed with many ideal points of the Democrats. This is also consistent with Senator Jim Jeffords' political standing. In fact, he left Republican party to become an independent and began caucusing with the Democrats since 2001.

We now investigate the roll calls. The value of $\hat{a}_{j1}$, i.e., the roll calls' coordinate on the $x$-axis, seems to represent the roll calls' liberalness-conservativeness. The more liberal roll calls lie on the left and the more conservative ones lie on the right. This interpretation is further confirmed by the voting records for the roll calls. In particular, for each roll call, we calculate the proportion of Republicans among the senators who voted ``Yea". A larger value of this proportion indicates that the roll call is more conservative. 
As we can see from panel (a) of Figure~\ref{fig:vote108_bill}, for roll calls from the left to the right, this proportion increases. In fact, the Kendall's tau rank correlation between $\hat a_{j1}$s and the proportions of ``Yea" from Republicans is as high as 0.88. We present the content of three roll calls as representative examples. As indicated in panel (a) of Figure~\ref{fig:vote108_bill}, these roll calls have substantially different coordinates along the $x$-axis. From left to right, they are
(1) ``To improve the availability of contraceptives for women",
(2)``Confirmation Thomas J. Ridge, of Pennsylvania, to be Secretary of Homeland Security",
and (3)``To provide financial security to family farm and small business owners by ending the unfair practice of taxing someone at death".


Although most of the roll calls lie near the $x$-axis (i.e., $\hat a_{j2} \approx 0$), there are still quite a few roll calls which spread out on the $y$-axis. It seems that the voting on such roll calls is heterogeneous within both parties. 
Specifically, we measure heterogeneity of voting within each party by a cross entropy measure, defined as
$$\mbox{CE}^{(i)}_j = - p_{jy}^{(i)}\log p_{jy}^{(i)} - p_{jn}^{(i)}\log p_{jn}^{(i)} - p_{jm}^{(i)}\log p_{jm}^{(i)},$$
where $i = 1, 2$ indicate Democrat and Republican, respectively, and $p_{jy}^{(i)}$, $p_{jn}^{(i)}$, and $p_{jm}^{(i)}$
denote the proportions of ``Yea", ``Nay", and ``Not voting" within the party for the $j$th roll call.   Cross entropy is a commonly used measure of heterogeneity \citep[Chapter 9,][]{friedman2001elements}. The larger the cross entropy, the more heterogeneous voting behavior within a party. In panel (b) of Figure~\ref{fig:vote108_bill}, we present the box plots of $\min\{\mbox{CE}^{(1)}_j, \mbox{CE}^{(2)}_j\}$, for roll calls lying near the $x$-axis ($|\hat a_{j2} | \leq 0.05$) and for those spreading out along the $y$-axis ($|\hat a_{j2} | > 0.05$). According to panel (b) of Figure~\ref{fig:vote108_bill}, the roll calls in the latter group ($|\hat a_{j2} | > 0.05$) tend to have a larger value of $\min\{\mbox{CE}^{(1)}_j, \mbox{CE}^{(2)}_j\}$, implying that the voting tends to be more heterogeneous within both parties for these roll calls.  The latter group contains roll calls, such as ``To provide for the distribution of funds under the infrastructure performance and maintenance program", ``To enhance the role of Congress in the oversight of the intelligence and intelligence-related activities of the United States Government",
and ``To strike provisions relating to energy tax incentives". Many of such roll calls may be explained by constituency specific factors.

\begin{figure}
  \centering
    \begin{subfigure}[b]{0.4\textwidth}
        \centering
        \includegraphics[scale = 0.4]{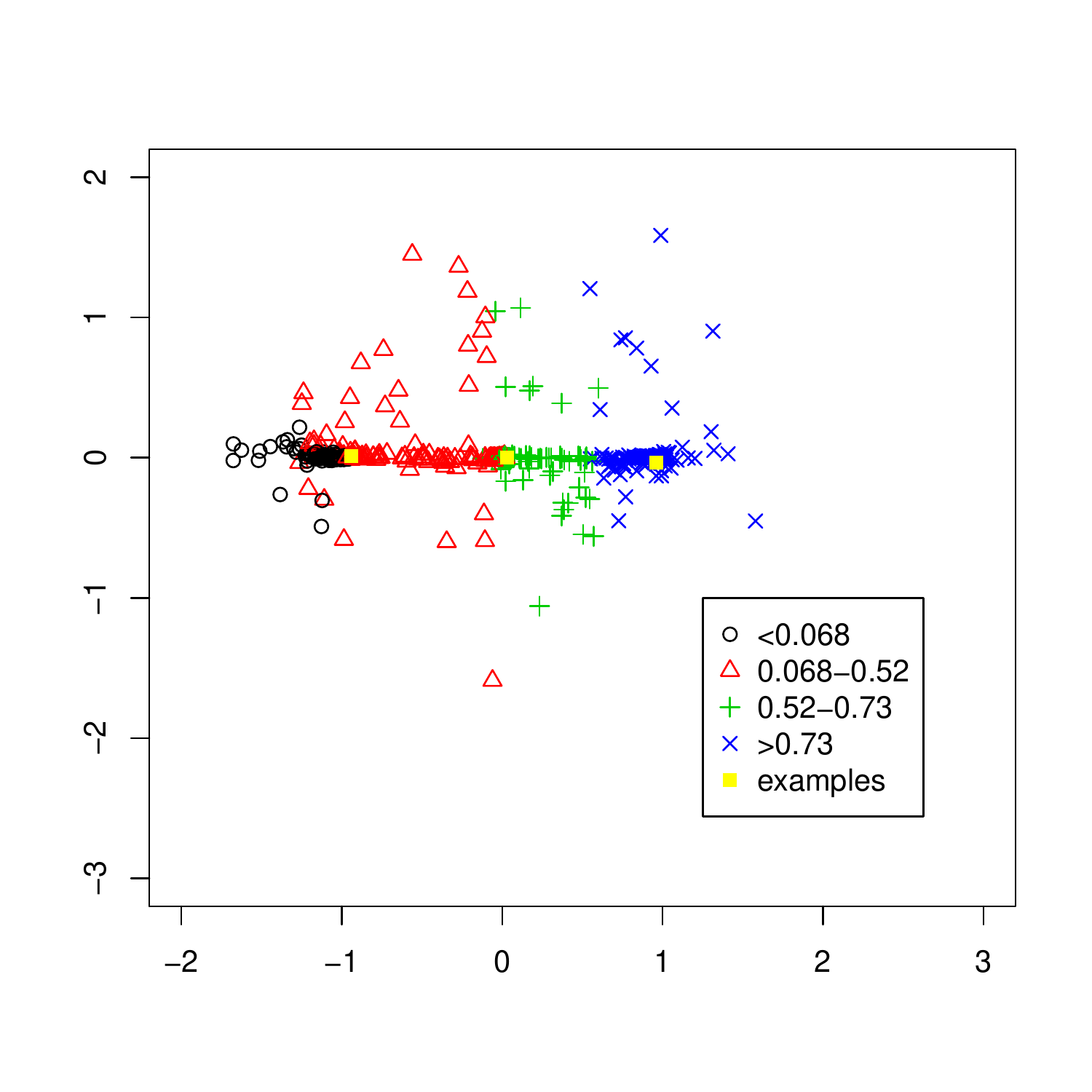}
        \caption{}
    \end{subfigure}
    \begin{subfigure}[b]{0.4\textwidth}
        \centering
        \includegraphics[scale = 0.4]{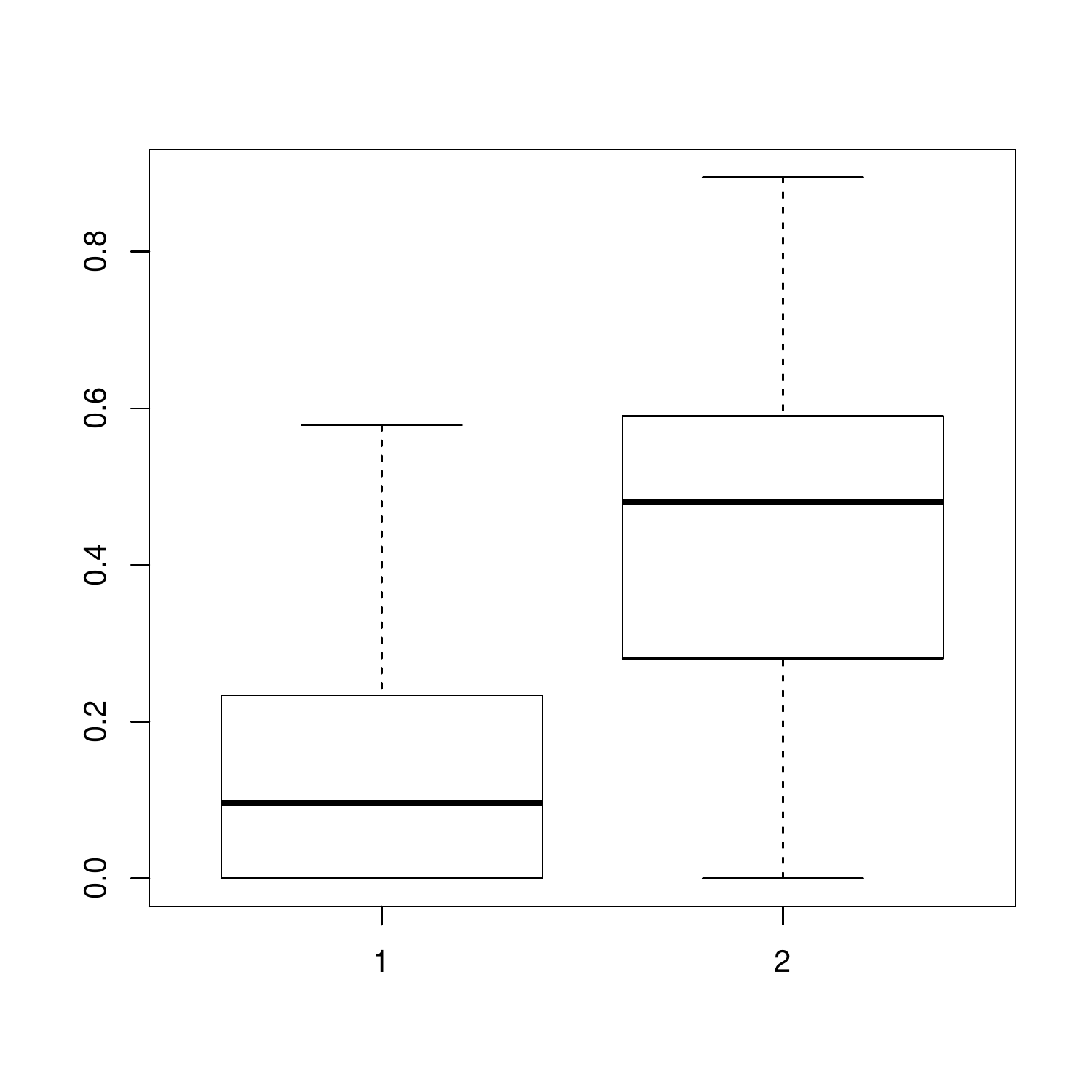}
        \caption{}
    \end{subfigure}
\caption{Analysis of senator roll call data. Panel (a): Visualization of roll call points, where roll calls are classified by the proportion of ``Yea" from Republicans. Specifically, roll calls who have the proportions less than 0.068, 0.068-0.52,0.52-0.73 and larger than 0.73 are indicated by black, red, green and blue points, respectively. The yellow solid points are example roll calls to be discussed. Panel (b): Box plots of $\min\{\mbox{CE}^{(1)}_j, \mbox{CE}^{(2)}_j\}$, for roll calls lying near the $x$-axis ($|\hat a_{j2} | \leq 0.05$) one the left and for those spreading out along the $y$-axis ($|\hat a_{j2} | > 0.05$) on the right.}
\label{fig:vote108_bill}
\end{figure}

\section{Concluding Remarks}\label{sec:disc}

In this paper, 
we provide a statistical framework for studying  unfolding-model-based visualization.
An estimator, together with an algorithm for its computation, is proposed, whose performance is examined by simulation studies. Under reasonable conditions, we
provide asymptotic results for the recovery of ideal-point configuration.
The proposed method is applied to two datasets, one on movie rating and the other on senator voting, for which interpretable results are obtained.

The ideal points obtained from the proposed method can be used in further analysis. For example, one can use the estimated person points as covariates in regression analysis. For another example, one may further conduct cluster analysis on the respondents and items, for example, by applying the K-means algorithm \citep{macqueen1967some}. In fact, as discussed in the supplementary material, there is a connection between our unfolding model and the stochastic co-blockmodel \citep{choi2014co,rohe2016co} for bi-cluster analysis. When data follow a stochastic co-blockmodel, then our consistency result for the unfolding model further guarantees the consistency of bi-cluster analysis.

The current analysis may be extended along multiple directions. First, the current analysis keeps the latent dimension $K$ fixed. In fact, the theoretical results established in this paper can be generalized to a setting where $K$ also diverges, a more appropriate setting for data of a very large scale.
Second, it is possible to make statistical inference about the person and item ideal points, such as testing
whether a person point is closer to one item point than another. Making statistical inference under our model is closely related to statistical inference for low-rank matrix completion \citep[see e.g.,][]{chen2019inference,xia2019statistical}, but
the non-linear link function in our model brings more challenges and thus methods and theory remain to be developed. Third, although we focus on binary data in this paper, the proposed modeling framework, theory and computational algorithm can be extended to other types of data, such as ratings and rankings.
Finally, it may also be of interest to extend the current framework to the modeling and analysis of large-scale preferential choice data with  informatively missing data entries.


%

\bigskip\bigskip

\appendix

\noindent
{\Large Appendix}

\section{Bi-Cluster Analysis}

The applications of multidimensional scaling, including multidimensional unfolding as a special case, are often followed by cluster analysis \citep[e.g.,][]{kruskal1978multidimensional,borg2005modern} for better understanding and interpretation of the data visualization. In our context, it is often of interest to cluster the respondents and the items, respectively. This task is known as bi-clustering or co-clustering \citep{hartigan1972direct,dhillon2001co}, which is often studied statistically under the stochastic co-blockmodel \citep{choi2014co,rohe2016co}, an extension of the widely used stochastic blockmodel \citep{holland1983stochastic}.

Following multidimensional unfolding, it is natural to bi-cluster the respondents and the items based on the estimated ideal points, using the Euclidian distance as a natural measure of dissimilarity. In particular, we
use the K-means algorithm \citep{macqueen1967some} to
cluster the respondents and the items into $k_1$ and $k_2$ clusters, respectively, for some pre-specified numbers of clusters $k_1$ and $k_2$. This two-step procedure for bi-cluster analysis is described in Algorithm~\ref{algo:cluster}. 

\begin{algorithm}[Two-step procedure for bi-cluster analysis]\label{algo:cluster}
~

\begin{itemize}
  \item[] \textbf{Step 1}: Apply Algorithm~\ref{algo:estimate} and obtain estimates $\{\hat \ttt_1,...,\hat \ttt_N,\hat \aaaa_1,...,\hat \aaaa_J\}.$
  \item[] \textbf{Step 2}: Perform the K-means algorithm to $\{\hat \ttt_1,...,\hat \ttt_N\}$ and $\{\hat \aaaa_1,...,\hat \aaaa_J\}$
  given $k_1$ and $k_2$ clusters, respectively.
 \item[] \textbf{Output}: The cluster membership of respondents $\hat \vartheta_i \in \{1, ..., k_1\}$ and cluster membership of items $\hat \upsilon_j \in \{1, ..., k_2\}\ (i=1, ..., N;\ j=  1, ..., J)$.
\end{itemize}

\end{algorithm}
We provide a connection between the multidimensional unfolding model studied in this paper and the
stochastic co-blockmodel.  Consider a special case under the multidimensional unfolding model, where there are finite possible locations for the respondent ideal points and also for the item ideal points, independent of $N$ and $J$.
We denote the possible locations for the respondent ideal points as $\{\bb_1^*,...,\bb_{k_1}^*\}$ and denote those for the item ideal points as $\{\cc_1^*,...,\cc_{k_2}^*\}$.
Under this setting, there exist $k_1$ respondent latent classes and $k_2$ item latent classes, regarding two respondents/items as from the same latent class when they have the same location.
We denote $\vartheta_i^* \in \{1, ..., k_1\}$ and $\upsilon_j^* \in \{1, ..., k_2\}$ the true latent class memberships of respondent $i$ and item $j$, respectively.
In this sense,
the model becomes a stochastic co-blockmodel, for which
the distribution of $Y_{ij}$ is only determined by the latent class memberships of respondent $i$ and item $j$ and $Y_{ij}$s are conditionally independent given all the latent memberships of the respondents and items.
In what follows, we show that the proportions of misclassified respondents and items converge to 0 in probability, when both $N$ and $J$ grow to infinity, if the $K$-means algorithm in Algorithm~\ref{algo:cluster} has converged to the global optima.

\begin{theorem}\label{thm:cluster}
Suppose A0, A3 and A4 are satisfied, and $K_+\geq K.$ Further suppose the multidimensional unfolding model degenerates to a stochastic co-blockmodel, satisfying
$\ttt_i^* \in \{\bb_1^*, ..., \bb_{k_1}^*\}$ and $\aaaa_j^* \in \{\cc_1^*, ..., \cc_{k_2}^*\}.$ If both $K$-means algorithms in Algorithm~\ref{algo:cluster} converge to the global optima, then the clustering result satisfies
\begin{equation}\label{eq:cluster}
\min\left\{\max_{\zeta \in \mathcal B_{k_1}}\frac{\sum_{i=1}^{N} 1_{\{\hat\vartheta_i = \zeta(\vartheta^*_i)\}}}{N}, \max_{\zeta \in \mathcal B_{k_2}}\frac{\sum_{j=1}^{J} 1_{\{\hat\upsilon_j = \zeta(\upsilon^*_j)\}}}{J}\right\}
\end{equation}
goes to 1 in probability as both $N$ and $J$ grow to infinity,
where $\mathcal B_{k}$ denotes the set of all permutations on $\{1, ..., k\}$, for $k = k_1, k_2$.
\end{theorem}

\begin{remark}
To handle ``label switching indeterminacy" in clustering, in the loss function \eqref{eq:cluster} we find permutations that best match the true latent class memberships and their estimates for both the respondents and the items.
\end{remark}

\section{Proof of Theoretical Results}
\subsection{Definitions and Notations}
In this appendix, we use $c, C, C_1, C_2$ to represent constants which do not depend on $N,J,$ the values of which may vary according to the context. With a little abuse of notation, we use $\AAA_{N,J}$ to denote the specified events, which may differ in different proofs.
For $\x\in \R^K,$ we use $B_{\bf \x}^{K}(C)$ to denote the closed ball in $\mathbb R^K$ centered at $\x$ with radius $C$.
Unless otherwise specified, all balls in the appendix is assumed to be closed.
For a set $G \subset \R^K,$ let $\text{int}(G)$ denote the set of all its interior points.
For a positive integer $n,$ we denote $[n] := \{1,...,n\}.$
We start with some notions which will be used in the proof of theorems, propositions and lemmas.

\begin{definition}
For points $\x_i, \x_i' \in \mathbb{R}^K, i = 1, ..., n$, we write $(\x_1,...,\x_{n})\sim (\x_1',...,\x_{n}')$, if there exists an isometry $F \in \AAA_K$, such that  $\x'_i = F(\x_i), i = 1,...,n$. 
\end{definition}

\begin{remark}
  It is easy to show that ``$\sim$" is an equivalence relation.
\end{remark}

\begin{definition}[Configuration]
We define an $n$-point configuration as an equivalence class. That is, we define a configuration
$$[\x_1,...,\x_n] := \{(\x_1', ..., \x_n'): (\x_1', ..., \x_n') \sim (\x_1, ..., \x_n)\}$$
as the equivalence class of $(\x_1, ..., \x_n)$.
\end{definition}

\begin{remark}
By the property of isometry mapping,  it is easy to see that all the elements in the same configuration have the same distance matrix.
\end{remark}

We now consider the space of all $n$-point configurations in $\mathbb R^{K}$, denoted by
$$\mathcal H_{n,K} := \left\{[\x_1,...,\x_n]: \x_i \in \mathbb R^{K}, i = 1, ..., n\right\}.$$

For two configurations $\tau_1 = [\x_1,...,\x_n], \tau_2 = [\y_1,...,\y_n] \in \mathcal H_{n,K}$, we define
\begin{equation*}
d(\tau_1,\tau_2) := \inf\limits_{F\in \mathcal A_K} \sqrt{\sum_{1\leq i \leq n} \|F(\x_i)-\y_i\|^2}.
\end{equation*}
First, we note that $d(\cdot, \cdot)$ is a well-defined mapping from $\mathcal H_{n,K} \times \mathcal H_{n,K}$ to $\mathbb R$. That is, for any $(\x_1', ..., \x_n') \in [\x_1, ..., \x_n]$ and $(\y_1', ..., \y_n') \in [\y_1, ..., \y_n]$,
$$\inf\limits_{F\in \mathcal A_K} \sqrt{\sum_{1\leq i \leq n} \|F(\x_i)-\y_i\|^2} = \inf\limits_{F\in \mathcal A_K} \sqrt{\sum_{1\leq i \leq n} \|F(\x_i')-\y_i'\|^2}.$$
Second, we notice that $d(\cdot, \cdot)$ is a metric on $\mathcal H_{n,K}$, as summarized in Lemma \ref{lem:metric} below.
\begin{lemma}\label{lem:metric}
$d(\cdot, \cdot)$ is a metric on $\mathcal H_{n,K}.$ 
\end{lemma}
\begin{remark}\label{rmk:higher dim}
For $[\x_1,...,\x_n] \in \HH_{n,K},$ we have $[(\x_1^\top,0)^\top,...,(\x^\top_n,0)^\top] \in \HH_{n,K+1}$ in which sense
we can say $\HH_{n,K} \subset \HH_{n,K+1}.$ Thus $\HH_{n,K_1} \subset \HH_{n,K_2}$ if $K_1 \leq K_2.$
For $\tau_1 = [\x_1,...,\x_n] \in \HH_{n,K_1}, \tau_2 = [\y_1,...,\y_n] \in \HH_{n, K_2},$ the $d(\tau_1, \tau_2)$ is defined in the same way by seeing both $\tau_1$ and $\tau_2$ as elements in $\HH_{n,\max\{K_1,K_2\}}.$ 
\end{remark}

We further denote $\mathcal P_{a,b,K}$
as the set of $a \times b$ partial distance matrices for configurations in $\R^K:$
$$\mathcal P_{a,b,K} :=\left\{(\Vert\x_i -\y_j \Vert^2)_{a\times b}:  [\x_1, ..., \x_a, \y_1, ..., \y_b] \in \mathcal H_{a+b,K}\right\}.$$
It is easy to check that $\PP_{a,b,K} \subset \PP_{a,b,K+1}.$

For $A_1,...,A_{n} \subset \mathbb{R}^K$, denote $[A_1,...,A_{n}]$ as a subset of $\mathcal H_{n,K}:$
$$[A_1,...,A_{n}] := \{[\x_1, ..., \x_n]: \x_i \in A_i, i = 1, ..., n\}.$$

For $A,B \subset \mathcal H_{n,K}$, the distance between $A$ and $B$ is defined as
\begin{equation}\label{eq:dist_set}
d(A,B) := \inf\limits_{\tau_1 \in A, \tau_2 \in B}d(\tau_1,\tau_2).
\end{equation}

We further denote
$$\mathcal H_{n,K,C} := \{[\x_1, ..., \x_n] \in \HH_{n,K}: \Vert \x_i\Vert  \leq C\}$$
as a compact subset of $\mathcal H_{n,K},$ and
\begin{equation}\label{eq:P abKC def}
\mathcal P_{a,b,K,C} :=\left\{(\Vert\x_i -\y_j \Vert^2)_{a\times b}:  [\x_1, ..., \x_a, \y_1, ..., \y_b] \in \mathcal H_{a+b,K,C}\right\}
\end{equation}
as a compact subset of $\PP_{a,b,K}.$
We consider a mapping defined as following:
\begin{equation*}
\begin{aligned}
\Phi_{a,b,K}:  \R^{(a+b)\times K} &\rightarrow \mathcal P_{a,b,K},\\
 (\x_1,....,\x_{a+b})^\top &\mapsto D,
\end{aligned}
\end{equation*}
where $D$ is the $a \times b$ partial distance matrix of $\{(\x_1,...,\x_{a}),(\x_{a+1},...,\x_{a+b})\}$.
It is not difficult to check that $\Phi_{a,b,K}$
is invariant with respect to isometry. Then, for $\tau = [\x_1,...,\x_{a+b}],$ we denote $$
\begin{aligned}
\Phi_{a,b,K}(\tau) &:= \Phi_{a,b,K}(X),
\end{aligned}
$$ where $X^\top = (\x_1,...,\x_{a+b}).$

Having introduced the notions above, we give the following lemma, which is crucial to the proof of Theorem \ref{thm:loss}. It
essentially shows that for any partial distance matrix $D' \in \PP_{k_1,k_2,K_+,M}$ that approximates to another partial distance matrix $D \in \PP_{k_1,k_2,K,M},$ whose configuration $\tau$ contains a collection of anchor points,
then any configuration $\tau'$ of $D'$ will also approximate to $\tau.$



\begin{lemma}\label{lem:anchor_continuous}
For compact subsets $\BBB_1,...,\BBB_{k_1},\CCC_1,...,\CCC_{k_2}\subset B_{\0}^K(M),$ let $$\BBB =[\BBB_1,...,\BBB_{k_1},\CCC_1,...,\CCC_{k_2}].$$ Suppose that for any $(\x_1,...,\x_{k_1+k_2}) \in \BBB_1 \times \cdots \times \BBB_{k_1} \times \CCC_1 \times \cdots \times \CCC_{k_2},$ $\{\x_1,...,\x_{k_1}\}$ and $\{\x_{k_1+1},...,\x_{k_1+k_2}\}$ are a collection of anchor points in $\R^K.$
Then, for any $\epsilon_c > 0,$ there exists $\epsilon_d >0$ such that for any
$\tau' \in \HH_{k_1+k_2,K_+,M}$ and $\tau \in \BBB$ satisfying $$
\|\Phi_{k_1,k_2,K_+}(\tau')-\Phi_{k_1,k_2,K_+}(\tau)\|_F < \epsilon_d,
$$ we have $$
d(\tau',\tau) < \epsilon_c.
$$
\end{lemma}

We end this section by the following lemma, which will also be used in the proof of Theorem \ref{thm:loss}.
\begin{lemma}\label{lem:another}
Suppose $\{\bb^*_1,...,\bb^*_{k_1}\}, \{\cc^*_1,...,\cc^*_{k_2}\} \subset B_{\bf 0}^K(C)$ are a collection of anchor points in $\R^K$. Then, for any $\x \in B_{\bf 0}^K(C),$ the $\{\x,\bb^*_1,...,\bb^*_{k_1}\}, \{\cc^*_1,...,\cc^*_{k_2}\}$ are also a collection of anchor points in $\R^K.$
\end{lemma}

\subsection{Proof of Theorems}
\begin{proof}[Proof of Theorem \ref{thm:loss}]
We first show the proof of \eqref{eq:theorem}.
For $\epsilon$ which is given in condition A3, there exist constant 
$p_{\epsilon} \in (0, 1)$, and balls of radius $\epsilon$ in $\mathbb R^{K}$, denoted by $\tilde B_1(\epsilon)$, ..., $\tilde B_{k_1}(\epsilon)$, $\tilde G_{1}(\epsilon), ..., \tilde G_{k_2}(\epsilon)$, such that for $N,J$ large enough, 
\begin{equation*}
  \begin{aligned}
  & \frac{\sum_{l=1}^N 1_{\{\ttt_l^*\in B_{\mathbf b_i^*}(\epsilon), \tilde\ttt_l \in \tilde B_i(\epsilon)\}}}{N} > p_{\epsilon}, i = 1, ..., k_1,\\
  & \frac{\sum_{l=1}^J 1_{\{\mathbf a_l^*\in B_{\mathbf c_i^*}(\epsilon), \tilde{\mathbf a}_l \in \tilde G_i(\epsilon)\}}}{J} > p_{\epsilon}, i = 1, ..., k_2.\\
  \end{aligned}
\end{equation*}
This comes straightforwardly from condition A0 and requirement (2) of anchor points in condition A3. Note that the centers of $\tilde B_{k}(\epsilon)$ and $\tilde G_{l}(\epsilon)$ may vary through $N,J$. We also use $B_k^*(\epsilon)$ and $G^*_{l}(\epsilon)$ to denote $B_{\bb^*_k}(\epsilon)$ and $B_{\cc^*_l}(\epsilon),$ respectively.

We first focus on the set of person points $$I_1(\epsilon) := \bigcup_{k=1}^{k_1}\{i \in [N]: \ttt_i^*\in B^*_{k}(\epsilon), \tilde{\ttt}_i \in \tilde B_k(\epsilon)\}$$ and the set of item points
$$I_2(\epsilon) := \bigcup_{l=1}^{k_2}\{j \in [J]: \mathbf a_j^*\in G_{l}^*(\epsilon), \tilde{\mathbf a}_j \in \tilde G_l(\epsilon)\}.$$

Let $\ttt_i^+ = \left((\ttt_i^*)^\top, \0^\top\right)^\top, \aaaa_j^+ = \left( (\aaaa_j^*)^\top, \0^\top \right) \in \R^{K_+}.$
We will show that there exists an isometry mapping $F_{N, J} \in \mathcal A_{K_+}$, under which $F_{N, J}(\tilde \ttt_i) \approx \ttt_i^+$ and
$F_{N, J}(\tilde \aaaa_j) \approx \aaaa_j^+$, for all $i\in I_1(\epsilon)$ and $j \in I_2(\epsilon)$. This is formalized in the following lemma.

\begin{lemma}\label{lem:part1+}
For $N,J$ large enough,
there exists an isometry $F_{N, J} \in \mathcal A_{K_+}$, such that
$$\Vert F_{N, J}(\mathbf x)\Vert \leq 4M, \mbox{~for all~}  \mathbf x \in B_{\0}^{K_+}(M),
$$
and for all $i \in I_1(\epsilon)$ and for all $j \in I_2(\epsilon)$,
$$\Vert F_{N, J}(\tilde \ttt_i) - \ttt_i^+ \Vert \leq 5\epsilon,$$
and
$$\Vert F_{N, J}(\tilde \aaaa_j) - \aaaa_j^+ \Vert \leq 5\epsilon.$$
\end{lemma}


We then show that for most of the person points $i \notin I_1(\epsilon)$ and for most of the item points $j \notin I_2(\epsilon)$, we still have $F_{N, J}(\tilde \ttt_i) \approx \ttt^+_i$ and
$F_{N, J}(\tilde \aaaa_j) \approx \aaaa^+_j$, under the same isometry mapping $F_{N, J}$ as in Lemma~\ref{lem:part1+}. This is formalized in Lemma~\ref{lem:part2+} below.

\begin{lemma}\label{lem:part2+}
For $N,J$ large enough,
there exists a constant $\kappa > 0,$ such that
for the isometry mapping $F_{N, J}$ defined in Lemma \ref{lem:part1+}, the proportions
$$\lambda_{1, N, J} = \frac{\sum_{i=1}^N 1_{\{\Vert F_{N, J}(\tilde\ttt_i) - \ttt^+_i \Vert > \kappa\epsilon\}}}{N}$$
and
$$\lambda_{2, N, J} = \frac{\sum_{j=1}^J 1_{\{\Vert  F_{N, J}(\tilde \aaaa_j) - \aaaa_j^+ \Vert > \kappa\epsilon \}}}{J}$$
satisfy
\begin{equation}\label{eq:lambda}
\lambda_{k, N, J} \to 0,
\end{equation}
for $k=1,2,$ as $N,J$ grow to infinity.
\end{lemma}

Since by Lemma~\ref{lem:part1+}, we have $F_{N, J}$ maps $B^{K_+}_{\0}(M)$ to $B_{\0}^{K_+}(4M)$, then for all $\tilde{\ttt}_i$ and for all $\tilde{\aaaa}_j$,
$$\Vert F_{N, J}(\tilde\ttt_i) - \ttt_i^+ \Vert \leq 5M$$ and $$\Vert F_{N, J}(\tilde \aaaa_j) - \aaaa_j^+ \Vert \leq 5M.$$
Combining this with Lemma~  
\ref{lem:part2+}, we have
\begin{equation}
\begin{aligned}
& \min_{F\in \mathcal A_{K_+}} \left\{ \frac{\sum_{i=1}^N \Vert F(\tilde\ttt_i) - \ttt_i^+\Vert^2}{N} + \frac{\sum_{j=1}^J \Vert F(\tilde\aaaa_j) - \aaaa_j^+\Vert^2}{J}\right\} \\
\leq &\frac{\sum_{i=1}^N \Vert F_{N,J}(\tilde\ttt_i) - \ttt_i^+ \Vert^2}{N} + \frac{\sum_{j=1}^J \Vert F_{N,J}(\tilde\aaaa_j) - \aaaa_j^+ \Vert^2}{J}\\
\leq & \left(25(M)^2\lambda_{1, N, J} + \kappa^2\epsilon^2\right) + \left(25(M)^2\lambda_{2, N, J} + \kappa^2\epsilon^2\right) \\
\leq & 25(M)^2(\lambda_{1,N,J}+\lambda_{2,N,J}) + 2\kappa^2\epsilon^2
\end{aligned}
\end{equation}
By \eqref{eq:lambda}, \eqref{eq:theorem} holds.
\eqref{eq:theorem_2} holds if $\epsilon$ can be arbitrarily small.
We complete the proof.
\end{proof}

\begin{proof}[Proof of Theorem \ref{thm:main}]
Combining Theorem \ref{thm:loss} and Proposition \ref{prop:distance}, we have the result.
\end{proof}

\begin{proof}[Proof of Theorem \ref{thm:random}]
Theorem \ref{thm:random} is a special case of Proposition \ref{prop:randommissing}. See the proof of Proposition \ref{prop:randommissing}.
\end{proof}

\begin{proof}[Proof of Theorem \ref{thm:cluster}]
For simplicity of writing, we suppose $K_+ = K$ in this proof.
We only prove the result for the respondents.
The proof for the items is the same.
Under the conditions of Theorem~\ref{thm:cluster}, the result of Theorem~\ref{thm:main} is satisfied and with a slight change in the proof, we can get
$$
\max_{F \in \mathcal A_K}\frac{\sum_{i=1}^N \Vert \hat \ttt_i -  F(\mathbf b_{\vartheta_i^*}^*)\Vert}{N}^2 = o_p(1).
$$
Consequently, there exists isometry $F_{N,J}^*$, such that
\begin{equation}\label{eq:person}
\frac{\sum_{i=1}^N \Vert \hat \ttt_i -  F_{N, J}^* (\mathbf b_{\vartheta_i^*}^*)\Vert^2}{N} = o_p(1),
\end{equation}
noting that $\bb_{\vartheta_i^*}^* = \ttt_i^*$.

\begin{lemma}\label{lem:kmeans}
Under the same conditions as Theorem \ref{thm:cluster},
suppose that
$$
\frac{\sum_{i=1}^N \Vert \hat \ttt_i -  F_{N, J}^* (\mathbf b_{\vartheta_i^*}^*)\Vert^2}{N} = o_p(1).
$$
Then we have
$$
\max_{\zeta \in \mathcal B_{k_1}}\frac{\sum_{i=1}^{N} 1_{\{\vartheta_i^* = \zeta(\hat \vartheta_i)\}}}{N} = o_p(1). 
$$
\end{lemma}
With Lemma \ref{lem:kmeans}, we complete the proof for the respondents.
\end{proof}

\subsection{Proof of Propositions}

\begin{proof}[Proof of Proposition \ref{prop:configuration}]
It suffices to prove in the case when $\sum_{i=1}^n \x_i = \0$ and $\sum_{i=1}^n \y_i = \0.$ Denote $D = (d_{ij})_{n\times n},$ where $d_{ij} = \|\x_i -\x_j \|^2 = \|\y_i - \y_j\|^2$ and let $B = (b_{ij})_{n\times n} = -\frac{1}{2}JDJ,$ where $J = I_n - 1_n 1_n^\top / n.$ Then $B$ is inner product matrix of both $\{\x_1,...,\x_n\}$ and $\{\y_1,...,\y_n\}.$ That is,
$b_{ij} = \x_i^\top \x_j = \y_i \y_j^\top,$ for $1\leq i, j \leq n.$ We refer readers to \cite{critchley1988certain} for the relation between inner product matrix and distance matrix. So if we denote $$P_1 = (\x_1,...,\x_n)^\top, \quad P_2 = (\y_1,...,\y_n)^\top,$$ then we have $$P_1P_1^\top = P_2P_2^\top = B.$$ Let $$P_1^\top = Q_1R_1, \quad P_2^\top = Q_2R_2$$ be the QR decomposition (see \cite{cheney2009linear}) of $P_1, P_2,$ where $Q_1, Q_2$ are $k\times k$ orthogonal matrix and $R_1, R_2$ are $k\times n$ upper-triangular matrix with non-negative diagonal entries. Since $\x_i^\top\x_j = \y_i^\top\y_j,$ for $1\leq i,j \leq n,$ it is not difficult to check that $R_1 = R_2.$ If we define $O = Q_2Q_1^\top,$ then $$OP_1^\top = OQ_1R_1 = Q_2Q_1^\top Q_1R_1 = Q_2R_1 = Q_2R_2 = P_2^\top,$$ which means $O\x_i = \y_i,$ for $1\leq i\leq n.$ We complete the proof.
\end{proof}

\begin{proof}[Proof of Proposition \ref{prop:random}]
We first introduce a lemma as following.
\begin{lemma}\label{lem:anchor_exist_new}
There exists a collection of anchor points $\{\bb^*_1,...,\bb^*_{k_2}\}, \{\cc^*_1,...,\cc^*_{k_2}\} \subset \text{int}(G),$ where $G$ is the ball defined in Proposition \ref{prop:random}.
\end{lemma}

We fix such collection of anchor points. For any $\epsilon >0,$ we denote $B_k^*(\epsilon), G_l^*(\epsilon),$ for $1\leq k\leq k_1$ and $1\leq l\leq k_2,$
as balls centered at $\bb_k^*$ and $\cc^*_l,$ respectively. For sufficiently small $\epsilon >0,$ it is easy to see that for any $$\bb_1 \in B_1^*(\epsilon),...,\bb_{k_1} \in B_{k_1}^*(\epsilon),\cc_1 \in G_1^*(\epsilon),...,\cc_{k_2} \in G_{k_2}^*(\epsilon),$$ the $\{\bb_1,...,\bb_{k_1}\}, \{\cc_1,...,\cc_{k_2}\}$ are a collection of anchor points in $\R^K.$
Therefore, the (1) of A3 holds.
We define $$\beta_\epsilon := \frac{1}{2}\min_{\substack{1\leq k \leq k_1\\ 1\leq l \leq k_2}} \left\{ P_1B_k^*(\epsilon), P_2G_l^*(\epsilon) \right\}$$
and use $\AAA_{N,J}$ to denote the following event
\begin{equation}\label{eq:hoeffding2}
\begin{aligned}
&\left| \frac{1}{N} \sum_{i=1}^N 1_{\{\ttt_i^* \in B^*_k(\epsilon)\}} - P_1B_k^*(\epsilon) \right| \leq \beta_\epsilon, \quad k=1,...,k_1,\\
&\left| \frac{1}{J} \sum_{j=1}^J 1_{\{\aaaa_j^* \in G^*_l(\epsilon)\}} - P_2G_l^*(\epsilon) \right| \leq \beta_\epsilon, \quad l=1,...,k_2,
\end{aligned}
\end{equation}
where $P_1B_k^*(\epsilon), P_2G_l^*(\epsilon)$ represent the probability measure of $B_k^*(\epsilon), G^*_l(\epsilon)$ with respect to $P_1$ and  $P_2,$ respectively.
By Hoeffding's inequality, we have
\begin{equation}\label{eq:prob2}
\Pr\left(\eqref{eq:hoeffding2} \text{ holds }\right) \geq 1-2k_1\exp(-\frac{1}{2}N\beta_\epsilon^2) - 2k_2\exp(-\frac{1}{2}J\beta_\epsilon^2).
\end{equation}
So we have $$\Pr(\AAA_{N,J}) \to 1$$ as $N,J$ grow.
On $\AAA_{N,J},$ we have
\begin{equation}\label{eq:number anchor2}
\begin{aligned}
&\frac{1}{N} \sum_{i=1}^N 1_{\{\ttt_i^* \in B_k^*(\epsilon)\}}  \geq \beta_\epsilon, \quad 1\leq k\leq k_1,\\
&\frac{1}{J} \sum_{j=1}^J 1_{\{\aaaa_j^* \in G_l^*(\epsilon)\}} \geq \beta_\epsilon, \quad 1 \leq l \leq k_2.
\end{aligned}
\end{equation}
On $\AAA_{N,J},$ \eqref{eq:number anchor2} holds. Then, the (2) of A3 holds almost surely.
\end{proof}

\begin{proof}[Proof of Proposition \ref{prop:distance}]
Proposition \ref{prop:distance} is a special case of Proposition \ref{prop:distancemissing}. See the proof of Proposition \ref{prop:distancemissing}.
\end{proof}

\begin{proof}[Proof of Proposition \ref{prop:distancemissing}]
The proof of Proposition \ref{prop:distancemissing} is similar to Theorem 1 of \cite{davenport20141}. We only state the main steps. 

We denote $D$ as the partial distance matrix of $(\ttt_1,...,\ttt_N)$ and $(\aaaa_1,...,\aaaa_J)$ (to simplify the notation, we ignore the subscripts $N$ and $J$ for $D$). Since the likelihood function depends on $(\ttt_1,...,\ttt_N)$ and $(\aaaa_1,...,\aaaa_J)$ only through their partial distance matrix, we re-parameterize the likelihood function by $D$. We denote
$$l_{\Omega,Y}(D) = \log L^{\Omega}(\ttt_1,...,\ttt_N,\aaaa_1,...,\aaaa_J),$$
where the subscripts $\Omega = (\omega_{ij})_{N\times J}$ and $Y = (Y_{ij})_{N\times J}$ indicate the random variables in the likelihood function and $D$ contains the parameters.

Let
\begin{equation}
\bar{l}_{\Omega, Y}(D) = l_{\Omega, Y}(D) - l_{\Omega, Y}(\bf{0}),
\end{equation}
where $\bf{0}$ represents an $N\times J$ matrix whose elements are all 0 and let
\begin{equation}
G = \left\{D \in \mathbb R^{N \times J}: \|D\|_{*} \leq 4M^2 \sqrt{(K_++2)NJ}  \right\}.
\end{equation}
\begin{lemma}
Under the same conditions as Proposition \ref{prop:distancemissing}, there exist constant $C_1$ and $C_2$ such that
\begin{equation*}
\begin{aligned}
&\Pr\left(    \sup \limits_{D \in G} | \bar{l}_{\Omega, Y}(D) - E\bar{l}_{\Omega, Y}(D)  | \geq 4 M^2 C_1  L_{4M^2} \sqrt{K_++2} \sqrt{n(N+J)+NJ\log(NJ)}   \right) \\
\leq &\frac{C_2}{N+J}.
\end{aligned}
\end{equation*}
\label{lem:concentration}
\end{lemma}

Let
$H = \{   D: d_{ij}=\|\ttt_i-\aaaa_j\|^2, \text{ where } \|\ttt_i\|, \|\aaaa_j\| \leq M, i = 1,..., N, j = 1, ..., J \}.$
It is easy to check that  $H \subset G$. Consequently,
\begin{align*}
 & \ \ \Pr\left(    \sup \limits_{D \in H} | \bar{l}_{\Omega, Y}(D) - E\bar{l}_{\Omega, Y}(D)  | \geq 4C_1M^2L_{4M^2} \sqrt{K_++2}  \sqrt{n(N+J)+NJ\log(NJ)}   \right)\\
\leq & \ \ \Pr\left(    \sup \limits_{D \in G} | \bar{l}_{\Omega, Y}(D) - E\bar{l}_{\Omega, Y}(D)  | \geq 4C_1M^2L_{4M^2} \sqrt{K_++2}  \sqrt{n(N+J)+NJ\log(NJ)}   \right)\\
\leq & \ \ \frac{C_2}{N+J}.
\end{align*}

Given the above development, Proposition \ref{prop:distancemissing} is implied by the following lemma.

\begin{lemma} \label{lem:bddist}
Under the same conditions as Proposition \ref{prop:distancemissing},
$$\frac{1}{NJ} \|D^*_{N,J}-\hat{D}_{N,J}\|_F^2 \leq \frac{16}{n} \beta_{4M^2} \sup \limits_{D \in H} \vert \bar{l}_{\Omega,Y}(D)- E\bar{l}_{\Omega,Y}(D) \vert.$$
\end{lemma}

Therefore, with probability at least $1-{C_2}/{(N+J)},$
\begin{equation*}
\frac{1}{NJ} \|D^*_{N,J}-\hat{D}_{N,J}\|_F^2 \leq   64C_1M^2L_{4M^2}\beta_{4M^2} \sqrt{K_++2} \sqrt{\frac{N+J}{n}}\sqrt{1+\frac{NJ\log(NJ)}{n(N+J)}}.
\end{equation*}
We complete the proof by absorbing $64\sqrt{K_++2}$ into $C_1$.
\end{proof}

\begin{proof}[Proof of Proposition \ref{prop:randommissing}]
We use $\AAA_{N,J}$ to denote the event that the result in Proposition \ref{prop:distancemissing} holds. By Theorem \ref{thm:loss} and Proposition \ref{prop:random}, on $\AAA_{N,J},$ we have $$
\min_{F\in \mathcal A_{K_+}}\frac{\sum_{i=1}^N \Vert \ttt_i^+ - F(\hat \ttt_i^{\Omega})\Vert^2}{N} + \frac{\sum_{j=1}^J \Vert \aaaa_j^+ - F(\hat \aaaa_j^{\Omega})\Vert^2}{J}
$$ goes to $0,$ as $N,J$ grow to infinity. Since $\Pr(\AAA_{N,J}) \to 0,$ we complete the proof.
\end{proof}

\subsection{Proof of Lemmas}

\begin{proof}[Proof of Lemma \ref{lem:metric}]
Let $\tau_1 = [\x_1,...,\x_n], \tau_2 = [\y_1,...,\y_n], \tau_3 = [\z_1,...,\z_n]$.
Define $$\tilde d(\tau_1,\tau_2) :=  \min_{F\in \AAA_K}\max_i\|F(\y_i) - \x_i \|$$ and it is easy to check that $$\tilde d(\tau_1,\tau_2) \leq d(\tau_1,\tau_2) \leq \sqrt{n}\tilde d(\tau_1,\tau_2).$$ So we just need to verify that function $\tilde d(\cdot, \cdot)$ satisfies the triangle inequality.
Let isometries $F_{21}, F_{31}$ satisfy
\begin{align*}
&\tilde d(\tau_1,\tau_2) = \max_i\|F_{21}(\y_i)-\x_i\| = \|F_{21}(\y_l)-\x_l\|, \\
&\tilde d(\tau_1,\tau_3) = \max_i\|F_{31}(\z_i)-\x_i\| = \|F_{31}(\z_m)-\x_m\|. \\
\end{align*}
Then
\begin{align*}
\tilde d(\tau_2,\tau_3) &\leq \max_{i}\left\{\|F_{31}(\z_i)-F_{21}(\y_i)\|\right\} \\
& \leq\max_{i}\{\|F_{31}(\z_i)-\x_i\| + \|F_{21}(\y_i) - \x_i\|\} \\
& \leq \|F_{21}(\y_l)-\x_l\| + \|F_{31}(\z_m) - \x_m\| \\
& = \tilde d(\tau_1,\tau_2) + \tilde d(\tau_1,\tau_3).
\end{align*}
We complete the proof.
\label{proof:metric}
\end{proof}

\begin{proof}[Proof of Lemma \ref{lem:anchor_continuous}]
Otherwise there exist $\epsilon_0>0$ and sequences $\{\tau_1^{(n)}\}_{n=1}^{\infty} \subset \mathcal \HH_{k_1+k_2,K_+,M}$, and $\{\tau_2^{(n)}\}_{n=1}^{\infty} \subset \BBB$ such that $$\left\|\Phi_{k_1,k_2,K_+}(\tau_1^{(n)})-\Phi_{k_1,k_2,K_+}(\tau_2^{(n)})\right\|_F<\frac{1}{n}$$
and $$d(\tau_1^{(n)},\tau_2^{(n)})>\epsilon_0.$$
Since both $\mathcal H_{k_1+k_2,K_+,M}$ and $\BBB$ are compact, there exists a subsequence $\{n_k\}_{k=1}^{\infty} \subset \mathbb N^+$, such that $\lim_{k\rightarrow \infty}\tau_1^{(n_k)} = \tilde\tau \in \mathcal H_{k_1+k_2,K_+,M}$ and $\lim_{k\rightarrow \infty}\tau_2^{(n_k)} = \tau_0 \in \BBB$. The two configurations $\tilde\tau$ and $\tau_0$ have the same partial distance matrix but $d(\tilde\tau,\tau_0)>\epsilon_0$. This makes a contradiction because $\tau_0 \in \BBB$ is the only configuration of its partial distance matrix, by the requirement of $\BBB.$
\end{proof}

\begin{proof}[Proof of Lemma \ref{lem:another}]
For a collection of points $\{\x,\bb^*_1,...,\bb^*_{k_1}\}, \{\cc^*_1,...,\cc^*_{k_2}\},$ it is not difficult to verify that condition A2 holds. So we only need to verify A1.

To verify A1, it suffices to show that if $\{\bb_1,...,\bb_{k_1}\},\{\cc_1,...,\cc_{k_2}\}$ is a collection of anchor points, then for any $\x \in B^K_{\bf 0 }(C),$
 $[\x,\bb_1,...,\bb_{k_1},\cc_1,...,\cc_{k_2}]$ is the unique configuration corresponding to its $(k_1+1)\times k_2$ partial distance matrix.

Suppose that $
\tau = [\x,\bb_1,...,\bb_{k_1},\cc_1,...\cc_{k_2}]$ and $\tau' = [\x',\bb'_1,...,\bb'_{k_1},\cc'_1,...\cc'_{k_2}]$ satisfy $$\Phi_{k_1+1,k_2,K}(\tau) = \Phi_{k_1+1,k_2,K}(\tau').$$
Then $$
\Phi_{k_1,k_2,K}([\bb_1,...,\bb_{k_2},\cc_1,...,\cc_{k_2}]) = \Phi_{k_1,k_2,K}([\bb'_1,...,\bb'_{k_1},\cc'_1,...\cc'_{k_2}]).
$$ Since $\{\bb_1,...,\bb_{k_1}\}, \{\cc_1,...\cc_{k_2}\}$ are a collection of anchor points, then $[\bb_1,...,\bb_{k_1},\cc_1,...\cc_{k_2}] = [\bb'_1,...,\bb'_{k_1},\cc'_1,...\cc'_{k_2}].$ Without loss of generality, we suppose $\bb_l = \bb_l'$ and $\cc_m = \cc'_m.$ Then,
the two configurations, $[\x,\cc_1,...,\cc_{k_2}]$ and $[\x,\cc'_1,...,\cc'_{k_2}],$ have the same complete distance matrix, which further leads that $$
[\x, \cc_1,...,\cc_{k_2}] = [\x',\cc_1,...,\cc_{k_2}].
$$ Since $\cc_1,...,\cc_{k_2}$ can affine span $\R^K,$ it is not difficult to see that $\x = \x'.$
Then, we get $\tau = \tau',$ and A1 has been verified.
\end{proof}

\begin{proof}[Proof of Lemma \ref{lem:part1+}]
We define
\begin{equation*}
\begin{aligned}
S^*_{N,J}(\epsilon) &=
\left[B_{1}^*(\epsilon),...,B_{k_1}^*(\epsilon),G_{1}^*(\epsilon),...,G_{k_2}^*(\epsilon)\right] \subset \HH_{k_1+k_2,K,M},\\
\tilde S_{N,J}(\epsilon) &=
\left[\tilde B_1(\epsilon), ...,\tilde B_{k_1}(\epsilon), \tilde G_{1}(\epsilon), ...,\tilde G_{k_2}(\epsilon)\right] \subset \HH_{k_1+k_2,K_+, M},
\end{aligned}
\end{equation*}
where $B^*_k(\epsilon),\tilde B_k(\epsilon), G^*_l(\epsilon), \tilde G_l(\epsilon)$ are defined in the proof of Theorem \ref{thm:loss}.
Let
\begin{equation}\label{eq:sigma definition}
\sigma_{N,J} := d(\tilde S_{N,J}(\epsilon),S^*_{N,J}(\epsilon))
\end{equation}
By \eqref{eq:dist_set} and triangle inequality, there exists an iosmetry $F_{N,J} \in \AAA_{K_+},$ such that for all $\x^*_k \in B_k^*(\epsilon), \y_{l}^* \in G_l^*(\epsilon), \tilde \x_k \in \tilde B_k(\epsilon), \tilde \y_l \in \tilde G_l(\epsilon),$
\begin{equation}\label{eq:triangle_bound}
\begin{aligned}
&\|F_{N,J}(\tilde \x_k) - \x_k^+\| \leq 4\epsilon + \sigma_{N,J}, \quad 1\leq k \leq k_1,\\
&\|F_{N,J}(\tilde \y_l) - \y_l^+\| \leq 4\epsilon + \sigma_{N,J}, \quad 1 \leq l \leq k_2.\\
\end{aligned}
\end{equation}

In what follows, we will show that $\sigma_{N,J}\leq \epsilon$ for $N,J$ large enough.
We first define
\begin{align}\label{eq:gamma definition}
\gamma_{N,J} = \inf\{    \|\Phi_{k_1,k_2,K_+}(\tilde \tau) - \Phi_{k_1,k_2, K_+}(\tau^*)\|_F: \tilde \tau \in \tilde S_{N,J}(\epsilon),
\tau^* \in S^*_{N,J}(\epsilon) \}
\end{align}
and 
we have
\begin{equation*}
\gamma_{N,J}^2(p_\epsilon N)(p_\epsilon J) \leq \| \tilde D_{N,J} - D_{N,J}^* \|_F^2 = o(NJ),\end{equation*}
which leads to
\begin{equation}\label{eq:rate_gamma}
\gamma_{N,J} = o(1). 
\end{equation}
By \eqref{eq:sigma definition}, there exist $\tilde \tau \in \tilde S_{N,J}(\epsilon)$ and $\tau^* \in S^*_{N,J}(\epsilon)$ such that
$$  \|\Phi_{k_1,k_2,K_+}(\tilde \tau) - \Phi_{k_1,k_2, K_+}(\tau^*)\|_F \leq 2\gamma_{N,J}.$$ Then by \eqref{eq:sigma definition}, we have
\begin{equation}\label{eq:sigma ineq}
\sigma_{N,J} = d(\tilde S_{N,J}(\epsilon),S^*_{N,J}(\epsilon))\leq d(\tau^*, \tilde\tau).
\end{equation}

As shown in the beginning of proof for Theorem \ref{thm:loss} and according to Definition \ref{def:anchor}, the $\tau^*$ is the unique configuration corresponding to its $k_1 \times k_2$ partial distance matrix. Since $\tilde \tau \in \tilde S_{N,J}(\epsilon) \subset \HH_{k_1+k_2,K_+,M},$ by Lemma \ref{lem:anchor_continuous}, we know $d(\tau^*,\tilde \tau) \to 0$ as $N,J$ grow to infinity, and thus
\begin{equation}\label{eq:sigma bound}
d(\tau^*,\tilde \tau) < \epsilon
\end{equation}
for $N,J$ large enough.

Finally, since $$B_k^*(\epsilon), G_l^*(\epsilon) \subset B^{K}_{\bf 0}(M), \quad \tilde B_k(\epsilon),\tilde G_l(\epsilon) \subset B^{K_+}_{\bf 0}(M),$$ we have, for $N,J$ large enough,
$$\Vert F_{N, J}(\mathbf x)\Vert \leq 4M, \mbox{~for~}  \mathbf x \in B_{\0}^{K_+}(M).$$ To see this, if there exists $\x \in B_{\0}^{K_+}(M)$ such that $\|F_{N,J}(\x)\| > 4M,$ then by simple geometry, $$\min_{\x \in B_{\0}^{K_+}(M)}\| F_{N,J}(\x) - \x \| > M.$$
According to \eqref{eq:triangle_bound} and \eqref{eq:sigma bound}, we will get $$M < \|F_{N,J}(\tilde \x_k) - \x_k^+\| \leq 4\epsilon + \sigma_{N,J} \leq 5\epsilon,$$ which contradicts with the fact that $\epsilon< \frac{1}{10}M \leq \frac{1}{10}M.$

\end{proof}

\begin{proof}[Proof of Lemma \ref{lem:part2+}]
Let $\tilde\cc_1,...,\tilde\cc_{k_2}$ denote the centers of $\tilde G_{1}(\epsilon),....,\tilde G_{k_2}(\epsilon)$ and $\tilde \cc_l^+ = (\tilde\cc_l^\top, \0^\top)^\top \in \R^{K^+}.$
We first give the following lemma. 

\begin{lemma}\label{lem:add a point}
For any $$
\begin{aligned}
\tau_1 &= [\x,\x_1,...,\x_{k_2}] \in [B_{\0}^K(M),G_1^*(\epsilon),...,G_{k_2}^*(\epsilon)],\\
\tau_2 &= [\y,\y_1,...,\y_{k_2}] \in [B_{\0}^{K_+}(M),B_{\tilde \cc_1^+}(\epsilon),...,B_{\tilde \cc_{k_2}^+}(\epsilon)],\\
\end{aligned}
$$ we have $$
\|\x^+-\y\| \leq c\max\left\{ d(\tau_1,\tau_2), \sqrt{\sum_{l=1}^{k_2} \| \x_l^+-\y_l \|^2} \right\},
$$ for a constant $c,$ which only depends on the set $\{\cc_1^*,...,\cc^*_{k_2}\}$ and $M.$
\end{lemma}


Define
\begin{equation}\label{eq:H1 definition}
H_1(\epsilon) := \{i\in [N]: \|F_{N,J}(\tilde \ttt_i) - \ttt_i^+\| > 5\max(c,1)\sqrt{k_1+k_2}\epsilon\}
\end{equation}
and
\begin{equation}\label{eq:H2 definition}
H_2(\epsilon) := \{j\in [J]: \|F_{N,J}(\tilde \aaaa_j) - \aaaa_j^+\| > 5\max(c,1)\sqrt{k_1+k_2}\epsilon\},
\end{equation}
where $c$ is the constant in Lemma \ref{lem:add a point}.
We set the constant $\kappa$ in Lemma \ref{lem:part2+} to be $5\max(c,1)\sqrt{k_1+k_2}\epsilon.$
and then we have $|H_1(\epsilon )| = N\lambda_{1,N,J}, |H_2(\epsilon )| = J\lambda_{2,N,J}.$
Note that $I_1(\epsilon ) \cap H_1(\epsilon ) = \emptyset, I_2(\epsilon ) \cap H_2(\epsilon ) = \emptyset$ for $N,J$ large.

We choose $i_1,...,i_{k_1} \in I_1(\epsilon)$ and $j_1,...,j_{k_2} \in I_2(\epsilon)$ such that
\begin{equation*}
\begin{aligned}
\ttt^*_{i_k} \in B_k^*(\epsilon), &\quad \tilde \ttt_{i_k} \in \tilde B_k(\epsilon),\\
\aaaa_{j_l}^* \in G_l^*(\epsilon), &\quad \tilde \aaaa_{j_l} \in \tilde G_l(\epsilon)
\end{aligned}
\end{equation*}
for $1\leq k \leq k_1$ and $1\leq l \leq k_2.$
For any $i \in H_1(\epsilon),$ we consider the following configurations
\begin{equation*}
\begin{aligned}
&\tau^* = [\ttt_i^*,\ttt_{i_1}^*,...,\ttt_{i_{k_1}}^*,\aaaa_{j_1}^*,...,\aaaa_{j_{k_2}}^*] \in \HH_{k_1+k_2+1,K,M},\\
&\tilde \tau = [\tilde \ttt_i,\tilde \ttt_{i_1},...,\tilde \ttt_{i_{k_1}},\tilde \aaaa_{j_1},...,\tilde \aaaa_{j_{k_2}}] \in \HH_{k_1+k_2+1,K_+,M}
\end{aligned}
\end{equation*}
and
\begin{equation*}
\begin{aligned}
\tau^*_1 &= [\ttt_i^*,\aaaa_{j_1}^*,...,\aaaa_{j_{k_2}}^*] \in [B^{K}_{\bf{0}}(M),G_{1}^*(\epsilon),...,G_{k_2}^*(\epsilon)], \\
\tilde\tau_1 &= [\tilde \ttt_i,\tilde \aaaa_{j_1},...,\tilde \aaaa_{j_{k_2}}] \in [B^{K_+}_{\bf{0}}(M),\tilde G_1(\epsilon),...,\tilde G_{k_2}(\epsilon)].
\end{aligned}
\end{equation*}
It is obvious that $$d(\tilde \tau, \tau^*) \geq d\left(\tilde\tau_1,\tau^*_1 \right).$$


By Lemma \ref{lem:part1+}, we have $$\sqrt{\sum_{l=1}^{k_2} \|F_{N,J}(\tilde\aaaa_{j_l}) - \aaaa_{j_l}^*\|^2 } \leq 5\sqrt{k_2} \epsilon \leq 5\sqrt{k_1+k_2}\epsilon.$$ Combining it with \eqref{eq:H1 definition} and Lemma \eqref{lem:add a point},
we have $$d(\tilde\tau_1, \tau^*_1) > 5\sqrt{k_1+k_2}\epsilon,$$ which leads to
\begin{equation}\label{eq:temp1}
d(\tilde\tau, \tau^*) > 5\sqrt{k_1+k_2}\epsilon.
\end{equation}


According to Lemma \ref{lem:another}, $\{\ttt_i^*,\ttt_{i_1}^*,...,\ttt_{i_{k_1}}^*\},\{\aaaa_{j_1}^*,...,\aaaa_{j_{k_2}}^*\}$ are a collection of anchor points.
Let $\tilde D, D \in \PP_{k_1+1,k_2,K_+,M}$ be the partial distance matrix of $\tilde \tau$ and $\tau^*,$ respectively. Combining \eqref{eq:temp1} and Lemma \ref{lem:anchor_continuous}, there exists a constant $\delta_\epsilon>0$ such that
\begin{equation}\label{eq:temp2}
\|\tilde D - D\|_F \geq \delta_\epsilon.
\end{equation}
For each $i \in H_1(\epsilon),$ we choose $i_1,...,i_{k_1} \in I_1(\epsilon)$ to form a group $\{i,i_1,...,i_{k_1}\} \subset [N]$ such that $$(\ttt_i^*,\ttt_{i_1}^*,...,\ttt_{i_{k_1}}^*) \in B^{K}_{\bf{0}}(M) \times B_{1}^*(\epsilon)\times \cdots \times B_{k_1}^*(\epsilon)$$ and $$(\tilde\ttt_i,\tilde\ttt_{i_1},...,\tilde\ttt_{i_{k_1}}) \in B^{K_+}_{\bf{0}}(M)\times \tilde B_{1}(\epsilon)\times \cdots \times \tilde B_{k_1}(\epsilon).$$ We could find at least $\min\{\lambda_{1,N,J}, p_{\epsilon}\}\times N$ such groups which are mutually exclusive. We could also find at least $p_\epsilon J$ mutually exclusive groups of $\{j_1,...,j_{k_2}\} \subset [J]$ such that $$(\aaaa_{j_1}^*,...,\aaaa_{j_{k_2}}^*) \in G_{1}^*(\epsilon)\times \cdots \times G_{k_2}^*(\epsilon)$$ and $$(\tilde\aaaa_{j_1},...,\tilde\aaaa_{j_{k_2}}) \in \tilde G_{1}(\epsilon)\times \cdots \times \tilde G_{k_2}(\epsilon).$$
By \eqref{eq:bound} and \eqref{eq:temp2}, we have $$\min\{\lambda_{1,N,J},p_\epsilon\} N p_\epsilon J  \delta^2_\epsilon \leq o(NJ).$$
So $$\min\{\lambda_{1,N,J}, p_\epsilon\} = o(1),$$
which means $\lambda_{1,N,J} \to 0,$ as $N,J$ grow to infinity.
Similar result holds for $\lambda_{2,N,J}$ and we do not repeat it.
\end{proof}

\begin{proof}[Proof of Lemma \ref{lem:kmeans}]
Consider the K-means clustering of the person points in Algorithm~\ref{algo:cluster}. We define a loss function
$$\mathcal L(\vartheta_1, ..., \vartheta_N) = \frac{1}{N} \sum_{i=1}^N \Vert \hat \ttt_i -  \boldsymbol \mu_{\vartheta_i}\Vert^2,$$
as the loss function for K-means clustering,  where $\vartheta_i \in \{1, ..., k_1\}$ represents the cluster membership of person $i$ and
$$\boldsymbol \mu_{k} = \frac{\sum_{i=1}^N \hat \ttt_i 1_{\{ \vartheta_i = k\}}}{\sum_{i=1}^N 1_{\{ \vartheta_i = k\}}}$$
denotes the centroid of the $k$th cluster. Under the conditions of Theorem~\ref{thm:cluster}, the K-means clustering converges to the global optima, which implies that
\begin{equation}\label{eq:cluloss}
\mathcal L(\hat \vartheta_1, ..., \hat\vartheta_N) = \min_{\vartheta_i \in \{1, ..., k_1\}, i = 1, ..., N} \mathcal L(\vartheta_1, ..., \vartheta_N).
\end{equation}
So for any isometry $F \in \mathcal A_K$, $$\sum_{i=1}^N\|\hat\ttt_i-\boldsymbol \mu_{\hat\vartheta_i}\|^2 \leq \sum_{i=1}^N\|\hat\ttt_i-F(\bb^*_{\vartheta_i^*})\|^2.$$
By triangle inequality,
\begin{align*}
\left(\sum_{i=1}^N\|\boldsymbol \mu_{\hat\vartheta_i} - F(\bb^*_{\vartheta_i^*})\|^2\right)^{\frac{1}{2} } &\leq \left(\sum_{i=1}^N( \|\boldsymbol \mu_{\hat\vartheta_i} - \hat\ttt_i \|^2 \right)^{\frac{1}{2}}
+ \left(\sum_{i=1}^{N}\|\hat\ttt_i-F(\bb^*_{\vartheta_i^*})\|^2\right)^{\frac{1}{2}}, \\
& \leq 2 \left(\sum_{i=1}^{N}\|\hat\ttt_i-F(\bb^*_{\vartheta_i^*})\|^2\right)^{\frac{1}{2}}.
\end{align*}
Define $d = \min_{i\neq j}\|\bb_i^*-\bb_j^*\|$ and for $F \in \mathcal A_K$, define $$A_F := \{1\leq i \leq N: \|\boldsymbol \mu_{\hat \vartheta_i} - F(\bb^*_{\vartheta_i^*})\| < \frac{d}{2}\},$$
and denote $A_F^{\mathsf{c}} := \{1,...,N\}/A_F.$

Then
\begin{align}\label{eq:set}
\frac{\sum_{i \in A_{F^*_{N,J}}}1}{N}
&= 1 - \frac{\sum_{i \in A_{F^*_{N,J}}^\mathsf{c}} 1}{N}  \nonumber \\
&\geq 1 - \frac{4}{d^2} \frac{\sum_{i \in A_{F^*_{N,J}}^\mathsf{c}} \|\boldsymbol \mu_{\hat \vartheta_i} - F^*_{N,J}(\bb^*_{\vartheta_i^*})\|^2}{N} \nonumber \\
&\geq 1 - \frac{4}{d^2} \frac{\sum_{i=1}^N \|\boldsymbol \mu_{\hat \vartheta_i} - F^*_{N,J}(\bb^*_{\vartheta_i^*})\|^2}{N}\\
&\geq 1 - \frac{16}{d^2} \frac{\sum_{i=1}^N \|\hat\ttt_i - F^*_{N,J}(\bb^*_{\vartheta_i^*})\|^2}{N} \nonumber \\
&\stackrel{pr}\to 1 \nonumber
\end{align}

\begin{lemma}\label{lem:perm}
Under the same conditions as Lemma \ref{lem:kmeans},
if there exists $\zeta_1 \in \mathcal B_{k_1}$ satisfying $$\| \boldsymbol \mu_{\zeta_1(l)} - F_{N,J}^*(\bb_l^*) \| < \frac{d}{2},$$
where $\boldsymbol \mu_l$ is the centroid of the $l$th cluster, $F_{N,J}^*$ is defined in \eqref{eq:person} and $d$ is defined above,
then there exists $\zeta_2 \in \mathcal B_{k_1}$, such that for all $i \in A_{F^*_{N,J}}, \hat \vartheta_i = \zeta_2(\vartheta^*_i)$.
\end{lemma}
Let $\Omega_{N,J} := \{\omega:\exists \zeta \in \mathcal B_{k_1}, \text{ s.t. } \| \boldsymbol \mu_{\zeta(l)}(\omega)  - F_{N,J}^*(\bb_l^*) \| < \frac{d}{2}, ~ i=1,...,k_1\}$. Notice that $\Omega_{N,J}$ is a subset of the whole probability space.
By Lemma \ref{lem:perm}, for any $\omega \in \Omega_{N,J}$, there exists $\zeta_{N,J} \in \mathcal B_{k_1}$, which corresponds to $\zeta_2$ in Lemma \ref{lem:perm}, such that
\begin{align*}
\max_{\zeta \in \mathcal B_{k_1}}\frac{\sum_{i=1}^{N} 1_{\{\vartheta_i^* = \zeta(\hat \vartheta_i(\omega))\}}}{N} \geq
\frac{\sum_{i=1}^N 1_{\{\hat \vartheta_i = \zeta_{N,J}(\vartheta_i^*)\}}}{N} \geq \frac{\sum_{i \in A_{F^*_{N,J}}}1}{N}
\end{align*}
\begin{lemma}\label{lem:temp}
Under the same conditions as Lemma \ref{lem:kmeans}, we have
$$\lim_{N,J \to \infty}\Pr\left( \Omega_{N,J} \right) = 1,$$
where $\Omega_{N,J}$ is defined above.
\end{lemma}
By Lemma \ref{lem:temp} and \eqref{eq:set}, we complete the proof.
\end{proof}

\begin{proof}[Proof of Lemma \ref{lem:anchor_exist_new}]
Without loss of generality, we suppose that the ball $G \subset \R^K$ has center at orgin.
By Theorem 3.3 of \cite{alfakih2003uniqueness}, we know there exist $k_1, k_2 \geq K+1$ and two sets of points, $\{\bb^*_1,...,\bb^*_{k_1}\}, \{\cc^*_1,...,\cc^*_{k_2}\} \subset \text{int}(G),$
satisfying condition A2 whose partial distance matrix $D^*$ has unique configuration. Furthermore, points near $\bb^*_i,\cc^*_j$ also have this property. Specifically, there exists $\epsilon >0$ such that for $$\bb_i \in B^K_{\bb_i^*}(\epsilon)  \subset G ,\quad \cc_j \in B^K_{\cc_j^*}(\epsilon) \subset  G,$$
the $\{\bb_1,...,\bb_{k_1}\}, \{\cc_1,...,\cc_{k_2}\}$ satisfy condition A2 and their partial distance matrix $D$ has unique configuration.
Then, by Lemma \ref{lem:anchor_continuous},
condition A1 holds and $\{\bb^*_1,...,\bb^*_{k_1}\}, \{\cc^*_1,...,\cc^*_{k_2}\}$ are anchor points in $\R^K.$
\end{proof}

\begin{proof}[Proof of Lemma \ref{lem:concentration}]
The proof of Lemma \ref{lem:concentration} is similar to Lemma A.1 of \cite{davenport20141}.
\end{proof}

\begin{proof}[Proof of Lemma~\ref{lem:bddist}]
We have
\begin{align*}
0 \leq \bar l_{\Omega,Y}(\hat D_{N,J}) - \bar l_{\Omega,Y}(D^*_{N,J}) =& \bar l_{\Omega,Y}(\hat D_{N,J}) - \mathbb E \bar l_{\Omega,Y}(\hat D_{N,J}) +  \mathbb E \bar l_{\Omega,Y}(\hat D_{N,J}) - \mathbb E \bar l_{\Omega,Y}(D^*_{N,J})\\
&+ \mathbb E \bar l_{\Omega,Y}(D^*_{N,J}) - \bar l_{\Omega,Y}(D^*_{N,J})\\
\leq& \left(\mathbb E \bar l_{\Omega,Y}(\hat D_{N,J}) - \mathbb E \bar l_{\Omega,Y}(D^*_{N,J})\right) + 2\sup_{D \in H}\vert \bar l_{\Omega,Y}(D) - \bar l_{\Omega,Y}(D)  \vert.
\end{align*}
So $$ \mathbb E \left(\bar l_{\Omega,Y}(D^*_{N,J}) - \bar l_{\Omega,Y}(\hat D_{N,J}) \right)  \leq 2\sup_{D \in H}\vert \bar l_{\Omega,Y}(D) - \bar l_{\Omega,Y}(D)  \vert.$$
Notice that
\begin{align*}
\mathbb E \left(\bar l_{\Omega,Y}(D^*_{N,J}) - \bar l_{\Omega,Y}(\hat D_{N,J}) \right) &= \mathbb E \left(l_{\Omega,Y}(D^*_{N,J}) - l_{\Omega,Y}(\hat D_{N,J}) \right)\\
&= \frac{n}{NJ}\sum_{i,j}f(d^*_{ij})\log(\frac{f(d^*_{ij})}{f(\hat d_{ij})}) + (1-f(d^*_{ij})) \log(\frac{1-f(d^*_{ij})}{1-f(\hat d_{ij})})
\end{align*}
For two distributions $\mathcal P$ and $\mathcal Q$, let $D_{KL}(\mathcal P \Vert \mathcal Q)$ denote the  Kullback-Leibler  divergence
$$ D_{KL}(\mathcal P\Vert \mathcal Q) := \int p(x)\log\left(\frac{p(x)}{q(x)}\right)dx,$$ where $p(x)$ and $q(x)$ are the density functions for $\mathcal P$ and $\mathcal Q$, respectively. For $0 < p,q < 1$, we use
$$D_{KL}(p\Vert q) := p\log(\frac{p}{q}) + (1-p)\log(\frac{1-p}{1-q})$$ to denote the Kullback-Leibler divergence between two Bernoulli distributions with parameter $p$ and $q$, respectively. For $P,Q \in (0,1)^{N\times J}$, we define $$D_{KL}(P\Vert Q) := \frac{1}{NJ}\sum_{i,j}D_{KL}(P_{ij}\Vert Q_{ij}).$$
For a partial distance matrix $D_{N,J}$, denote $f(D_{N,J})$ as the matrix $(f(d_{ij}))_{N\times J}$. So from above, we know that $$nD_{KL}(f(D^*_{N,J})\Vert f(\hat D_{N,J}))\leq 2\sup_{D \in H}\vert \bar l_{\Omega,Y}(D) - \bar l_{\Omega,Y}(D)  \vert.$$
Still for $0<p,q<1$, let $$d^2_H(p,q) := (\sqrt{p} - \sqrt{q})^2 + (\sqrt{1-p} - \sqrt{1-q})^2$$ denote the Hellinger distance between two Bernoulli distributions with parameters $p$ and $q$, respectively.
For $P,Q \in (0,1)^{N\times J}$, we define $$d^2_H(P\Vert Q) := \frac{1}{NJ}\sum_{i,j}d^2_H(P_{ij},Q_{ij}).$$
It is easy to check that $d^2_H(p,q) \leq D_{KL}(p\Vert q)$. So $$d^2_H(f(D^*_{N,J}),f(\hat D_{N,J})) \leq \frac{2}{n}\sup_{D \in H}\vert \bar l_{\Omega,Y}(D) - \bar l_{\Omega,Y}(D)  \vert$$
By Lemma A.2 of \cite{davenport20141}, we have
\begin{align*}
\frac{1}{NJ}\|\hat D_{N,J} - D^*_{N,J}\|_F^2  &\leq  8\beta_{4M^2}d^2_H(f(D^*_{N,J}),f(\hat D_{N,J}))\\
&\leq \frac{16}{n}\beta_{4M^2}\sup_{D \in H}\vert \bar l_{\Omega,Y}(D) - \bar l_{\Omega,Y}(D)  \vert.
\end{align*}

\end{proof}

\begin{proof}[Proof of Lemma \ref{lem:add a point}]
Denote $$
\eta := \max\left\{ d(\tau_1,\tau_2), \sqrt{\sum_{l=1}^{k_2} \| \x_l^+-\y_l \|^2} \right\}$$ and then
$d(\tau_1,\tau_2) \leq \eta$ and
\begin{equation}\label{eq:cond1}
\|\x^+_l-\y_l\| \leq \eta, \quad l = 1,...,k_2.
\end{equation}
Therefore there exist $A\in \OO_{K_+}$ and $\bb \in \R^{K_+}$ such that
\begin{equation*}
\sqrt{\|A\x^+ + \bb-\y\|^2+\sum_{l=1}^{k_2}\|A\x^+_l+\bb-\y_l\|^2} \leq \eta,
\end{equation*}
which leads that \begin{equation}\label{eq:cond21}
\| A\x^++\bb-\y\| \leq \eta
\end{equation} and \begin{equation}\label{eq:cond22}
\|A\x^+_l+\bb-\y_l\| \leq \eta,\quad l=1,...,k_2.
\end{equation} Combining \eqref{eq:cond1} and \eqref{eq:cond22}, we get
\begin{equation*}
\|A\x^+_l+\bb-\x^+_l\| \leq 2\eta, \quad l = 1,...,k_2.
\end{equation*}
According to condition A2, $\x_1,...,\x_{k_2}$ can affine span $\R^K$.
Then there exists $\alpha_1, ..., \alpha_{k_1}$ satisfying $\sum_{l=1}^{k_2}\alpha_l = 1,$ such that
$\x = \sum_{l=1}^{k_2} \alpha_l\x_l.$ So we have $$
\begin{aligned}
\|A\x^++\bb-\x^+\|  = \left\|\sum_{l=1}^{k_2}\alpha_l(A\x^+_l+\bb-\x^+_l)\right\| \leq 2(\sum\limits_{j=1}^{k_2}|\alpha_j|) \eta.
\end{aligned}
$$ Combining it with \eqref{eq:cond21}, we have
$\|\x^+-\y\| \leq (2\sum\limits_{j=1}^{k_1}|\alpha_j|+1)\eta.$ We complete the proof by setting the constant $c$ in Lemma \ref{lem:add a point} to be $$\max_{\substack{\x\in B_{\0}^K(M) \\ \x_l \in G_{l}^*(\epsilon)}} \inf_{ \substack{ \sum_{l}\alpha_l = 1\\ \x = \sum_{l}\alpha_l\x_l } }2\sum_{l=1}^{k_2}|\alpha_l|.$$
\end{proof}

\begin{proof}[Proof of Lemma \ref{lem:perm}]
For $i,j \in A_{F^*_{N,J}}$, if $\vartheta_i^* = \vartheta_j^*$ and suppose they are both equal to $k$, then $\|\boldsymbol \mu_{\hat \vartheta_i} - F^*_{N,J}(\bb^*_{k})\| < d/2$ and $\|\boldsymbol \mu_{\hat \vartheta_j} - F^*_{N,J}(\bb^*_{k})\| < d/2$. Given the condition in Lemma \ref{lem:perm}, it is easy to check that there is only one $\boldsymbol \mu_l$ among $\{\boldsymbol \mu_1,...,\boldsymbol \mu_{k_1}\}$ satisfying
$$\|\boldsymbol \mu_{l} - F^*_{N,J}(\bb^*_{k})\| < d/2,$$
then $\hat \vartheta_i = \hat \vartheta_j$.
If $\vartheta_i^* \neq \vartheta_j^*$, then $\|\boldsymbol \mu_{\hat \vartheta_i} - F^*_{N,J}(\bb^*_{\vartheta_i^*})\| < d/2$ and $\|\boldsymbol \mu_{\hat \vartheta_j} - F^*_{N,J}(\bb^*_{\hat \vartheta_j})\| < d/2$. So
\begin{align*}
\| \boldsymbol \mu_{\hat \vartheta_i} -  \boldsymbol \mu_{\hat \vartheta_j} \| &\geq \| F^*_{N,J}(\bb^*_{\vartheta_i^*}) - F^*_{N,J}(\bb^*_{\vartheta_j^*}) \| - \|\boldsymbol \mu_{\hat \vartheta_i} - F^*_{N,J}(\bb^*_{\vartheta_i^*})\| -
\|\boldsymbol \mu_{\hat \vartheta_j} - F^*_{N,J}(\bb^*_{\hat \vartheta_j})\| \\
& > d - \frac{d}{2} - \frac{d}{2} > 0,
\end{align*}
which means $\hat \vartheta_i \neq \hat \vartheta_j$.
So there exists $\zeta_2$ such that for $i \in A_{F^*_{N,J}}, \hat \vartheta_i = \zeta_2(\vartheta^*_i)$.
\end{proof}

\begin{proof}[Proof of Lemma \ref{lem:temp}]
Let $$\Gamma_{N,J}^{(\epsilon')} := \{  \omega: \frac{\sum_{i \in A_{F^*_{N,J}}}1}{N} \geq 1-\epsilon'     \},$$
which is a subset of the whole probability space.
By \eqref{eq:set}, for any $\epsilon' > 0$, we have $$\lim_{N,J \to \infty}\Pr(\Gamma_{N,J}^{(\epsilon')}) = 1.$$
For any $\omega \notin \Omega_{N,J}$, there exists $l$ such that $\|\boldsymbol \mu_m(\omega) - F_{N,J}^*(\bb_l^*)\| \geq d/2,$ for $m = 1,...,k_1.$ So for $i$ satisfying $\vartheta_i^* = l$, we have $\|\boldsymbol \mu_{\hat \vartheta_i}(\omega) - F_{N,J}^*(\bb_{\vartheta_i^*}^*)\| \geq d/2$. According to (2) of condition A3, for sufficiently small $\epsilon'$, if $N,J$ are sufficiently large, then $\omega \notin \Gamma_{N,J}^{(\epsilon')},$ which means $\Gamma_{N,J}^{(\epsilon')} \subset \Omega_{N,J}$. By \eqref{eq:set},
for sufficiently small $\epsilon'$,
$$
\lim_{N,J \to \infty}\Pr(\Omega_{N,J}) \geq \lim_{N,J \to \infty}\Pr(\Gamma_{N,J}^{(\epsilon')}) = 1.
$$
We complete the proof.
\end{proof}

\section{Algorithm-based MDU: Real Data Examples}

To compare the proposed method with classical algorithm-based MDU methods, we apply ordinal MDU \citep{busing2005avoiding} to both real datasets analyzed in the paper.
The application is based on the implementation
in R package \emph{smacof} \citep{smacof}. For both  examples,
the latent dimension is set to two, and all the tuning parameters are set to be the default ones.
The results below show that the ordinal MDU approach provides similar visualization results as the proposed one, especially for the roll call voting data due to its unidimensional nature.
The results for the movie rating dataset are also similar for the two methods, but the interpretable patterns from the ordinal MDU approach is not as clear as the proposed one. 



Figures~\ref{fig:classical_movie} through \ref{fig:classical_user_xy} show the same plots as in Figures \ref{fig:whole} through \ref{fig:user_xy} in Section~\ref{subsec:real1}, respectively, for the movie rating dataset. Figure~\ref{fig:classical_movie} provides
the simultaneous visualization of the movie and user points. Similar to the plot in Figure~\ref{fig:whole} given by our method, the movies and the users tend to form two giant clusters that only slightly overlap.

\begin{figure}
  \centering
        \includegraphics[scale = 0.6]{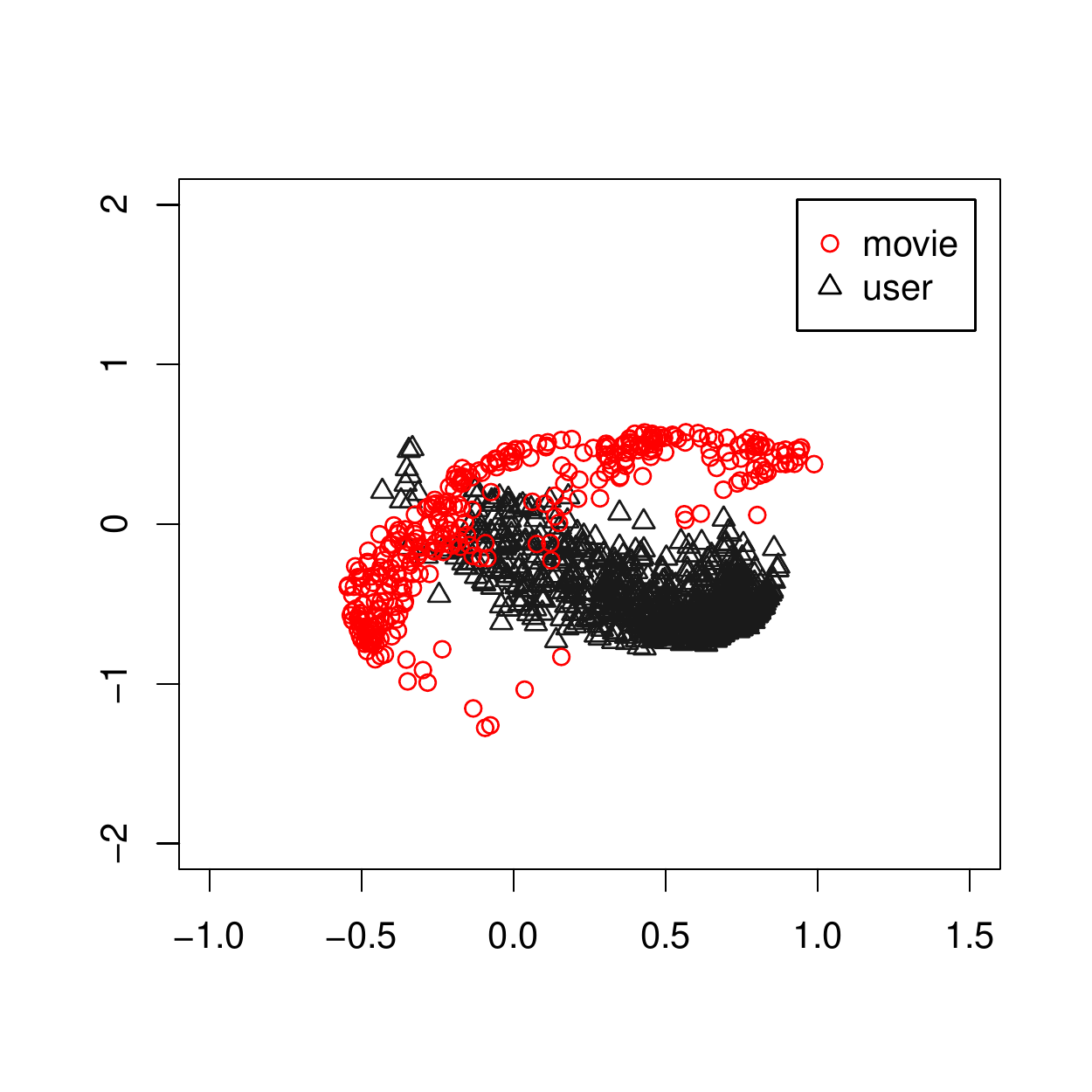}
  \caption{Analysis of movie rating data: Simultaneous visualization of the estimated movie and user points.}
 \label{fig:classical_movie}
\end{figure}

Figure~\ref{fig:classical_movie_xy}
is similar to Figure~\ref{fig:movie_xy}, where the two panels show the same scatter plot for the movie points. In the left panel, the movies are stratified by the the numbers of ratings that they received, where different stratums are marked by different colors. In the right panel, the movies are stratified by their release time. Recall that the patterns of popularity and release time are captured by the proposed method as shown in Figure~\ref{fig:movie_xy}. Figure~\ref{fig:classical_movie_xy} seems also to capture these patterns, but not as clear as those in Figure~\ref{fig:movie_xy}.
According to panel (a) of Figure~\ref{fig:classical_movie_xy}, the more popular movies tend to be located near the origin, while the less popular movies tend to be located away from the origin. According to panel (b) of Figure~\ref{fig:classical_movie_xy}, the clustering patten of the movies can be largely explained by the three categories of release dates. From the left to the right of the space, the points correspond to movies from the relatively older ones to the relatively more recent ones.

\begin{figure}
  \centering
    \begin{subfigure}[b]{0.4\textwidth}
        \centering
        \includegraphics[scale = 0.5]{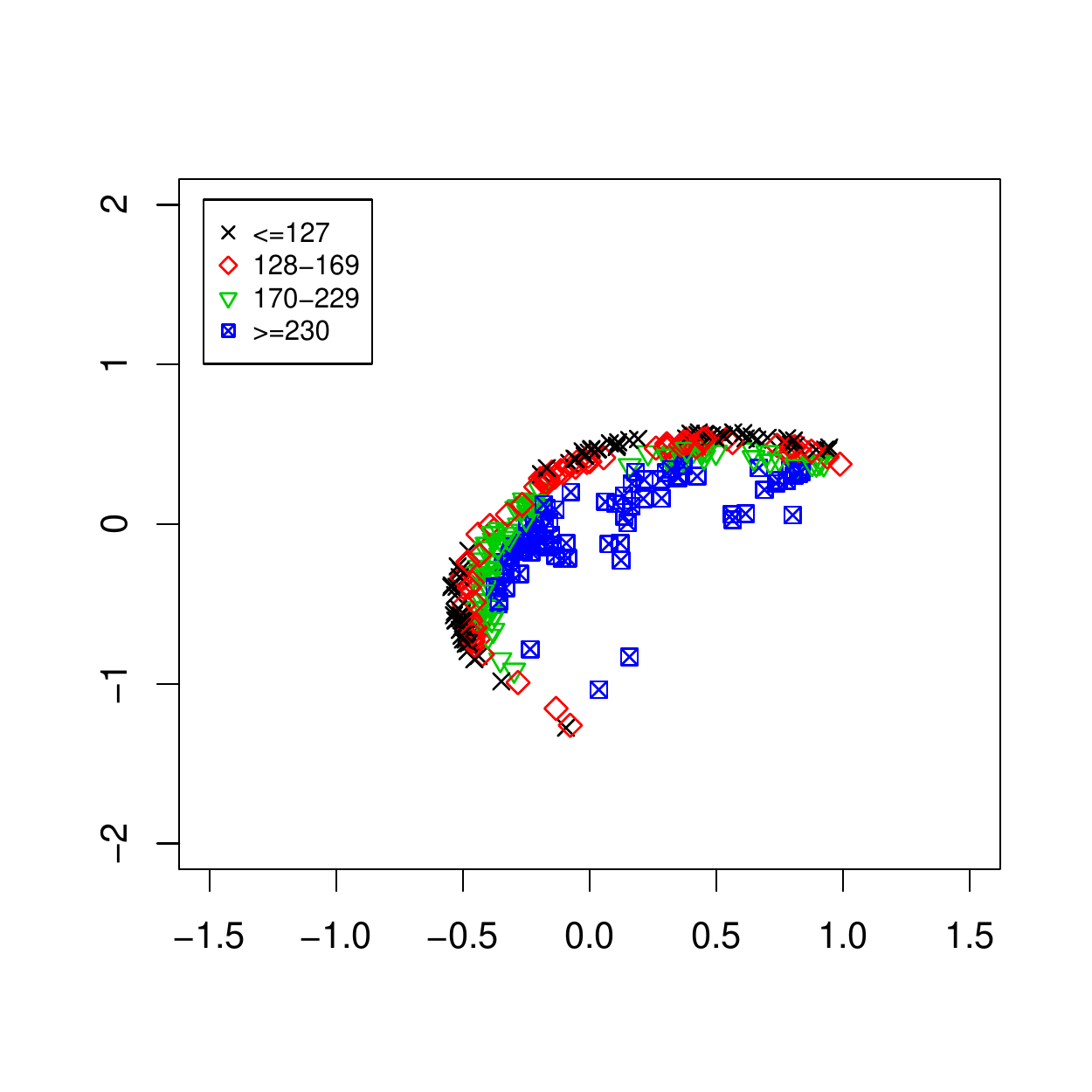}
        \caption{}
    \end{subfigure}
    \begin{subfigure}[b]{0.4\textwidth}
        \centering
        \includegraphics[scale = 0.5]{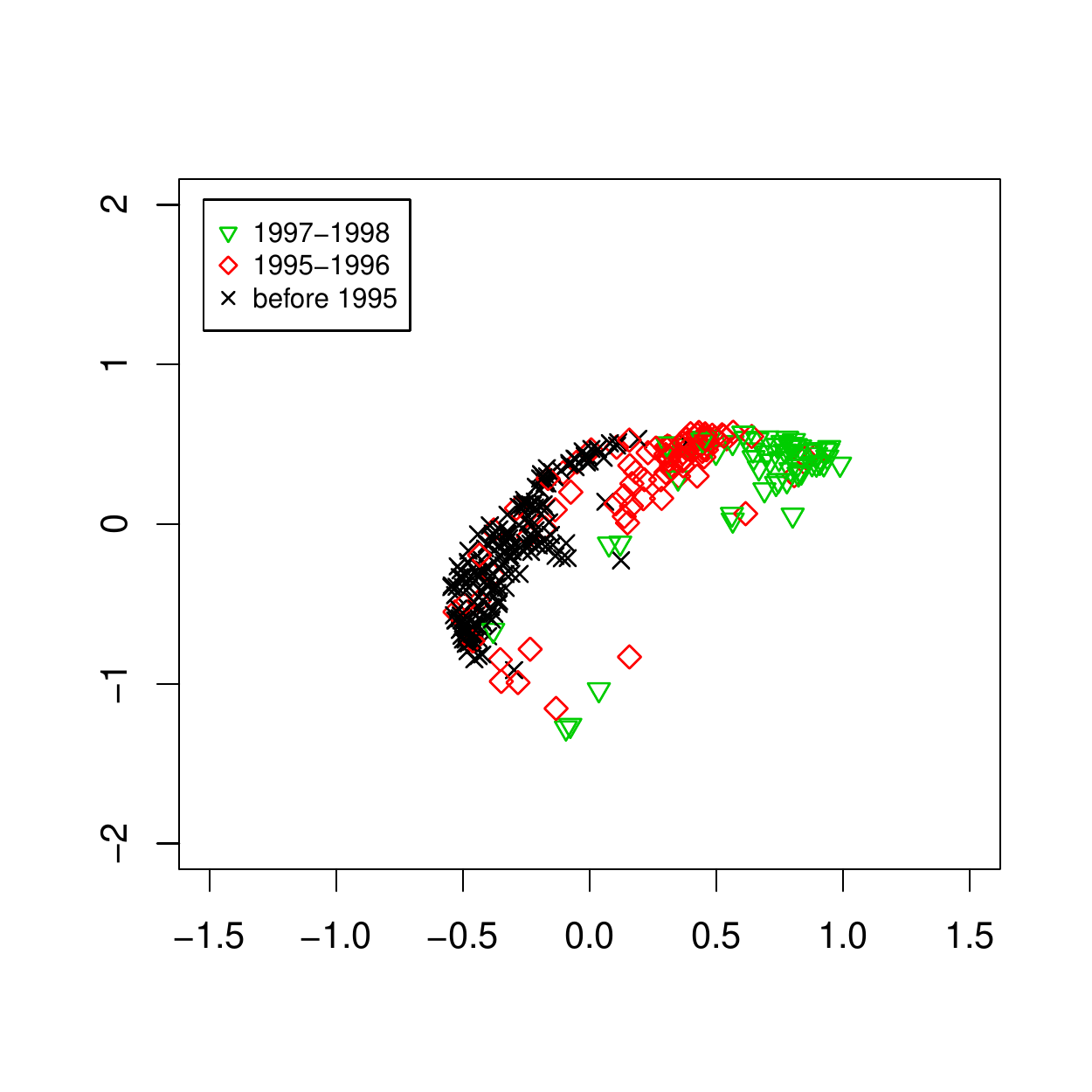}
        \caption{}
    \end{subfigure}
  \caption{Analysis of movie rating data. Panel (a): Visualization of movie points, with movies stratified into four equal-size categories based on the numbers of rating. Movies with numbers of rating less than 127, 128-169, 170-229 and more than 230 are indicated by black, red, green and blue points, respectively. Panel (b): Visualization of movie points, with movies stratified into three categories based on their release time. Movies released in 1997-1998, 1995-1996, and before 1995 are indicated by green, red and black points, respectively.}
 \label{fig:classical_movie_xy}
\end{figure}

Figure~\ref{fig:classical_user_xy} shows the same plots as in Figure~\ref{fig:user_xy}. Similar pattern is shown that the shorter the average distance from a user point to the movies points, the more active the user is. In Figure~\ref{fig:classical_user_xy}, users are classified into four equal-size groups depending on the numbers of movies they rated. These groups of users, from the most active one to the least active one, lie from the top left to the bottom right.


\begin{figure}
 \centering
        \includegraphics[scale = 0.6]{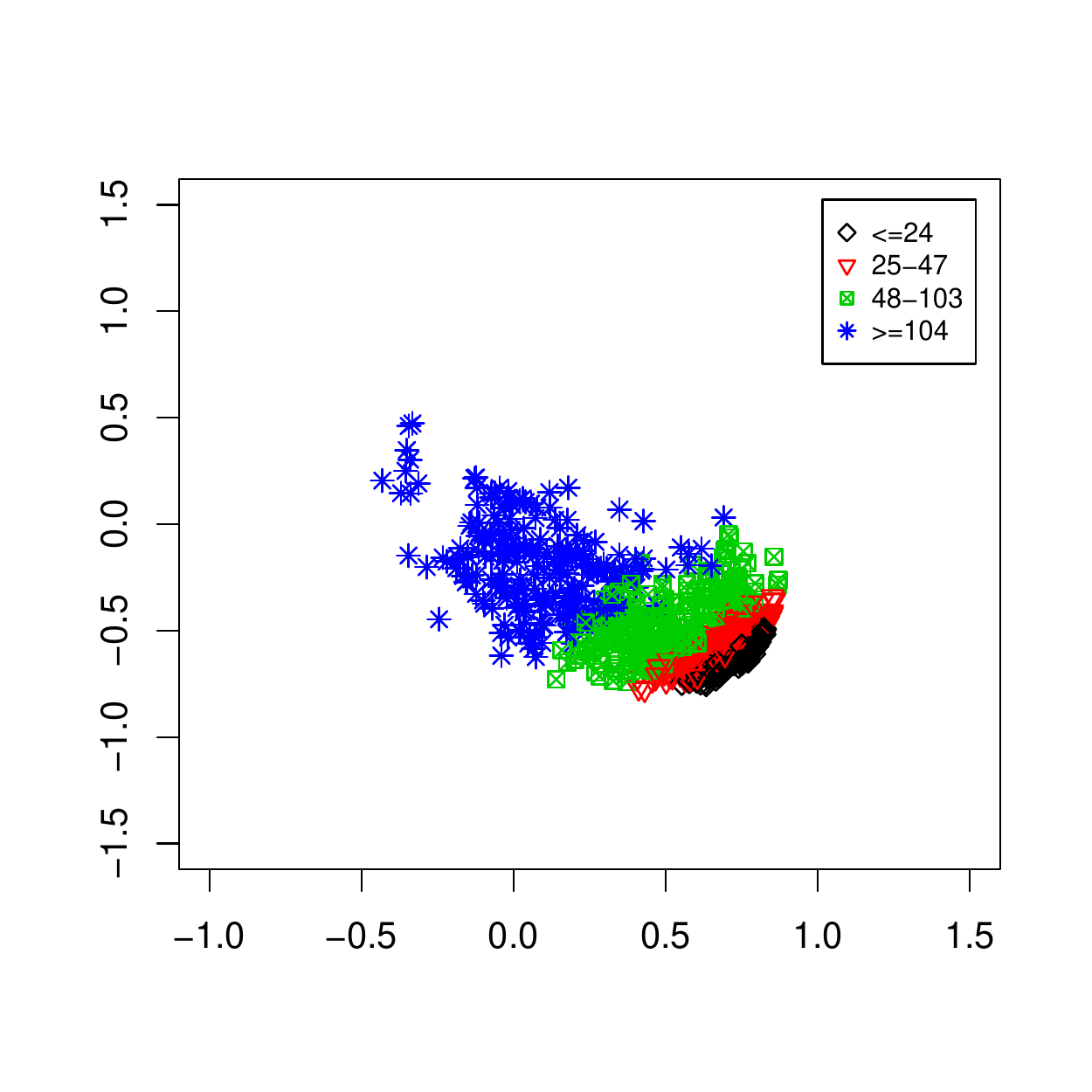}
        \caption{Analysis of movie rating data: Visualization of user points, with users classified into four equal-size categories based on the numbers of rating. Users who rated less than 24, , 25-47, 48-103 and more than 104 movies are indicated by
        black, red, green and blue points, respectively.}
 \label{fig:classical_user_xy}
\end{figure}

Figures \ref{fig:classical_voting} through \ref{fig:classical_voting_bill}  show the same plots as in Figures  \ref{fig:vote108_whole} through \ref{fig:vote108_bill}  in Section \ref{subsec:real2}, respectively, for the roll call voting dataset. Figure  \ref{fig:classical_voting} provides the simultaneous visualization of senators and roll calls. Similar to the plot in Figure~\ref{fig:vote108_whole}, most of the points tend to lie on a straight line.

\begin{figure}
  \centering
        \includegraphics[scale = 0.6]{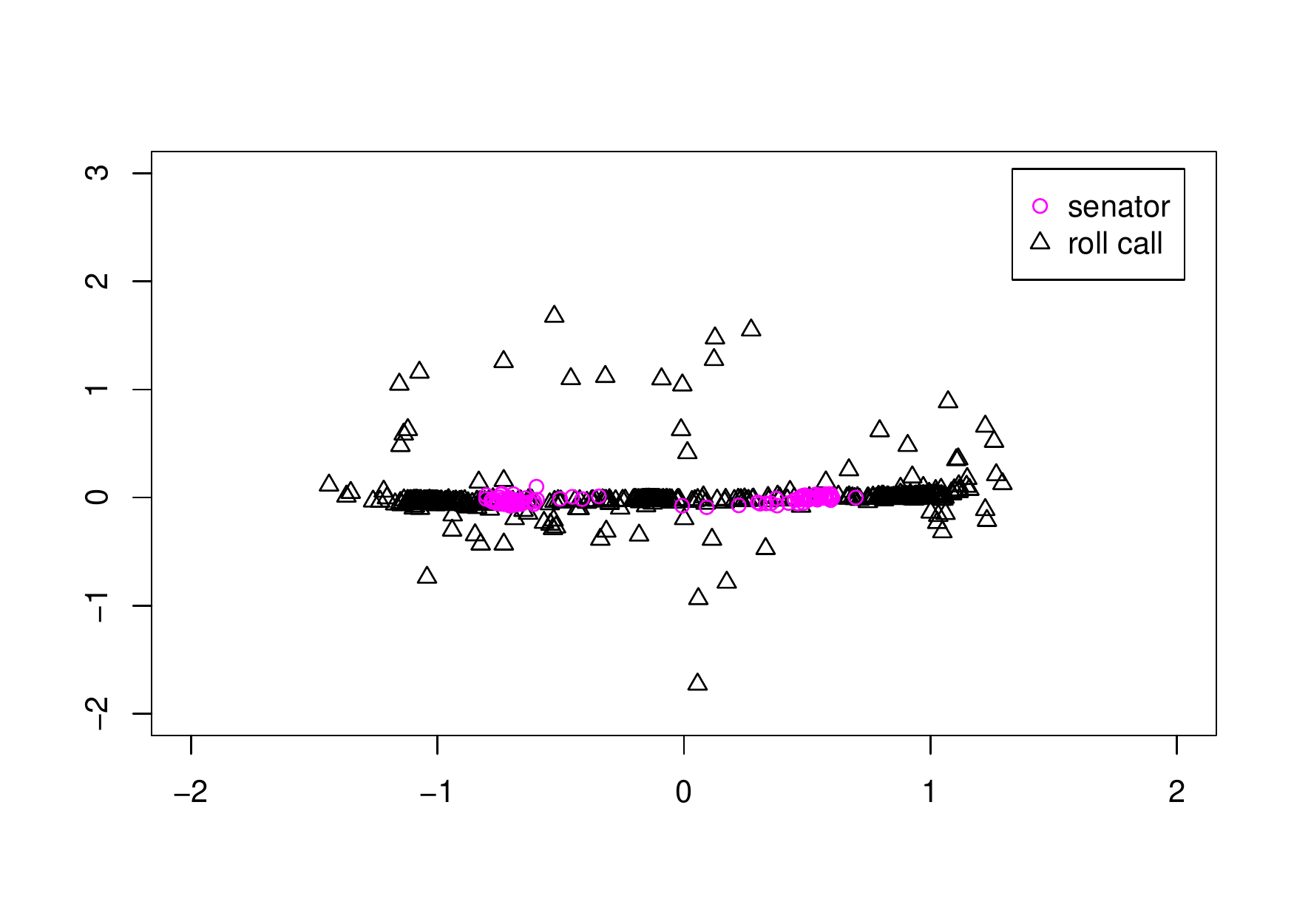}
  \caption{Analysis of senator roll call data: Simultaneous visualization of the estimated senator and roll call ideal points.}
 \label{fig:classical_voting}
\end{figure}

Figure~\ref{fig:classical_voting_senator} provides a scatter plot of the senator points. Similar to Figure~\ref{fig:vote108_senate}, most of the senator points tend to locate around a straight line, with the Democrats on one side and the Republicans on the other side. Also similar to  Figure~\ref{fig:vote108_senate},  the independent senator,  Jim Jeffords from the state of
Vermont, is mixed together with the Democrats, while
the Democrat senator, Zell Miller from the state of Georgia, is mixed together with the Republicans.

\begin{figure}
\centering
\includegraphics[scale = 0.6]{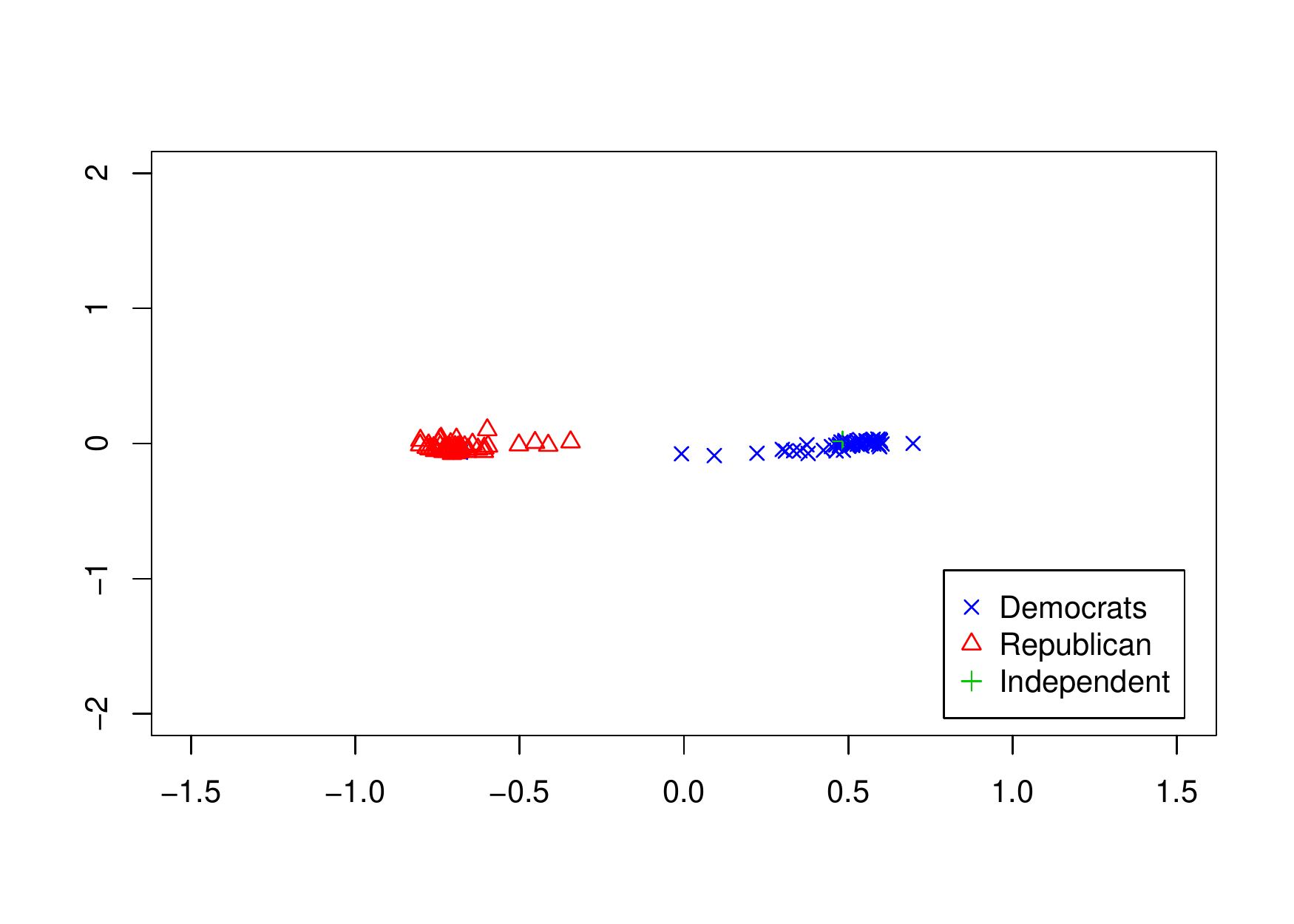}
\caption{Analysis of senator roll call data: Visualization of senator points, where senators are classified by their party membership. Specifically,
The Democrats, Republicans and an independent politician are indicated by blue, red, and green, respectively. }
\label{fig:classical_voting_senator}
\end{figure}


Finally, Figure~\ref{fig:classical_voting_bill} shows the unfolding results for the roll calls. The pattern in panel (a) of Figure~\ref{fig:classical_voting_bill} is similar to that of Figure~\ref{fig:vote108_bill}, where from the right to the left, the proportion of ``Yeas" from the Republicans increases.
Also similar to Figure~\ref{fig:vote108_bill}, although most of the roll calls lie near the $x$-axis, there are still quite a few of them spreading out along the $y$-axis.
According to panel (b) of Figure~\ref{fig:classical_voting_bill} based on the cross entropy measure,
the voting behavior on these roll calls tends to be heterogeneous within both parties.
This result is similar to that given in panel (b) of Figure~\ref{fig:vote108_bill}.


%
%

\begin{figure}
  \centering
    \begin{subfigure}[b]{0.4\textwidth}
        \centering
        \includegraphics[scale = 0.5]{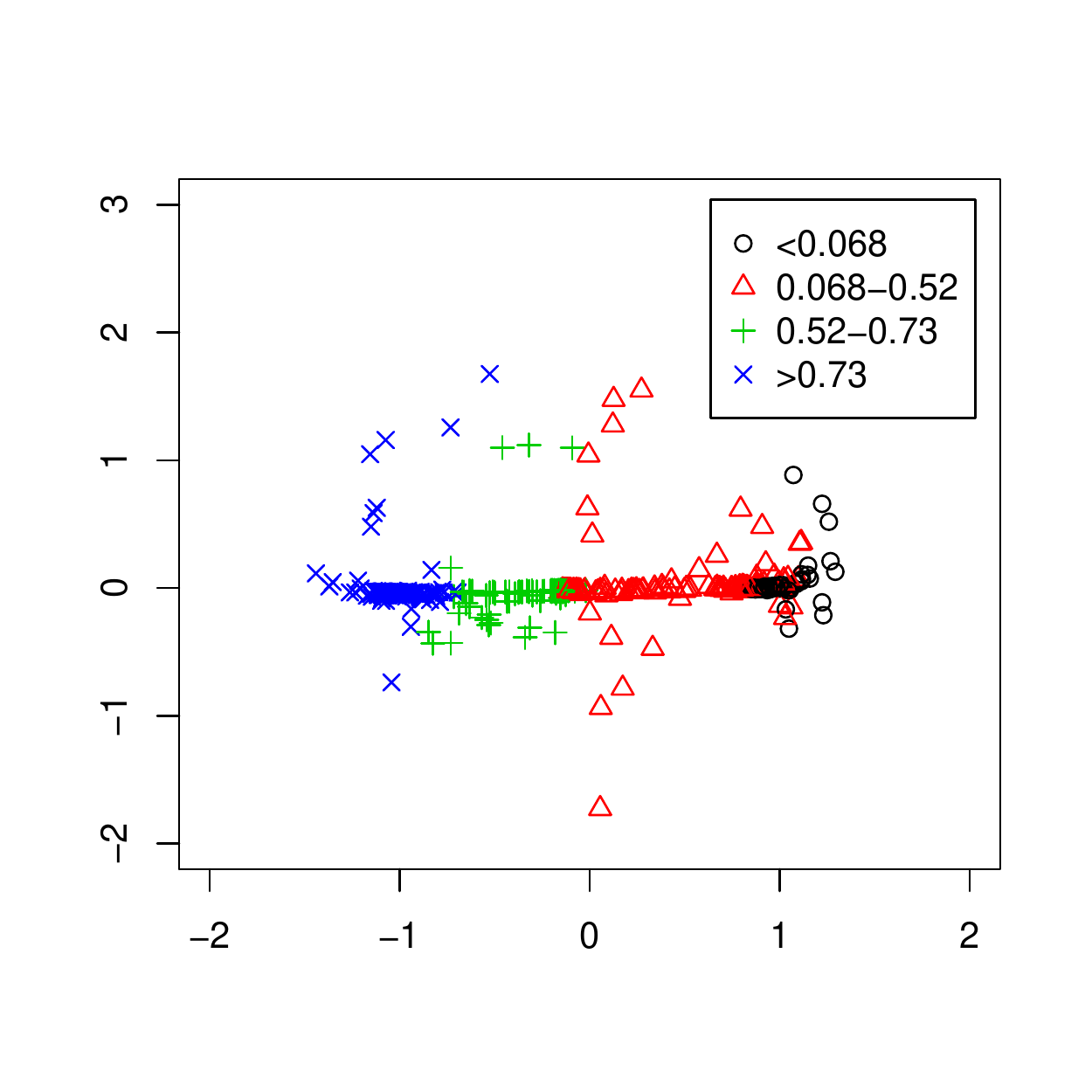}
        \caption{}
    \end{subfigure}
    \begin{subfigure}[b]{0.4\textwidth}
        \centering
        \includegraphics[scale = 0.5]{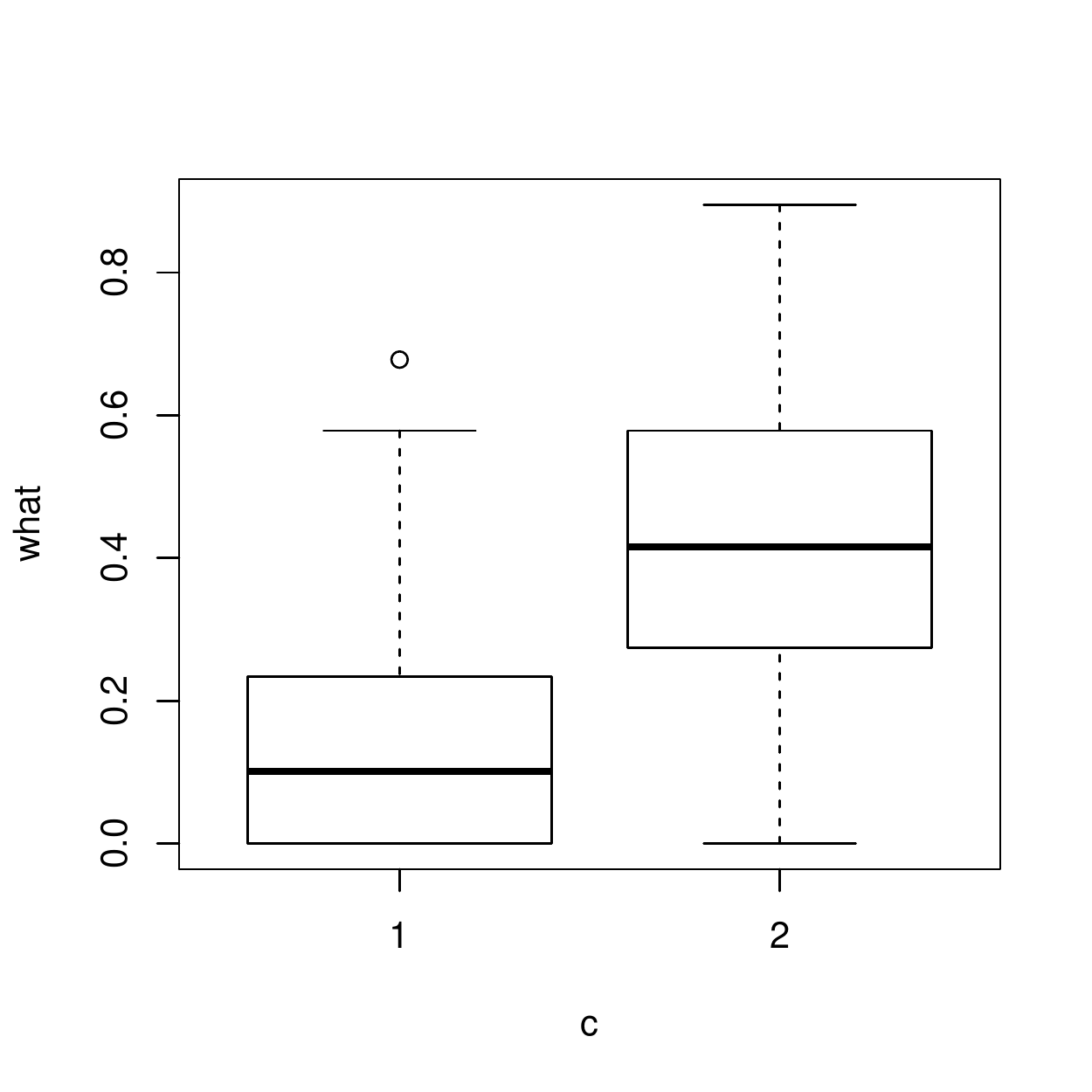}
        \caption{}
    \end{subfigure}
\caption{Analysis of senator roll call data. Panel (a): Visualization of roll call points, where roll calls are classified by the proportion of Yeas from Republicans. Specifically, roll calls who have the proportions less than 0.068, 0.068-0.52,0.52-0.73 and larger than 0.73 are indicated by black, red, green and blue points, respectively. 
Panel (b): Box plots of $\min\{\mbox{CE}^{(1)}_j, \mbox{CE}^{(2)}_j\}$, for roll calls lying near the $x$-axis ($|\hat a_{j2} | \leq 0.05$) one the left and for those spreading out along the $y$-axis ($|\hat a_{j2} | > 0.05$) on the right.}
\label{fig:classical_voting_bill}
\end{figure}

\bibliographystyle{apalike}
\bibliography{ref}
\end{document}